\documentclass[11pt]{article}
\usepackage[margin=1in]{geometry}
\usepackage{booktabs} 
\usepackage{times}

\usepackage[backref,colorlinks,citecolor=blue,bookmarks=true]{hyperref}
\usepackage{mathtools, amssymb, amsthm, bbm}
\usepackage[ruled,vlined]{algorithm2e}
\usepackage{svg}
 \usepackage{tikz}
\SetAlCapSkip{1em}
\usepackage{algorithmic}
\usepackage[title]{appendix}

\title{Testable Learning with Distribution Shift}
\author{
     Adam R. Klivans\thanks{\texttt{klivans@cs.utexas.edu}. Supported by NSF award AF-1909204 and the NSF AI Institute for Foundations of Machine Learning (IFML).} \\
	 UT Austin
	 \and Konstantinos Stavropoulos\thanks{\texttt{kstavrop@cs.utexas.edu}. Supported by NSF award AF-1909204, the NSF AI Institute for Foundations of Machine Learning (IFML) and by scholarships from Bodossaki Foundation and Leventis Foundation.} \\
	 UT Austin
     \and
    Arsen Vasilyan\thanks{\texttt{vasilyan@mit.edu}. Supported in part by NSF awards CCF-2006664, DMS-2022448, CCF-1565235, CCF-1955217,\\ CCF-2310818, Big George Fellowship and Fintech@CSAIL. Work done in part while visiting UT Austin.} \\
	 MIT
}
\date{\today}

\usepackage{mathtools, amssymb, amsthm, bbm}
\usepackage[capitalize]{cleveref}
\usepackage{enumitem}

\theoremstyle{plain}
\newtheorem{theorem}{Theorem}[section]
\newtheorem{lemma}[theorem]{Lemma}
\newtheorem{corollary}[theorem]{Corollary}
\newtheorem{proposition}[theorem]{Proposition}
\newtheorem{fact}[theorem]{Fact}

\newtheorem*{claim}{Claim}

\theoremstyle{definition}
\newtheorem{definition}[theorem]{Definition}

\theoremstyle{remark}
\newtheorem{remark}[theorem]{Remark}

\numberwithin{equation}{section}

\def\A{\mathcal{A}}
\def\B{\mathcal{B}}
\def\C{\mathcal{C}}
\def\D{\mathcal{D}}

\def\L{\mathcal{L}}
\def\X{\mathcal{X}}
\def\Y{\mathcal{Y}}

\def\S{\mathbb{S}}

\newcommand*{\N}{{\mathbb{N}}}

\newcommand*{\R}{{\mathbb{R}}}

\let\eps\epsilon
\let\phi\varphi

\DeclareMathOperator*{\pr}{\mathbb{P}}
\DeclareMathOperator*{\ex}{\mathbb{E}}
\DeclareMathOperator*{\E}{\mathbb{E}}

\DeclareMathOperator*{\argmax}{arg\,max}
\DeclareMathOperator*{\argmin}{arg\,min}

\DeclarePairedDelimiter{\norm}{\|}{\|}

\let\hat\widehat
\DeclareMathOperator{\poly}{poly}

\DeclareMathOperator{\sign}{sign}

\DeclareMathOperator{\ind}{\mathbbm{1}}

\newcommand{\cube}[1]{\{\pm 1\}^{#1}}

\newcommand{\ignore}[1]{} 

\newcommand*{\w}{\mathbf{w}}
\newcommand*{\vv}{\mathbf{v}}

\newcommand*{\x}{\mathbf{x}}
\newcommand*{\z}{\mathbf{z}}

\newcommand*{\Dgeneric}{\D}

\newcommand*{\Dtrainjoint}{\D_{\X\Y}^{\mathrm{train}}}
\newcommand*{\Dtestjoint}{\D_{\X\Y}^{\mathrm{test}}}
\newcommand*{\Dtrainmarginal}{\D_{\X}^{\mathrm{train}}}
\newcommand*{\Dtestmarginal}{\D_{\X}^{\mathrm{test}}}

\newcommand*{\Dgenericjoint}{\D_{\X\Y}}
\newcommand*{\Dgenericmarginal}{\D_{\X}}

\newcommand*{\error}{\mathrm{err}}

\newcommand{\train}{\mathrm{train}}
\newcommand{\test}{\mathrm{test}}
\newcommand{\optcommon}{\lambda}
\newcommand{\optcommontrain}{\optcommon_{\train}}
\newcommand{\optcommontest}{\optcommon_{\test}}

\newcommand{\concept}{f}

\newcommand{\coptcommon}{\concept^*}

\newcommand{\Slabelled}{S}
\newcommand{\Sunlabelled}{X}
\newcommand{\Strain}{\Slabelled_\train}
\newcommand{\Stest}{\Sunlabelled_\test}

\newcommand{\pup}{p_{\mathrm{up}}}
\newcommand{\pdown}{p_{\mathrm{down}}}

\newcommand{\Gauss}{\mathcal{N}}
\newcommand{\Unif}{\mathrm{Unif}}

\newcommand{\pbound}{B}
\newcommand{\mslack}{\Delta}
\newcommand{\degbound}{k}

\newcommand{\mindex}{\alpha}
\newcommand{\moment}{\mathrm{M}}
\newcommand{\momentempirical}{\hat{\moment}}

\newcommand{\mtrain}{m_{\train}}
\newcommand{\mtest}{m_{\test}}

\newcommand{\mconc}{m_{\mathrm{conc}}}
\newcommand{\mgen}{m_{\mathrm{gen}}}

\newcommand{\Polys}{\mathcal{P}}

\newcommand{\event}{\mathcal{E}}

\newcommand{\nhalfspaces}{\ell}
\newcommand{\dtsize}{s}
\newcommand{\aczsize}{s}
\newcommand{\aczdepth}{\ell}

\newcommand{\disc}{\mathrm{disc}}

\newcommand{\rademacher}{\mathfrak{R}}
\newcommand{\rademacherempirical}{\hat{\mathfrak{R}}}

\newcommand{\T}{\mathcal{T}}
\newcommand{\e}{\mathbf{e}}
\newcommand{\paramneighborhood}{\mathcal{V}}

\newcommand{\Sphere}{\mathbb{S}}
\newcommand{\vc}{\mathrm{VC}}

\newcommand{\disagreementregion}{\mathbf{D}}
\newcommand{\disagreementcoef}{\theta}

\newcommand{\Xregion}{\mathbf{X}}

\newcommand{\majority}{\mathrm{maj}}

\newcommand{\conv}{\mathrm{conv}}

\newcommand{\vopt}{\vv^*}

\newcommand{\badregion}{\mathbf{Z}}
\newcommand{\goodevent}{\mathcal{G}}

\begin{document}
\maketitle

\begin{abstract}%
  We revisit the fundamental problem of learning with distribution
shift, in which a learner is given labeled samples from training
distribution $\Dgeneric$, {\em unlabeled} samples from test distribution $\Dgeneric'$
and is asked to output a classifier with low test error.  The standard
approach in this setting is to bound the loss of a classifier in terms
of some notion of distance between $\Dgeneric$ and $\Dgeneric'$.  These distances, however, seem difficult to compute and do not lead to efficient algorithms.

We depart from this paradigm and define a new model called {\em testable learning with distribution shift}, where we {\em can} obtain provably efficient
algorithms for certifying the performance of a classifier on a test
distribution.  In this model, a learner outputs a classifier with low
test error whenever samples from $\Dgeneric$ and $\Dgeneric'$ pass an associated test;
moreover, the test must accept (with high probability) if the marginal of $\Dgeneric$ equals the marginal of $\Dgeneric'$.
We give several positive results for learning well-studied concept
classes such as halfspaces, intersections of halfspaces, and decision
trees when the marginal of $\Dgeneric$ is Gaussian or uniform on $\cube{d}$.  Prior to our
work, no efficient algorithms for these basic cases were known without strong assumptions on $\Dgeneric'$.

For halfspaces in the realizable case (where there exists a halfspace consistent with both $\Dgeneric$ and $\Dgeneric'$), we combine a moment-matching approach with ideas from active learning to simulate an efficient oracle for
 estimating disagreement regions. To extend to the non-realizable setting, we apply recent work from testable (agnostic) learning.  More generally, we prove that any
 function class with low-degree $\L_2$-sandwiching polynomial
 approximators can be learned in our model.  Since we require $\L_2$- sandwiching (instead of the usual $\L_1$ loss), we cannot
 directly appeal to convex duality and instead apply constructions
 from the pseudorandomness literature to obtain the required
 approximators.  We also provide lower bounds to show that the guarantees we obtain on the performance of our output hypotheses are best possible up to constant factors, as well as a separation showing that realizable learning in our model is incomparable to (ordinary) agnostic learning.
\end{abstract}
\thispagestyle{empty}
\newpage
\setcounter{page}{1}

\section{Introduction}
Mitigating distribution shift remains one of the major challenges of machine learning.  Training distributions can deviate significantly from test distributions, and pre-trained models are commonly deployed without a precise
understanding of these differences.  In such cases, a model may have poor performance with potentially dangerous consequences.  For example, several recent studies in the AI/healthcare community highlight the lack of generalization among many AI models trained to detect disease (e.g., skin cancer or pneumonia), often due to distribution shift.  As such, developing best practices for using these models in a clinical setting remains a vexing and
difficult problem \cite{zech2018variable, wong2021external, ternov2022generalizability}.

The computational landscape of traditional supervised learning--- where
training sets and tests are drawn from the same distribution--- is by
now well understood.  There is a rich literature of efficient
algorithms and computational hardness results for broad sets of
concept classes and distributions.  In contrast, little is known in
terms of efficient algorithms for classification in the context of
distribution shift or domain adaptation.  The most common approach is
to prove a generalization bound in terms of some notion of distance
between $\D$ and $\D'$ \cite{ben2006analysis, ben2010theory,mansour2009domadapt}.  These distances,
however, involve an enumeration of all functions in the underlying
concept class and seem difficult to compute.  Other recent work
requires oracles for empirical risk minimization \cite{goldwasser2020beyond,kalai2021efficient} or the existence of distribution-free reliable 
learners, which are believed to require superpolynomial time for even simple concept classes (e.g., reliably learning conjunctions is known to be harder than PAC learning DNF formulas) \cite[Section 4.2]{kalai2021efficient}.

In this work we define a new model called {\em testable learning with
distribution shift} (TDS learning) and show that this model does admit
efficient algorithms for several well-studied concept classes and
distributions.  Inspired by recent work in testable learning \cite{rubinfeld2022testing,gollakota2022moment,gollakota2023efficient, gollakota2023tester}, we allow a learner to reject unless ${\cal
D}$ and ${\cal D'}$ pass an efficiently computable test.  Whenever the test accepts, the learner outputs a
classifier that is assured to have low error with respect to ${\cal
D'}$.  Further, we require that the test accept with high probability whenever the marginal of ${\cal D}$
equals the marginal of ${\cal D'}$.   This approach allows us to take no assumptions on ${\cal D'}$
whatsoever and still provide meaningful guarantees.

It is easy to see that TDS learning generalizes the traditional PAC
model of learning, and, moreover, TDS learning seems considerably more
challenging.  For example, even an algorithm to amplify
the success probability of a TDS learner is nontrivial, since we do not get to see labeled examples
from ${\cal D'}$ (we show how to do this in \Cref{appendix:success-probability-boosting}).  It is also tempting to apply property testing algorithms in this setting to ``detect" when ${\cal D}$ is ``close" to ${\cal D}'$, but even for simple cases, distribution testing requires an exponential (in the dimension) number of samples (see e.g. \cite{canonne2022topics}).  While testable learning and TDS learning both encounter similar issues, they are fundamentally distinct models. Specifically, the realizable setting, where there exists a classifier with zero train and test loss, is a trivial case in testable learning. We further discuss separations among these models in \Cref{sec: related work}.

\subsection{Our Results}
\label{sec: setup and results}
Here we formally define TDS learning and summarize our main results.  For readability, we have placed some notation and basic definitions in \Cref{appendix:notation}. 
\vspace{-10pt}
\paragraph{Learning Setup.}
Let $\C$ be a function class over $\R^d$ and $\Dgeneric$ be a distribution over $\R^d$. Suppose $\mathcal{A}$ is given as input a set $\Strain$ consisting of i.i.d. examples from $\Dgeneric$ labelled by some $f \in \C$, together with a set of i.i.d. unlabelled examples $\Stest$ from some distribution $\Dtestmarginal$ over $\R^d$. The algorithm $\mathcal{A}$ is  allowed to either output REJECT or (ACCEPT, $\hat{f}$) for some concept $\hat{f}$. The algorithm $\mathcal{A}$ is a \textbf{TDS-learning algorithm} for $\C$ under distribution $\Dgeneric$ if it satisfies the following two properties:
\begin{enumerate}
    \item \textbf{Soundness.} With probability $1-\delta$, if the algorithm $\mathcal{A}$ outputs (ACCEPT, $\hat{f}$), then hypothesis $\hat{f}$ satisfies $\pr_{\x \in \Dtestmarginal}[f(\x) \neq \hat{f}(\x)]\leq \epsilon$.
    \item \textbf{Completeness.} If $\Dtestmarginal=\Dgeneric$, then with probability $1-\delta$, the algorithm $\mathcal{A}$ outputs (ACCEPT, $\hat{f}$).
\end{enumerate}
\paragraph{TDS Learning: the Agnostic Setting.} 
Sometimes the training data or the testing data cannot be captured perfectly by any function in the function class $\C$ and, instead, follow labeled distributions $\Dtrainjoint,\Dtestjoint$, where the marginal of $\Dtrainjoint$ is $\Dtrainmarginal=\Dgeneric$ and $\Dtrainjoint,\Dtestjoint$ are otherwise arbitrary. We extend our setup to apply in this setting as well. To this end, a key quantity is the \textbf{smallest sum} of expected training error and expected test error among all functions in the concept class $\C$, i.e. $\optcommon = \min_{\concept\in\C}\error(\concept;\Dtrainjoint)+\error(\concept;\Dtestjoint)$, where $\error(\concept;\Dgenericjoint) = \pr_{(\x,y)\sim\Dgenericjoint}[y\neq\concept(\x)]$. We denote this quantity as $\optcommon$, and note that it is standard in the domain adaptation literature (see, e.g., \cite{ben2006analysis,blitzer2007learning,ben2010theory,david2010impossibility}). 

With this definition at hand, we modify the soundness condition to require that with probability $1-\delta$, if the algorithm $\A$ outputs (ACCEPT, $\hat{f}$), then hypothesis $\hat{f}$ satisfies $\pr_{(\x,y) \sim \Dtestjoint }[y \neq \hat{f}(\x)]\leq O(\optcommon) + \epsilon$. In \Cref{theorem:error-lower-bound-easy}, we show that a dependence of $\Omega(\optcommon)$ is unavoidable.

\begin{proposition}[Informal]
    No TDS learning algorithm can have an error guarantee better than $\Omega(\optcommon)+\eps$.
\end{proposition}

\paragraph{Results.}
We show that TDS learning can be achieved efficiently for a number of natural high-dimension function classes.
These include 
halfspaces, decision trees, intersections of halfspaces and low-depth formulas. See \Cref{tbl: comparison of our work and previous work} for the full list.



\begin{table*}[!htbp]\begin{center}
\begin{tabular}{c p{3.0cm}  c  c  p{2.5cm}} 
 \hline
 & \textbf{Function class} &  \textbf{Training Distribution} & \textbf{TDS Setting} &\textbf{Run-time }      \\ \midrule
1 & Homogeneous halfspaces  & Isotropic Log-Concave  & Agnostic  &
$\poly\left(d/\epsilon\right)$ (Theorem \ref{theorem:improved tds for halfspaces agnostic}) \\ \midrule
2 & General halfspaces & Standard Gaussian & Realizable  & $d^{O\left(\log 1/\epsilon\right)}\;\;\;$ (Theorem \ref{thm: improved TDS learning for general halfspaces})  \\ \midrule
3 & General halfspaces & \begin{tabular}{c}
      Standard Gaussian \\
    Uniform on $\{\pm 1\}^d$ 
\end{tabular}  & Agnostic  & $d^{\tilde{O}\left(1/\epsilon^2\right)}\;\;\;$ (\Cref{corollary:tds-dts-intersections-of-halfspaces})  \\ \midrule
4 & Intersection of $\nhalfspaces$ halfspaces & 
\begin{tabular}{c}
     Standard Gaussian \\
    Uniform on $\{\pm 1\}^d$ 
\end{tabular}  & Agnostic  & $d^{\widetilde{O}({\nhalfspaces}^6/\eps^2})\;\;\;$ (\Cref{corollary:tds-dts-intersections-of-halfspaces})  \\ \midrule
5 & Decision trees of size $s$ & Uniform on $\{\pm 1\}^d$ & Agnostic  & $d^{O(\log(s/\epsilon))}$ (\Cref{corollary:tds-dtrees})  \\  \midrule
6 & Formulas 
of size $s$, depth $\ell$ 
 & Uniform on $\{\pm 1\}^d$ & Agnostic  & $d^{\sqrt{{\aczsize}} \cdot O\left(\log(\aczsize/\eps)\right)^{\frac{5\aczdepth}{2}}}$ (\Cref{corollary:tds-acz})  \\ 
 \bottomrule
\end{tabular}
\end{center}
\caption{Our TDS learning results for various function classes. Since agnostic TDS learning is more general than realizable TDS learning, algorithms for the agnostic setting also apply to the realizable setting.}
\label{tbl: comparison of our work and previous work}
\end{table*}

Given the abundance of positive results, it is natural to ask whether TDS learning can always be achieved efficiently for any function class $\mathcal{F}$ that can be efficiently PAC-learned under a distribution $\D$. 
We answer this question in the negative by proving separations between TDS learning and PAC learning. Our separations hold for the natural and well-studied function classes of monotone functions over $\{\pm 1\}^d$ and convex sets over $\R^d$ (under uniform distribution on $\{\pm 1\}^d$ and the standard Gaussian distribution respectively). Even though for these function classes there are well-known PAC-learning algorithms \cite{bshouty1996fourier,klivans2008learning} that run in time $2^{\tilde{O}(\sqrt{d} \poly(1/\epsilon))}$, we show that any TDS-learning algorithm for these function classes needs to run in time $2^{\Omega(d)}$. 


\subsection{Techniques}
Here we summarize the technical ideas that we use to develop the TDS learning algorithms in \Cref{tbl: comparison of our work and previous work}.
\paragraph{Moment Matching/Sandwiching Polynomials.} We present a general approach for obtaining TDS learning algorithms for a wide variety of function classes via a \textbf{moment matching} approach. In brief, the algorithm for this approach is as follows:

\begin{itemize}
    \item Estimate all the degree-$k$ moments of $\Dtestmarginal$ up to a high accuracy. REJECT if some of the moments are not close to the corresponding moments of $\D$.
    
    \item Otherwise, fit the best degree-$k$ polynomial $p$ on the training data, and output (ACCEPT, $\sign(p)$).
\end{itemize}

This algorithm above runs in time $d^{O(k)}$, and we show that this algorithm is a valid TDS-learning algoithm for the wide class of functions whose $\L_2$\textbf{-sandwiching degree} is bounded by $k$, which we define as follows: For an approximation parameter $\epsilon$,
the $\L_2$-sandwiching degree of a function $f$ is the smallest degree for a pair of polynomials $\pdown$ and $\pup$ satisfying: i) $\pdown(\x) \leq f(\x) \leq \pup(\x)$ for all $\x$ in the learning domain and ii) ${\E_{\x \sim \Dgeneric}[(\pup(\x)-\pdown(\x))^2]}\leq \epsilon$.


The related notion of $\L_1$-sandwiching was recently used to obtain several results in testable learning \cite{gollakota2022moment}. These results, however, do not seem to apply to TDS learning. 
Instead, we prove a ``transfer lemma" showing that we can relate the test error under ${\Dtestmarginal}$ of a polynomial to its training error under $\D$ by leveraging the simple fact that the squared loss between two polynomials is itself a polynomial.  As such, low-degree moment matching between the training and test marginals ensures that the squared loss between any pair of low-degree polynomials is approximately preserved (\Cref{lemma:transfer-lemma}). Absolute loss cannot be computed by a low-degree polynomial, ruling out this type of transfer lemma based on $\L_1$-sandwiching. 




Even though we need the more stringent property of small $\L_2$-sandwiching degree, we show that constructions from works in the pseudorandomness literature that explicitly construct $\L_1$-sandwiching polynomials  (e.g., \cite{diakonikolas2010bounded} and \cite{gopalan2010fooling}) can be extended to bound the $\L_2$-sandwiching degree. This allows us to obtain efficient TDS learning algorithms for the classes of intersections of halfspaces, decision trees and small-depth formulas (see lines 3-6 in Table \ref{tbl: comparison of our work and previous work}). We also note that this technique yields TDS learning algorithms not only in the realizable setting, but also in the agnostic setting.    

\paragraph{Beyond Moment Matching.}
It is a natural question whether it is possible to beat the moment-matching approach. We answer this question in the affirmative by showing that for the class of halfspaces this is indeed possible. It is a standard fact that one needs polynomials of degree $\tilde{\Omega}(1/\epsilon^2)$ to $\epsilon$-approximate halfspaces up to $\L_1$ error better than $\epsilon$ under the standard Gaussian distribution. Therefore the moment-matching approach requires a run-time of at least $d^{\tilde{\Omega}(1/\epsilon^2)}$ to TDS learn halfspaces under the standard Gaussian. We overcome this obstacle and give a TDS learning algorithm for halfspaces that runs in time $d^{O(\log(\frac{1}{\epsilon}))}$ (Line 2, Table \ref{tbl: comparison of our work and previous work}). 

One ingredient we use to design our algorithm is what we call TDS learning via the \textbf{disagreement region} method. Suppose we are able to recover the parameters of a halfspace $f^*$ up to some accuracy $\beta$. Then, for some points $\x$ in $\R^d$ we will know $f^*(\x)$ with certainty, but for some others we will not. We say that the latter points form the disagreement region, and it gets smaller as $\beta$ decreases. The idea is to (i) use the training data to recover the parameters of halfspace $f^{*}$ up to such high accuracy $\beta$ that the probability that a Gaussian sample falls into the disagreement region is very small (ii) make sure that the recovered halfspace $\hat{f}$ generalizes on the testing dataset by checking that only a small fraction of the testing dataset falls into the uncertainty region. We note that this notion of disagreement region is also widely used in active learning (see discussion in \Cref{section:disagreement-tds}).

Although the disagreement region method gives an efficient algorithm for homogeneous (i.e. origin-centered) halfspaces (\Cref{theorem:improved tds for halfspaces}), it fails for general halfspaces. Indeed, in \Cref{sec: beyond TDS} we show that for general halfspaces under the standard Gaussian distribution the disagreement region method requires $2^{\Omega(d)}$ samples.
We design a $d^{O(\log(1/\epsilon))}$-time TDS learning algorithm for general halfspaces under the Gaussian distribution by \textbf{combining} the moment matching approach with the disagreement region approach: 
\begin{itemize}
    \item Suppose the halfspace $f^{*}$ is not too biased, i.e. among $d^{O(\log(1/\epsilon))}$ training samples we see labels with values of both $+1$ and $-1$. We show that the parameters of such a halfspace can be recovered up to a very high accuracy using only $d^{O(\log(1/\epsilon))}$ additional training samples. This allows us to leverage the disagreement region method to achieve TDS learning.
    \item Otherwise, the halfspace $f^*$ is highly biased and it almost always takes the same label $L$ on a Gaussian input. For such halfspaces there is no hope to recover their parameters with $d^{O(\log(1/\epsilon)}$ samples. Yet,
    we show that using the moment-matching approach with degree parameter $k$ of only $O(\log(1/\epsilon))$ allows us to certify that even under the test distribution $\Dtestmarginal$ the halfspace $f^*$ will be biased and very likely to take the label $L$. Therefore, a predictor $\hat{f}$ that assigns the label $L$ to all points in $\R^d$ will generalize.
\end{itemize}
\paragraph{Techniques from Testable Learning.} Additionally, in the setting of \textbf{agnostic} TDS learning we give an algorithm for the class of homogeneous (i.e. origin-centered) halfspaces under any isotropic log-concave distribution (see line 1 in \Cref{tbl: comparison of our work and previous work}). We achieve this using techniques from testable learning \cite{gollakota2022moment,gollakota2023efficient}. 
The first phase of our TDS learning algorithm uses an approximate agnostic learning algorithm for halfspaces \cite{awasthi2017power, diakonikolas2020non} in order to obtain a vector $\hat{\vv}$, such that the homogeneous halfspace defined by $\hat{\vv}$ has error $O(\lambda)+\epsilon$ in the training dataset. Since the training distribution $\D$ is isotropic and log-concave, this means that the angle between $\hat{\vv}$ and the vector $\vv$, defining the halfspace with optimal combined error on the training and testing datasets, is also at most $O(\lambda)+\epsilon$. Finally, we apply one of the core procedures from \cite{gollakota2022moment,gollakota2023efficient} in order to ensure that every halfspace defined by a vector $\vv'$ that forms an angle of at most $O(\lambda)+\epsilon$ with $\hat{\vv}$ agrees on at least $1-O(\lambda)-\epsilon$ fraction of the testing dataset with the halfspace defined by the vector $\hat{\vv}$. This allows us to certify that the halfspace defined by the vector $\hat{\vv}$ will indeed generalize to the testing distribution. 
Note that we can use tools from testable learning to remove the assumption on the training marginal; the algorithm would instead run a test that accepts when both $\Dtrainmarginal$ and $\Dtestmarginal$ equal the target $\Dgeneric$ without any assumptions on $\Dtrainjoint$ and $\Dtestjoint$ (see also \Cref{remark:tds-and-testable}). For clarity of exposition, we postpone formal statements composing the two models to future work.

\subsection{Related Work}
\label{sec: related work}

\paragraph{Domain Adaptation.} The field of {\em domain adaptation} has received significant attention over the past two decades (see \cite{ben2006analysis,blitzer2007learning,mansour2009domadapt,ben2010theory,david2010impossibility,redko2020survey} and references therein). Similar to our learning setting, domain adaptation considers scenarios where the learner has access to labeled training and unlabeled test examples and is asked to output a hypothesis with low test error without, however, being allowed to reject. \cite{ben2006analysis,blitzer2007learning,mansour2009domadapt} bound the test error of an empirical risk minimizer of training data by a sum of the parameter $\optcommon$ and some notion of distance between the training and test marginals (discrepancy or $d_A$ distances) which is statistically efficient to compute using unlabeled test and training examples. This implies a statistically efficient TDS learning algorithm with error $2\optcommon+\eps$ (\Cref{appendix:sc}). All known algorithms for computing discrepancy distance or $d_A$ distance, however,  require exponential time even for basic classes such as halfspaces and decision trees.
 By allowing the learning algorithm to reject, we design computationally efficient TDS learning algorithms with error $O(\optcommon)+\eps$ without explicitly computing the discrepancy distance.

\paragraph{PQ Learning.} Among the learning models that capture settings with distribution shift, PQ learning (see \cite{goldwasser2020beyond} and \cite{kalai2021efficient}) is most relevant to TDS learning. In PQ learning, the learner has access to labeled training data and unlabeled test data and must output a classifier $h$ and a set $\Xregion$.  The classifier needs to minimize the following two criteria simultaneously: (1) the test error of the hypothesis $h$ on test data points that fall into the region $\Xregion$ (in other words, $\Xregion$ is the region where one is confident in the predictions of the hypothesis $h$ for test data) and (2) the probability that a training example falls into $\Xregion$. \cite{goldwasser2020beyond} show that any concept class that can be agnostically learned in the distribution-free setting can be PQ learned. \cite{kalai2021efficient} improve this reduction by showing that PQ learning is equivalent to distribution-free reliable agnostic learning (see \cite{kalai2012reliable}).  The complexity of reliable learning is known to be ``in between" agnostic learning and PAC learning.  In particular, reliably learning conjunctions implies PAC learning DNF formulas.  In \Cref{sec: PQ implies TDS}, we show that PQ learning actually implies TDS learning.



\paragraph{Testable Learning.}
Although conceptually our definition of TDS learning is inspired by the recent line of work in testable learning \cite{rubinfeld2022testing,gollakota2022moment,gollakota2023efficient, gollakota2023tester}, the two frameworks address very different issues. Testable learning does not address distribution shift, as it assumes that the training and testing distributions are the same distribution $\Dtrainjoint$. What the framework of testable learning does (indirectly) test is whether $\Dtrainmarginal$ satisfies a certain assumption (e.g. Gaussianity) in order to make sure the learning algorithm gives a hypothesis $\hat{f}$ that satisfies the agnostic learning guarantee.

As noted in \cite{rubinfeld2022testing}, in the realizable setting one can trivially satisfy  the definition of testable learning by drawing a fresh set of samples and using them to validate the hypothesis  $\hat{f}$. Due to this, existing work on testable learning \cite{rubinfeld2022testing,gollakota2022moment,gollakota2023efficient, gollakota2023tester} focuses on the agnostic setting, where such validation procedure cannot be applied (see \cite{rubinfeld2022testing} for further detail). In contrast to this, even in the realizable setting, no such validation procedure exists for TDS learning, as indicated by our separations between PAC learning and TDS learning for monotone functions and convex sets (see \Cref{sec: setup and results}). In fact, for monotone functions and convex sets, realizable TDS learning is harder than agnostic learning as well.

Furthermore, there are cases where realizable TDS learning is easier than agnostic learning (and, therefore, easier than testable agnostic learning). Here are two examples:

\begin{enumerate}
    \item Due to statistical query lower bounds and cryptographic hardness results \cite{goel2020statistical,diakonikolas2020near,diakonikolas2021optimality,diakonikolas2023near}, the run-time required to agnostically learn a halfspace under the standard Gaussian distribution is believed to be $d^{\Omega(1/\epsilon^2)}$. In contrast to this, in this work we show that realizable TDS learning of halfspaces with respect to the Gaussian distribution can be achieved using only $d^{O(\log 1/ \epsilon))}$ run-time.
    \item The agnostic learning of parity functions, even under the uniform distribution on $\{\pm 1\}^d$, is believed to require $2^{\Omega (\frac{d}{\poly\log d})}$ time.
    In strong contrast with this, the class of parity functions can be TDS-learned in the realizable setting using only $\poly(d/\epsilon)$ time under any distribution over $\{\pm 1\}^d$. This follows from the PQ-learning algorithm of \cite{kalai2021efficient}, together with the connection between PQ learning and TDS learning (\Cref{sec: PQ implies TDS}). 
\end{enumerate}
Overall, we conclude that realizable TDS learning is incomparable to regular agnostic learning. In particular, there are examples where realizable TDS learning is easier than testable agnostic learning. Moreover, realizable TDS learning is harder than PAC learning, where distributional assumptions can be verified through validation.

\section*{Acknowledgements.}
We thank Aravind Gollakota for insightful discussions during early stages of this project and for feedback on a draft of this manuscript. We also thank Varun Kanade for helpful discussions regarding PQ-learning.

\section{Technical Overview}
\subsection{TDS Learning of Homogeneous Halfspaces}

We provide an efficient TDS learner for the class of homogeneous halfspaces over $\R^d$ with respect to any given isotropic log-concave distribution that achieves error $O(\optcommon)+\eps$, by applying results from prior work in the literature of testable learning (see \cite{gollakota2023efficient,gollakota2023tester}) and agnostic learning (see \cite{daniely2015ptas,awasthi2017power,diakonikolas2020non}). We provide the following theorem and a proof sketch. The full proof can be found in \Cref{appendix:tds-homogeneous-agnostic}.

\begin{theorem}[Agnostic TDS learning of Halfspaces]\label{theorem:improved tds for halfspaces agnostic}
    Let $\C$ be the class of origin-centered halfspaces over $\R^d$ and $C>0$ a sufficiently large universal constant. Let $\A,\T$ be as defined in \Cref{proposition:agnostic-learning-halfspaces,proposition:testably-bounding-halfspace-disagreement}. Let $m_\A$ be the sample complexity of $\A(\eps/C,\delta/4)$ and $m_\T = \frac{Cd^4}{\eps^2\delta}$. Then, there is an algorithm (\Cref{algorithm:improved for halfspaces agnostic}) that, given inputs $\Strain,\Stest$ of sizes $|\Strain| \ge m_\A$ and $|\Stest| \ge m_\T$ is a TDS learning algorithm for $\C$ w.r.t. any isotropic log-concave distribution $\Dgeneric$ with error $O(\optcommon)+\eps$ and run-time $\poly(d,\frac{1}{\eps})\log(\frac{1}{\delta})$, where is the accuracy parameter and $\delta$ is the failure probability.
\end{theorem}
\paragraph{Leveraging training data.} We first use an efficient agnostic learner on training data to recover a halfspace $\hat\concept:\x\to\sign(\hat\vv\cdot \x)$ with low training error. For example, we may use a (polynomial time) algorithm by \cite{diakonikolas2020non} (\Cref{proposition:agnostic-learning-halfspaces}) that outputs $\hat\concept$ with $\error(\hat\concept;\Dtrainjoint)\le O(\eta)+\eps$ whenever the training marginal is isotropic log-concave ($\eta$ is the optimal training error). There are other similar results in the literature of agnostic learning (e.g., see \cite{awasthi2017power}), but we use \cite{diakonikolas2020non} as it is more convenient for our setting. 

\paragraph{Approximate parameter recovery.} Let $\vopt$ be the parameter vector corresponding to the halfspace $\coptcommon$ that minimizes the common train and test error, i.e., $\error(\coptcommon;\Dtrainjoint)+\error(\coptcommon;\Dtestjoint) = \optcommon$. Then, we have $\pr_{\Dtrainmarginal}[\sign(\hat\vv\cdot \x)\neq\sign(\vopt\cdot \x)] \le\error(\hat\concept;\Dtrainjoint)+\error(\coptcommon;\Dtrainjoint) \le O(\eta)+\eps+\optcommon = O(\optcommon)+\eps$. Since $\Dtrainmarginal=\Dgeneric$ is isotropic log-concave, it is known that the disagreement over $\Dtrainmarginal$ between two halfspaces is proportional to the angular distance between their parameters, i.e., $\measuredangle(\hat\vv,\vopt) = O(\pr_{\Dtrainmarginal}[\sign(\hat\vv\cdot \x)\neq\sign(\vopt\cdot \x)])$, which we have bounded by $O(\optcommon+\eps)$.

\paragraph{Testing phase.} We have shown that $\hat{\vv}$ is geometrically close to $\vopt$, which achieves test error at most $\optcommon$, by definition. It remains to certify that the test marginal behaves like an isotropic log-concave distribution with respect to $\hat{\vv}$, i.e., for a large enough set of i.i.d. examples $\Stest$ from $\Dtestmarginal$ and for any $\vv'\in\Sphere^{d-1}$ we have that $\frac{1}{|\Stest|}\sum_{\x\in\Stest}\ind\{\sign(\hat\vv\cdot \x)\neq \sign(\vv'\cdot \x)\} := \pr_{\Stest}[\sign(\hat\vv\cdot \x)\neq \sign(\vv'\cdot \x)] = O(\measuredangle(\hat\vv,\vv'))$, because then we will be able to bound the empirical test error of $\hat\concept$ by $\optcommon+O(\measuredangle(\hat\vv,\vopt))$, which is $O(\optcommon+\eps)$. The result then would follow by standard VC dimension arguments. 

It turns out that recent work by \cite{gollakota2023tester} on testable learning has provided an efficient tester that achieves exactly what we need. Note that the proof of the following proposition (Lemma 3.1 in \cite{gollakota2023tester}) is nontrivial, requiring estimation of low-order moments and careful conditioning. We can apply this to our setting, because it only requires access to the marginal distribution.

\begin{proposition}[Testably Bounding Halfspace Disagreement, Lemma 3.1 in \cite{gollakota2023tester}]\label{proposition:testably-bounding-halfspace-disagreement}
    Let $\Dgeneric$ be a distribution over $\R^d$, $\vv_1\in\Sphere^{d-1}$, $\theta\in(0,\pi/4]$, $\delta\in(0,1)$ and $C>0$ a sufficiently large universal constant. Then, there is an algorithm $\T(\theta,\delta)$ that, upon drawing at least $\frac{Cd^4}{\theta^2\delta}$ examples $\Sunlabelled$ from $\Dgeneric$ and in time $\poly(d,\frac{1}{\theta},\frac{1}{\delta})$ either accepts or rejects and satisfies the following.
    \begin{enumerate}[label=\textnormal{(}\alph*\textnormal{)}]
        \item If $\T$ accepts, then for any $\vv_2\in\R^d$ with $\measuredangle(\vv_1,\vv_2)\le \theta$, it holds 
        \[  
            \pr_{\x\sim \Sunlabelled}[\sign(\vv_1\cdot \x)\neq \sign(\vv_2\cdot \x)] \le C\measuredangle(\vv_1,\vv_2)
        \]
        \item\label{condition:proposition-testably-bounding-halfspace-disagreement} If $\Dgeneric$ is isotropic log-concave, then $\T$ accepts with probability at least $1-\delta$.
    \end{enumerate}
\end{proposition}

\begin{remark}\label{remark:universal-tds}
    We note that, in fact, the original version of \Cref{proposition:testably-bounding-halfspace-disagreement} in \cite{gollakota2023tester} does not require the target marginal to be known, but works universally for any isotropic log-concave distribution (as well as distributions with heavier tails). This implies that the completeness criterion that \Cref{algorithm:improved for halfspaces agnostic} satisfies is actually much stronger: for an appropriate choice of the absolute constant $C$, \Cref{algorithm:improved for halfspaces agnostic} can be made to accept whenever $\Dtestmarginal$ is isotropic log-concave (and not necessarily equal to the training marginal).
\end{remark}

\begin{remark}\label{remark:tds-and-testable}
    Moreover, we point out that we can apply results from \cite{gollakota2023tester} and substitute algorithm $\A$ with a universal tester-learner for halfspaces.  This enables us to remove the assumption that $\Dtrainmarginal$ is some fixed isotropic log-concave distribution, and the final algorithm would accept with high probability whenever $\Dtrainmarginal$ is isotropic strongly log-concave {\em and} $\Dtestmarginal$ is isotropic log-concave. In that sense, TDS learning composes well with (universally) testable learning. For sake of presentation, however, we leave formal compositional arguments to future work.
\end{remark}

\subsection{TDS Learners for General Halfspaces}

\subsubsection{Warm-Up: Disagreement-Based TDS Learning}\label{section:disagreement-tds}

We provide a general TDS learner for the realizable setting, based on the notion of disagreement regions from active learning.  Not only is this approach interesting in and of itself, but it will also be useful in \Cref{sec: beyond TDS} where we present our main result for TDS learning of general halfspaces in the realizable setting. The main idea is to testably bound the probability that a test example falls in some region $\disagreementregion$, whose mass with respect to the target distribution becomes smaller as the number of training samples increases and, also, the output of the training algorithm achieves low error on any distribution that assigns small mass to $\disagreementregion$. It turns out that the quantity $\pr_{\x\sim\Dgeneric}[\x\in \disagreementregion]$, where $\Dgeneric$ is some given distribution over a space $\X\subseteq\R^d$, is a well-studied notion in the literature of active learning (see \cite{cohn1994improving,hanneke2009theoretical,balcan2006agnostic,hanneke2011rates,hanneke2014theory,balcan2010true,hanneke2007bound} and references therein). We now provide a formal definition for the disagreement region.

\begin{definition}[Disagreement Region]\label{definition:disagreement-region}
    Let $\X\subseteq\R^d$, $\Dgeneric$ a distribution over $\X$ and $\C$ a concept class of functions that map $\X$ to $\cube{}$. For $\eps>0$ and $\concept\in\C$, we define the $\eps$-disagreement region of $\concept$ under $\Dgeneric$, $\disagreementregion_\eps(\concept;\Dgeneric)$ as the subset of $\X$ such that if $\x\in \disagreementregion_\eps(\concept;\Dgeneric)$, then there are $\concept_1,\concept_2\in \C$ with $\error(\concept_1,\concept;\Dgeneric) \le \eps$, and $ \error(\concept_2,\concept;\Dgeneric) \le \eps$ and $\concept_1(\x)\neq\concept_2(\x)$. 
\end{definition}

In the literature of active learning, the quantity of interest is called the disagreement coefficient and is defined for a concept class $\C$ and a distribution $\Dgeneric$ as follows (see, e.g., \cite{hanneke2014theory}).
\begin{equation}
    \disagreementcoef(\eps) = \sup_{\concept\in\C}\sup_{\eps'>\eps}\frac{\pr_{\x\sim\Dgeneric}[\x\in\disagreementregion_{\eps'}(\concept;\Dgeneric)]}{\eps'}\label{equation:disagreement-coefficient}
\end{equation}
In particular, for active learning, is is crucial that $\disagreementcoef(\eps)$ is asymptotically bounded by a slowly increasing function of $1/\eps$ (e.g., $O(\log(1/\eps))$), since bounds on the disagreement coefficient directly provide rates on the label complexity of disagreement-based active learning, up to logarithmic factors \cite{hanneke2011rates}. In our setting, meaningful results are obtained even when $\disagreementcoef(\eps) = O(1/\eps^{1-c})$ for any constant $c\in(0,1)$. Moreover, we also focus on the dependence of the disagreement coefficient on other relevant parameters, like the dimension $d$. To emphasize this, in what follows, we will use the notation $\disagreementcoef(\eps,d)$ to refer to the disagreement coefficient. We obtain the following result, which implies, for example, a polynomial improvement in the sample complexity bound of realizable TDS learning of homogeneous halfspaces w.r.t. the Gaussian compared to the TDS learner we proposed in \Cref{theorem:improved tds for halfspaces agnostic} for the agnostic setting (see also \Cref{appendix:disagreement-tds}).

\begin{theorem}[Disagreement-Based TDS learning]\label{theorem:disagreement-tds}
    Let $\C$ be the class of concepts that map $\X\subseteq\R^d$ to $\cube{}$ with VC dimension $\vc(\C)$, let $\Dgeneric$ a distribution over $\X$ and $C>0$ a sufficiently large universal constant. Suppose that we have access to an ERM oracle for PAC learning $\C$ under $\Dgeneric$ and membership access to $\disagreementregion_{\eps'}(\concept;\Dgeneric)$ for any given $\concept\in\C$ and ${\eps'}>0$. Then, there is an algorithm (\Cref{algorithm:disagreement-tds}) that given inputs of sizes $|\Strain| \ge C\frac{\vc(\C)}{\eps'}\log(\frac{1}{\eps'\delta})$ and $|\Stest| \ge C(\frac{\vc(\C)}{\eps}+\frac{1}{\eps^2})\log(\frac{1}{\eps\delta})$ is a TDS learning algorithm for $\C$ w.r.t. $\Dgeneric$ that calls the $\eps'$-ERM oracle once and the $\eps'$-membership oracle $|\Strain|$ times, where $\eps$ is the accuracy parameter, $\delta$ is the failure probability and $\eps'$ such that $\eps'\cdot \disagreementcoef(\eps',d) \le \eps/2$.
\end{theorem}

\subsubsection{Beyond Disagreement: TDS Learners for General Halfspaces}
\label{sec: beyond TDS}

We give a TDS-learning algorithm for the class of halfspaces under the standard Gaussian distribution. The algorithm runs in quasi-polynomial time in all relevant parameters and, contrary to the case of homogeneous halfspaces, works in a setting where efficient parameter recovery is not possible. This happens because when a general halfspace has arbitrarily large bias, it is possible, for example, that all of the training examples have the same label. 

In particular, applying a pure disagreement-based TDS learning framework (\Cref{theorem:disagreement-tds}) in the case of {\em general halfspaces} can only give exponential-time algorithms for this problem.  To illustrate this, imagine that the ground truth is a general halfspace with bias $\tau = \sqrt{d}$ but unknown direction $\vv\in\Sphere^{d-1}$. Then, any general halfspace $\x\mapsto\sign(\vv'\cdot \x-\tau)$ with the same bias is $\exp(-\Omega(d))$-close to the ground truth with respect to the Gaussian distribution, due to standard Gaussian concentration, i.e., $\pr_{\x\sim\Gauss(0,I_d)}[\sign(\vv\cdot \x -\tau)\neq\sign(\vv'\cdot \x - \tau)] \le \pr_{\x\sim\Gauss(0,I_d)}[\sign(\vv\cdot \x -\tau)\neq\sign(-\vv\cdot \x - \tau)]$, which is upper bounded by $\pr_{\x\sim\Gauss(0,I_d)}[|\vv\cdot \x|>\sqrt{d}] \le 2\exp(-d/2)$. Let $\eps'=2\exp(-d/2)$. Suppose that ERM returns a halfspace $\hat\concept$ that is $\eps'$-close to the ground truth but has bias $\tau$. Any $\x\in\R^d$ with $\|\x\|_2\ge \sqrt{d}$, falls within the disagreement region $\disagreementregion_{\eps'}(\hat\concept;\Gauss(0,I_d))$ and therefore $\pr_{\x\sim\Gauss(0,I_d)}[\x\in \disagreementregion_{\eps'}(\hat\concept;\Gauss(0,I_d))]$ is constant. This implies that running the ERM oracle on training data even up to exponentially small accuracy $\eps'=\exp(-\Omega(d))$ does not meet the requirement of \Cref{theorem:disagreement-tds} (see also \cite{el2012active}) that the disagreement coefficient is bounded as $\eps'\cdot \disagreementcoef(\eps',d)\le \eps/2$.

In order to overcome this obstacle, we perform a case analysis that depends on the bias of the unknown halfspace. If the bias is bounded, then we may use a disagreement-based approach, since we can approximately recover the true parameters of the unknown halfspace using training data and it suffices to verify that the test distribution does not amplify the error between any pair of halfspaces close to the obtained approximations of the true parameters.  Now, consider the case when the bias is large. We may assume without loss of generality the constant hypothesis $+1$ has low training error (since the ground truth has large bias and the marginal is Gaussian). If we can certify that the test marginal is sufficiently concentrated in every direction, then this hypothesis must also have small test error. To certify concentration for the test distribution's marginals, we use a moment-matching approach. Checking the moment matching condition only up to degree $O(\log(\epsilon))$ turns out to be sufficient to certify the type of concentration we need. We thus obtain a quasi-polynomial TDS learning algorithm for general halfspaces with respect to the Gaussian distribution. Since the probability of success can be amplified through repetition (see \Cref{proposition:boosting-success-probability}), we provide a result with constant failure probability. For the full proof, see \Cref{appendix:tds-general-halfspaces}.

\begin{theorem}[TDS learning of General Halfspaces]
	\label{thm: improved TDS learning for general halfspaces}
    Let $\C$ be the class of general halfspaces over $\R^d$ and $C>0$ a sufficiently large universal constant. Then, there is an algorithm (\Cref{algorithm:improved for general halfspaces}) that, given inputs of size $|\Strain|=|\Stest|=C d^{C\log {1}/{\epsilon}}$
	is a TDS learning algorithm for $\C$ w.r.t. $\Gauss(0,I_d)$ with run-time $d^{O(\log {1}/{\epsilon})}$, where $\epsilon$ is the accuracy parameter, and the failure probability $\delta$ is at most $0.01$. 
\end{theorem}

Compared to \Cref{theorem:disagreement-tds}, our approach here incurs an increase in the amount of test samples required (from $\poly(d,1/\eps) $ to $d^{O(\log(1/\eps))}$, used for moment matching) but significantly decreases the amount of training samples required (from $\exp(\Omega(d))$ to $d^{O(\log(1/\eps))}$).  

\subsection{TDS Learning through Moment Matching}

In the previous section, we provided a TDS learner for general halfspaces in the realizable setting that requires ideas beyond parameter recovery and testably bounding the probability of falling in the disagreement region. Crucially, \Cref{thm: improved TDS learning for general halfspaces} uses a moment-matching approach in the case when the bias of the unknown halfspaces is large. As is explained in this section, we show that the moment-matching approach can actually provide a generic result which demonstrates that $\L_2$-sandwiching (see \Cref{definition:l2-sandwiching}) implies TDS learning, even in the non-realizable setting. We also instantiate our framework to several important concept classes (halfspace intersections, decision trees and Boolean formulas) with respect to the Gaussian and uniform distributions, by applying constructions from pseudorandomness literature to bound the $\L_2$-sandwiching degree of each of these classes and acquire entries 3-6 in \Cref{tbl: comparison of our work and previous work}.


We provide a general theorem, which demonstrates that $\L_2$-sandwiching implies TDS learning under some additional natural assumptions about the target marginal distribution, which are satisfied by the standard Gaussian distribution over $\R^d$ and the uniform distribution on $\{\pm 1\}^d$. While it is known that $\L_1$-sandwiching implies testable learning (see \cite{gollakota2022moment}), we require the stronger notion of $\L_2$-sandwiching. In particular, while $\L_1$-sandwiching would (testably) imply the existence of low degree polynomials with low test error, we do not get to see labeled examples from $\Dtestjoint$. Moreover, we cannot a priori assume that the output of the training algorithm is a sandwiching polynomial, even if we know one exists.   

In our analysis, we crucially use the fact that the square of the difference between two polynomials is itself a polynomial whose coefficients and degree are bounded by the degree and coefficient bounds of the original polynomials. Crucially, this enables us to use the following transfer lemma which relates the squared distance between polynomials under the test distribution to their squared distance under the training distribution. In what follows, we use the notation $\x^\mindex = \prod_{i\in[d]}\x_i^{\mindex_i}$, where $\mindex\in\N^d$.

\begin{lemma}[Informal, Transfer Lemma for Square Loss, see \Cref{lemma:transfer-lemma-formal}]\label{lemma:transfer-lemma}
    Let $\Dgeneric$ be a distribution over $\X\subseteq\R^d$ and $\Stest$ a (multi)set of points in $\R^d$. If $\E_{\x\sim\Stest}[\x^\mindex]\approx \E_{\x\sim\Dgeneric}[\x^\mindex]$ for all $\mindex\in\N^d$ with $\|\mindex\|_1\le 2\degbound$, then for any degree $\degbound$ polynomials $p_1,p_2$ with bounded coefficients, it holds
    \[
        \frac{1}{|\Stest|}\sum_{\x\in \Stest}(p_1(\x)-p_2(\x))^2 \approx \E_{\x\sim\Dgeneric}[(p_1(\x)-p_2(\x))^2]
    \]
\end{lemma}

Moreover, we use the fact that, due to the $\L_2$-sandwiching assumption, we can bound quantities of the form $\E[(p(\x)-\concept(\x))^2]$ for $\concept\in\C$ from above by $O(\E[(p(\x)-\pdown(\x))^2] + \E[(\pdown(\x)-\pup(\x))^2])$, irrespective of the distribution that the expectations are taken over. Over the training distribution, the quantity $\E_{\Dgeneric}[(\pdown(\x)-\pup(\x))^2]$ is small via the definition of $\L_2$-sandwiching degree, and the quantity $\E_{\Dgeneric}[(p(\x)-\concept(\x))^2]$ because $p$ is obtained from $\L_2$ polynomial regression. If $p,\pdown,\pup$ are all low degree and the 
dataset $\Stest$
matches low-degree moments with $\Dgeneric$, then we may apply \Cref{lemma:transfer-lemma} to bound $ \frac{1}{|\Stest|} \sum_{\x \in \Stest}[(p(\x)-\concept(\x))^2]$. Once it is shown that $p$ fits $f$ well on the testing dataset $\Stest$, standard generalization bounds allows us to conclude that it will also predict $f$ well on the testing distribution.
Therefore, by running polynomial regression on training data to obtain $p$ and testing whether the empirical test moments match the moments of the training distribution, we acquire the following result, whose proof can be found in \Cref{appendix:tds-moment-matching}.

\begin{theorem}[$\L_2$-sandwiching implies TDS Learning]\label{theorem:l2-sandwiching-implies-tds}
    Let $\Dgeneric$ be a distribution over a set $\X\subseteq\R^d$ and let $\C\subseteq\{\X\to \cube{}\}$ be a concept class. Let $\eps,\delta\in(0,1)$, $\eps'=\eps/100$ $\delta'=\delta/2$ and assume that the following are true.
    \begin{enumerate}[label=\textnormal{(}\roman*\textnormal{)}]
        \item
        ($\L_2$-Sandwiching) The $\eps'$-approximate $\L_2$ sandwiching degree of $\C$ under $\Dgeneric$ is at most $\degbound$ with coefficient bound $\pbound$. 
        \item
        (Moment Concentration) If $\Sunlabelled\sim\Dgeneric^{\otimes m}$ and $m\ge \mconc$ then, with probability at least $1-\delta'$, we have that for any $\mindex\in \N^d$ with $\|\mindex\|_1\le k$ it holds $|\E_\Dgeneric[\x^\mindex]-\frac{1}{|\Sunlabelled|}\sum_{\x\in\Sunlabelled}\x^\mindex | \le \frac{\eps'}{\pbound^2d^{4\degbound}}$.
        \item
        (Generalization) 
        If $\Slabelled\sim\Dgenericjoint^{\otimes m}$ where $\Dgenericjoint$ is any distribution over $\X\times\cube{}$ such that $\Dgenericmarginal = \Dgeneric$ and $m\ge \mgen$ then, with probability at least $1-\delta'$ we have that for any degree-$k$ polynomial $p$ with coefficient bound $\pbound$ it holds $|\E_{\Dgenericjoint}[(y-p(\x))^2]- \frac{1}{|\Slabelled|}\sum_{(\x,y)\in\Slabelled}[(y-p(\x))^2]|\le {\eps'}$.
    \end{enumerate}

    Then, there is an algorithm (\Cref{algorithm:l2-sandwiching}) that, upon receiving $\mtrain \ge  \mgen$ labelled samples $\Strain$ from the training distribution and $\mtest\ge C\cdot \frac{d^\degbound+\log(1/\delta)}{\eps^2} + \mconc$ unlabelled samples $\Stest$ from the test distribution (where $C>0$ is a sufficiently large universal constant), runs in time $\poly(|\Strain|,|\Stest|, d^{\degbound})$ and TDS learns $\C$ with respect to $\Dgeneric$ up to error $32\optcommon+\eps$ and with failure probability $\delta$.
\end{theorem}

\subsection{Lower Bounds}

We provide three lower bounds for TDS learning. The first one shows that TDS learning the class of monotone functions over $\cube{d}$ with respect to the uniform distribution requires an exponential number of examples from either the training or the test distribution, which implies a separation with regular agnostic learning. The second lower bound shows that TDS learning the class of indicators of convex sets also requires an exponential in the dimension number of samples. The third lower bound demonstrates that a linear dependence on the error term $\optcommon$ (as defined in \Cref{equation:definition:optcommon}) is necessary for TDS learning in the non-realizable setting.

\subsubsection{Lower Bound for Monotone Functions and Convex Sets in Realizable Setting}

Recent work on testable learning (which is a generalization of the classical agnostic learning framework, see \cite{rubinfeld2022testing,gollakota2022moment}) has demonstrated that the class of monotone functions over $\cube{d}$ cannot be testably learned with respect to the uniform distribution unless the learner draws at least $2^{\Omega(d)}$ training samples. Since the class of monotone functions can be agnostically learned in time $2^{\tilde{O}(\sqrt{d})}$ with respect to the uniform distribution over the hypercube $\cube{d}$, this implies that testable (agnostic) learning is strictly harder than regular agnostic learning. We show that the lower bound of $2^{\Omega(d)}$ extends to the problem of TDS learning monotone functions even in the realizable setting. Recall that we have shown that we can TDS learn halfspaces with respect to the standard Gaussian distribution in the realizable setting in time $d^{O(\log(1/\eps))}$ (\Cref{thm: improved TDS learning for general halfspaces}) but it is known that, for agnostic learning, any SQ algrorithm for the problem requires time $d^{\Omega(1/\eps^2)}$ (see \cite{goel2020statistical,diakonikolas2020near,diakonikolas2021optimality}). Therefore, we conclude that realizable TDS learning and agnostic learning are incomparable. We now provide our lower bound. For the proof, see \Cref{appendix:lower-bounds}.


\begin{theorem}[Hardness of TDS Learning Monotone Functions]
\label{thm: tds learning of monotone is hard}
Let the accuracy parameter $\epsilon$ be at most $0.1$ and the success probability parameter $\delta$ also be at most $0.1$. Then,
    in the realizable setting, any TDS learning algorithm for the class of monotone functions over $\{\pm 1\}^d$ with accuracy parameter requires either $2^{0.04 d}$ training samples or $2^{0.04 d}$ testing samples for all sufficiently large values of $d$.
\end{theorem}

We now provide a lower bound for convex sets (see also \Cref{appendix:lower-bounds}). Since the class of indicators of convex sets can be agnostically learned in time $2^{\tilde{O}(\sqrt{d})}$ with respect to the Standard Gaussian on $\R^d$, the following theorem implies yet another separation between agnostic learning and realizable TDS learning in the distribution specific setting under the Gaussian distribution for a well-studied concept class.


\begin{theorem}[Hardness of TDS Learning Convex Sets]
	\label{thm: tds learning of convex is hard}
	Let the accuracy parameter $\epsilon$ be at most $0.1$ and the success probability parameter $\delta$ also be at most $0.1$. Then,
	in the realizable setting, any TDS learning algorithm for the class of indicators of convex sets under the standard Gaussian distribution on $\R^d$ requires either $2^{0.04 d}$ training samples or $2^{0.04 d}$ testing samples for all sufficiently large values of $d$.
\end{theorem}

\begin{remark}
    In \Cref{theorem:tds-via-pq} of the Appendix, we show that TDS learning is not harder than PQ learning (which is a related learning primitive, see \cite{goldwasser2020beyond,kalai2021efficient}). \cite{kalai2021efficient} show that the class of parities over $\cube{d}$ can be efficiently PQ learned, which provides another example where TDS learning is easier than agnostic learning.
\end{remark}

\subsubsection{Lower Bound for the Error Guarantee in the Agnostic Setting}

We now focus on the agnostic setting and provide an information theoretic lower bound on the error upon acceptance. Our lower bound is simple and demonstrates that a linear dependence on the error factor $\optcommon$ (see \Cref{equation:definition:optcommon}) is unavoidable for TDS learning.

\begin{theorem}[Lower Bound for the Error in the Agnostic Setting]\label{theorem:error-lower-bound-easy}
    Let $\X$ be any domain, 
    $\Dgeneric$ a distribution over $\X$ 
    and $\C$ a class of concepts that map $\X$ to $\cube{}$ that is closed under complement, i.e., if $\concept\in \C$ then $-\concept\in\C$. Then, for any $\eta\in(0,1/2)$, any $\eps\in(0,\eta/2)$ and $\delta\in(0,1/3)$, no TDS learning algorithm for $\C$ w.r.t. $\Dgeneric$ with finite sample complexity and failure probability $\delta$, can have an error guarantee better than $\optcommon(1-2\eta)+\eps = \Omega(\optcommon)+\eps$.
\end{theorem}

\begin{proof}
    Let $\Dtrainjoint$ denote the training distribution and $\Dtestjoint$ the test distribution, which are both over $\X\times\cube{}$.
    Suppose that for $\eta\in(0,1/2)$ and $\eps\in(0,\eta/2)$ there exists an algorithm $\A$, that, upon acceptance and with probability at least $1-\delta$, outputs $\hat \concept\in\C$ with $\error(\hat\concept;\Dtestjoint) \le \optcommon(1-2\eta)+\eps$ ($\optcommon =\optcommon(\C;\Dtrainjoint,\Dtestjoint)$, see \Cref{equation:definition:optcommon}). Let $C>0$ be a sufficiently large universal constant.

    We consider the following algorithm $\T$. Algorithm $\T$ uses an oracle to $\A$ and accepts or rejects according to the following criteria.
    \begin{itemize}
        \item If $\A$ rejects, then $\T$ rejects.
        \item If $\A$ accepts and outputs $\hat \concept\in \C$, then $\T$ draws $\frac{C}{\eta^2}\log(1/\delta)$ examples $\Slabelled_{\T}$ from $\Dtrainjoint$ and rejects if $\pr_{(\x,y)\in\Slabelled_{\T}}[\hat\concept(\x)\neq y] > 3\eta/4$. Otherwise, $\T$ accepts.
    \end{itemize}

    Fix some $\concept\in \C$ and let $\Dtrainjoint$ be the distribution over $\X\times\cube{}$ whose marginal on $\X$ is $\Dgeneric$ and the labels are generated as $y(\x) = \concept(\x)$. Consider the following two cases about $\Dtestjoint$. 
    \paragraph{Case 1.} First, suppose that $\Dtestjoint$ has $\Dgeneric$ as marginal on $\X$ and $y(\x) = \concept(\x)$. Then, $\A$ accepts with probability at least $1-\delta$, due to completeness. We have $\optcommon = 0$ (attained by $\concept$) and, hence, upon acceptance, $\pr_{(\x,y)\sim\Dtrainjoint}[\hat\concept(\x)\neq y] = \pr_{(\x,y)\sim\Dtestjoint}[\hat\concept(\x)\neq y] \le \eps \le \eta/2$ with probability at least $1-\delta$. By a Hoeffding bound, we then have that $\T$ must accept with probability at least $1-\delta$.
    Overall, $\T$ accepts with probability at least $1-3\delta > 1/2$.
    
    \paragraph{Case 2.} Second, suppose that $\Dtestjoint$ has $\Dgeneric$ as marginal on $\X$ and $y(\x) = -\concept(\x)$. Then, we have that $\optcommon = 1$ (because for any point $\x\in\X$, any classifier will either classify $\x$ incorrectly under $\Dtrainjoint$ or under $\Dtestjoint$). By assumption, we have $\pr_{(\x,y)\sim\Dtestjoint}[\hat\concept(\x)\neq y] \le {\optcommon}{(1-2\eta)}+\eps \le{1-2\eta}+\eps$ with probability at least $1-2\delta$ (by completeness and soundness). Since the test labels are the negation of the train labels, we have $\pr_{(\x,y)\sim\Dtestjoint}[\hat\concept(\x)\neq y] = 1 - \pr_{(\x,y)\sim\Dtrainjoint}[\hat\concept(\x)\neq y]$, and $\pr_{(\x,y)\sim\Dtrainjoint}[\hat\concept(\x)\neq y] \ge 2\eta -\eps \ge \eta$ (since $\eps\le \eta/2$). By a Hoeffding bound, $\T$ will reject with probability at least $1-3\delta>1/2$.

    We have reached a contradiction, because in both cases, the input of $\T$ does not depend on the test labels, and everything else remains the same in both cases. Therefore, $\T$ should have the same behavior in both cases and we conclude that the algorithm $\A$ cannot exist as defined.
\end{proof}

\begin{remark}
    While the above lower bound demonstrates that the error of a TDS learning algorithm can be necessarily high in certain settings, we emphasize that the construction corresponds to a contrived case where the training distribution does not provide enough information about the test distribution and, therefore, any meaningful notion of learning should be hopeless (see also \cite{BenDavid2012OnTH}).
\end{remark} 

\section{Notation and Basic Definitions}\label{appendix:notation}

We let $\X\subseteq \R^d$ and, in particular, $\X$ will either be the $d$-dimensional hypercube $\cube{d}$ or the $d$-dimensional Euclidean space $\R^d$. For a distribution $\Dgeneric$ over $\X$, we use $\E_\Dgeneric$ (or $\E_{\x\sim \Dgeneric}$) to refer to the expectation over distribution $\Dgeneric$ and for a given (multi)set $\Sunlabelled$, we use $\E_\Sunlabelled$ (or $\E_{\x\sim \Sunlabelled}$) to refer to the expectation over the uniform distribution on $\Sunlabelled$ (i.e., $\E_{\x\sim \Sunlabelled}[g(\x)] = \frac{1}{|\Sunlabelled|}\sum_{\x\in \Sunlabelled}g(\x)$, counting possible duplicates separately). We let $\R_+ = (0,\infty)$. 

For a function $p:\X\to \R$ and $r\in\N$, we define the $\L_r$ norm of $p$ under $\Dgeneric$ as $\|p\|_{\L_r(\Dgeneric)} = {\ex_{\x\sim \Dgeneric}[p(\x)^r]}^{\frac{1}{r}}$. 
For $\x\in \X$ where $\x = (\x_1,\x_2,\dots,\x_d)$ and for $\mindex \in \N^d$, we denote with $\x^\mindex$ the product $\prod_{i\in[d]}\x_i^{\mindex_i}$, $\moment_\mindex = \E[\x^\mindex]$ and $\|\mindex\|_1 = \sum_{i\in[d]}\mindex_i$. For a polynomial $p$ over $\R^d$ and $\mindex\in \N^d$, we denote with $p_\mindex$ the coefficient of $p$ corresponding to $\x^\mindex$, i.e., we have $p(\x) = \sum_{\mindex\in\N^d}p_\mindex \x^\mindex$. If $p$ is a polynomial over $\cube{d}$, then we can always express it in a unique multilinear form, so we will only use coefficients $p_\mindex$ with $\mindex \in \{0,1\}^d$, i.e., $p(\x) = \sum_{\mindex\in \{0,1\}^d} p_\mindex \x^\mindex$. We define the degree of $p$ and denote $\deg(p)$ the maximum degree of a monomial whose coefficient in $p$ is non-zero, i.e., $\deg(p) = \max\{\|\mindex\|_1: p_\mindex \neq 0\}$. 

We denote with $\S^{d-1}$ the $d-1$ dimensional sphere on $\R^d$. For any $\vv_1,\vv_2\in\R^d$, we denote with $\vv_1\cdot \vv_2 $ the inner product between $\vv_1$ and $\vv_2$ and we let $\measuredangle (\vv_1,\vv_2)$ be the angle between the two vectors, i.e., the quantity $\theta\in [0,\pi]$ such that $\|\vv_1\|_2\|\vv_2\|_2\cos(\theta) = \vv_1\cdot \vv_2$. For $\vv\in\R^d,\tau\in \R$, we call a function of the form $\x\mapsto \sign(\vv\cdot \x)$ an origin-centered (or homogeneous) halfspace and a function of the form $\x\mapsto \sign(\vv\cdot \x - \tau)$ a general halfspace over $\R^d$. 

\paragraph{$\L_2$-sandwiching degree.} We now define the notion of $\L_2$-sandwiching polynomials for a function $\concept$ with respect to a distribution $\Dgeneric$, i.e., a pair of polynomials such that one of them is pointwise above $\concept$, the other one is pointwise below $\concept$ and the $\L_2$ distance between the two polynomials with respect to $\Dgeneric$ is small. While the notion of $L_1$ sandwiching polynomials is standard in the literature of pseudorandomness (see, e.g., \cite{bazzi2009polylogarithmic}) and has applications to testable learning (see \cite{gollakota2023efficient}), in order to obtain our main results, we make use of the stronger notion of $\L_2$-sandwiching polynomials which we define below.

\begin{definition}[$\L_2$-sandwiching polynomials]\label{definition:l2-sandwiching}
    Consider a product set $\X$ and a distribution $\Dgeneric$ over $\X$. For $\eps>0$ and $\concept:\X\to \cube{}$, we say that the polynomials $\pup, \pdown:\X \to \R$ are $\eps$-approximate $\L_2$-sandwiching polynomials for $\concept$ under $\Dgeneric$ if the following are true.
    \begin{enumerate}
        \item $\pdown(\x) \le \concept(\x) \le \pup(\x)$, for all $\x \in \X$.
        \item $\|\pup-\pdown\|_{\L_2(\Dgeneric)}^2 \le \eps$
    \end{enumerate}
    Moreover, for $\eps>0$, a concept class $\C\subseteq\{\X\to \cube{}\}$ and $k,\pbound>0$, we say that the $\eps$-approximate $\L_2$-sandwiching degree of $\C$ under $\Dgeneric$ is at most $k$ and with coefficient bound $\pbound$ if for any $\concept\in\C$ there are $\eps$-approximate $\L_2$-sandwiching polynomials $\pup, \pdown$ for $\concept$ such that $\deg(\pup),\deg(\pdown) \le k$ and each of the coefficients of $\pup,\pdown$ are absolutely bounded by $\pbound$.
\end{definition}

\paragraph{Learning Setup.} 
Consider $\Dtrainjoint, \Dtestjoint$ to be distributions over $\X\times \cube{}$ and let $\Dtrainmarginal, \Dtestmarginal$ be the corresponding marginal distributions on $\X\subseteq\R^d$. Our tester-learners receive labelled examples from $\Dtrainjoint$ and unlabelled examples from $\Dtestmarginal$ and their goal is to produce a hypothesis with low error on $\Dtestjoint$ or potentially reject but only if distribution shift is detected. Given a hypothesis class $\C\subseteq\{\X\to \cube{}\}$, $h_1, h_2:\X\to\cube{}$ and distributions $\Dgenericjoint, \Dtrainjoint, \Dtestjoint$ over $\X\times\cube{}$, we define $\error(h_1; \Dgenericjoint) = \pr_{(\x,y)\sim \Dgenericjoint}[y \neq h_1(\x)]$ and $\error(h_1, h_2; \Dgenericmarginal) = \pr_{\x\sim \Dgenericmarginal}[h_1(\x) \neq h_2(\x)]$ as well as the following quantity, which is standard in the domain adaptation literature (see, e.g., \cite{ben2006analysis,blitzer2007learning,ben2010theory,david2010impossibility}).
\begin{equation}
    \optcommon(\C;\Dtrainjoint,\Dtestjoint) := \min_{\concept\in \C}\{ \error(\concept; \Dtrainjoint) + \error(\concept; \Dtestjoint)\}, \text{ attained by }\coptcommon\in \C \label{equation:definition:optcommon}
\end{equation}
Observe that parameter $\optcommon$ becomes small whenever the training and test errors can be simultaneously minimized by a common classifier in $\C$. Clearly, if there is no relationship between the training and test distributions, then using data from the training distribution does not reveal any information about the test distribution and, therefore, learning is hopeless (see also \Cref{theorem:error-lower-bound-easy}). We will assume (as is common in the domain adaptation literature) that the parameter $\optcommon$ is a valid choice for quantifying the relationship between the training and test distributions, in the sense that considering $\optcommon$ to be small is not unrealistic. In particular, we will partly focus on the following setting where $\optcommon$ is zero. To distinguish between the two settings, we say that we are in the \textbf{agnostic setting} when $\optcommon\ge 0$ (arbitrary) and in the \textbf{realizable setting} when $\optcommon=0$. When $\optcommon=0$, there exists a classifier in $\C$ that achieves both zero training loss and test loss and we therefore refer to this setting as realizable. Another (slightly more specific) way to view the realizable setting is by considering the labelled distribution $\Dtrainjoint$ (resp. $\Dtestjoint$) formed as follows: for some $\coptcommon\in \C$, draw an example $\x$ from $\Dtrainmarginal$ (resp. $\Dtestmarginal$) and form the pair $(\x,y) \sim \Dtrainjoint$ (resp. $(\x,y)\sim \Dtestjoint$) by setting $y = \coptcommon(\x)$.
We now provide a formal definition of our learning model.

\begin{definition}[Testable Learning with Distribution Shift (TDS Learning)]
\label{definition:tds-learning}
    Let $\X\subseteq \R^d$ and consider a distribution $\Dgeneric$ over $\X$ and a concept class $\C\subseteq\{\X\to \cube{}\}$. For some $\psi:[0,1]\to [0,1]$ and $\eps,\delta\in(0,1)$, we say that an algorithm $\A$ testably learns $\C$ with distribution shift w.r.t. $\Dgeneric$ up to error $\psi(\optcommon)+\eps$ and probability of failure $\delta$ if the following is true. For any distributions $\Dtrainjoint, \Dtestjoint$ over $\X\times\cube{}$ such that $\Dtrainmarginal = \Dgeneric$, algorithm $\A$, upon receiving a large enough set of labelled samples $\Strain$ from the training distribution $\Dtrainjoint$ and a large enough set of unlabelled samples $\Stest$ from the test distribution $\Dtestmarginal$, either rejects $(\Strain,\Stest)$ or accepts and outputs a hypothesis $h:\X\to \cube{}$ with the following guarantees.
    \begin{enumerate}[label=\textnormal{(}\alph*\textnormal{)}]
    \item (Soundness.) With probability at least $1-\delta$ over the samples $\Strain,\Stest$ we have: 
    
    If $A$ accepts, then the output $h$ satisfies $\error(h;\Dtestjoint) \leq \psi(\optcommon) + \eps$.
    \item (Completeness.) Whenever $\Dtestmarginal = \Dtrainmarginal$, $A$ accepts with probability at least $1-\delta$ over the samples $\Strain,\Stest$.
    \end{enumerate}
    In particular, we say that $\A$ testably learns $\C$ with distribution shift w.r.t. $\Dgeneric$ in the realizable setting, if $\A$ is required to satisfy the above guarantees only when $\Dtrainjoint,\Dtestjoint$ and $\C$ are realizable (where $\optcommon = 0 = \psi(\optcommon)$).
    \end{definition}

\section{TDS Learning of Homogeneous Halfspaces}\label{appendix:tds-homogeneous-agnostic}

We now provide a proof of \Cref{theorem:improved tds for halfspaces agnostic}, which we restate here for convenience.

\begin{theorem}[Agnostic TDS learning of Halfspaces]
    Let $\C$ be the class of origin-centered halfspaces over $\R^d$ and $C>0$ a sufficiently large universal constant. Let $\A,\T$ be as defined in \Cref{proposition:agnostic-learning-halfspaces,proposition:testably-bounding-halfspace-disagreement}. Let $m_\A$ be the sample complexity of $\A(\eps/C,\delta/4)$ and $m_\T = \frac{Cd^4}{\eps^2\delta}$. Then, \Cref{algorithm:improved for halfspaces agnostic}, given inputs $\Strain,\Stest$ of sizes $|\Strain| \ge m_\A$ and $|\Stest| \ge m_\T$ is a TDS learning algorithm for $\C$ w.r.t. any isotropic log-concave distribution $\Dgeneric$ with error $O(\optcommon)+\eps$ and run-time $\poly(d,\frac{1}{\eps})\log(\frac{1}{\delta})$, where $\eps$ is the accuracy parameter and $\delta$ is the failure probability.
\end{theorem}

\begin{algorithm}
	\caption{Agnostic TDS Learning of Halfspaces}\label{algorithm:improved for halfspaces agnostic}
	\KwIn{Sets $\Strain$ from $\Dtrainjoint$, $\Stest$ from $\Dtestmarginal$, parameters $\eps>0$, $\delta\in(0,1)$}
		Set $\eps' = \eps/C$ where $C$ is some sufficiently large universal constant.\\
        Let $m_\A$ be the sample complexity of $\A(\eps',\delta/4)$.\\
        Split $\Strain$ to $\Slabelled_1,\Slabelled_2$ with sizes $m_\A,\frac{C}{\eps^2}\log(1/\delta)$ \\
        Run $\A(\eps',\delta/4)$ on $\Slabelled_1$ and obtain $\hat\vv\in\Sphere^{d-1}$ \\
        Let $\hat{\eps} = \pr_{(\x,y)\sim\Slabelled_2}[\sign(\hat\vv\cdot\x)\neq y]$. \\
        Run $\T(\hat\eps,\delta/2)$ on $\Stest$. \\
        \textbf{Reject} and terminate if $\T$ rejects. \\
        \textbf{Otherwise,} output $\hat\concept:\R^d\to\cube{}$ with $\hat\concept:\x\to \sign(\hat\vv\cdot \x)$.
\end{algorithm}

In order to prove the above theorem, we make use of the following result from \cite{diakonikolas2020non}.

\begin{proposition}[Agnostic Learning of Homogeneous Halfspaces, Theorem 3.1 in \cite{diakonikolas2020non}]\label{proposition:agnostic-learning-halfspaces}
    Let $\Dgenericjoint$ be a distribution over $\R^d\times\cube{}$ such that its marginal on $\R^d$ is isotropic log-concave. Then there is an algorithm $\A$ such that for any $\eps>0$ and $\delta\in(0,1)$, $\A(\eps,\delta)$, upon drawing $m = \tilde{O}(\frac{d}{\eps^4}\log(1/\delta))$ independent examples from $\Dgenericjoint$ and in time $\poly(d,1/\eps)\cdot \log(1/\delta)$, outputs $\hat\vv\in\Sphere^{d-1}$ such that, with probability at least $1-\delta$, the corresponding halfspace has error at most $O(\eta)+\eps$, where $\eta$ is the error of the optimal halfspace on $\Dgenericjoint$.
\end{proposition}

We also use the following fact about isotropic log-concave distributions.

\begin{fact}\label{fact:log-concave-disagreement}
    $\pr_{\x\sim \Dgeneric}[\sign(\hat\vv\cdot \x)\neq \sign(\vv^*\cdot \x)] = \Theta(\measuredangle (\hat\vv,\vv^*))$, when $\Dgeneric$ is isotropic log-concave.
\end{fact}

\begin{proof}
    Suppose that $\Strain$ is a set of $\mtrain$ independent samples from $\Dtrainjoint$, where the marginal of $\Dtrainjoint$ on $\R^d$ is the standard Gaussian distribution. Let also $\Stest$ be a set of $\mtest$ independent unlabelled samples from $\Dtestmarginal$. In what follows, let $\eps'=\eps/C$ and let $C>0$ be a sufficiently large universal constant. Let also $m_\A$ be the sample complexity of $\A(\eps',\delta/4)$ and $m_\T = \frac{Cd^4}{\eps^2 \delta}$.

    \paragraph{Soundness.} Suppose that the algorithm accepts. Let $\vopt\in\S^{d-1}$ define the halfspace $\coptcommon$ that achieves $\error(\coptcommon;\Dtestjoint)+\error(\coptcommon;\Dtrainjoint) = \optcommon$. Note that since $|\Slabelled_2|\ge \frac{C}{\eps^2}\log(1/\delta)$, we have that $\hat\eps \le \error(\hat\concept;\Dtrainjoint)+\eps'$. By \Cref{proposition:agnostic-learning-halfspaces}, since $|\Slabelled_1|\ge m_\A$ we have $\error(\hat\concept;\Dtrainjoint)\le \eta+\eps'$, where $\eta\in(0,1)$ is the error of the optimum halfspace, say $\concept:\x\mapsto\sign(\vv\cdot \x)$ on $\Dtrainjoint$. Note that $\eta\le \optcommon$. We have that $\error(\hat\concept,\concept;\Dtrainmarginal) \le \error(\hat\concept;\Dtrainjoint) + \error(\concept;\Dtrainjoint) \le 2\eta+\eps'$. Therefore, due to \Cref{fact:log-concave-disagreement}, and since $\Dtrainmarginal=\Dgeneric$, we obtain $\measuredangle(\hat\vv,\vv) \le 2C'\eta + C'\eps'$ for some sufficiently large $C'>0$ (with $C\gg C'$). 

    Moreover, we have that $\error(\coptcommon;\Dtrainjoint) \le \optcommon$ and, hence $\error(\coptcommon,\concept;\Dtrainmarginal) \le \optcommon+\eta$. 
    We now apply \Cref{proposition:testably-bounding-halfspace-disagreement}, to obtain $\error(\hat\concept,\coptcommon;\Stest) \le \sqrt{C}\measuredangle(\hat\vv,\vopt)$. Since $|\Stest|\ge \frac{\sqrt{C}}{\eps^2}\log(1/\delta)$, due to standard VC dimension arguments, we have $\error(\hat\concept,\coptcommon;\Dtestmarginal) \le \sqrt{C}\measuredangle(\hat\vv,\vopt) + \eps'$. By \Cref{fact:log-concave-disagreement}, $\measuredangle(\hat\vv,\vopt) \le C'\error(\hat\concept,\coptcommon;\Dtrainmarginal)$. Therefore, with probability at least $1-\delta$, we have
    \begin{align*}
        \error(\hat\concept;\Dtestjoint) &\le \error(\hat\concept,\coptcommon;\Dtestmarginal)+\error(\coptcommon;\Dtestjoint) \le \sqrt{C}\error(\hat\concept,\coptcommon;\Dtrainmarginal) + \eps' + \optcommon \\
        &\le \sqrt{C}\error(\hat\concept,\concept;\Dtrainmarginal) + \sqrt{C}\error(\concept,\coptcommon;\Dtrainmarginal) + \eps' + \optcommon\\
        &\le C\optcommon + C\eps' \le \eps
    \end{align*}
    
    \paragraph{Completeness.} Readily follows from \Cref{proposition:testably-bounding-halfspace-disagreement} and $|\Stest|\ge m_\T$.
\end{proof}

\section{Realizable TDS Learning}\label{appendix:realizable-tds}

\subsection{Disagreement-Based TDS Learners}\label{appendix:disagreement-tds}

In this section, we prove \Cref{theorem:disagreement-tds}. First, we prove the following a special version regarding realizable TDS learning of homogeneous halfspaces with respect to the Gaussian distribution.

\begin{proposition}[TDS learning of Homogeneous Halfspaces]\label{theorem:improved tds for halfspaces}
    Let $\C$ be the class of origin-centered halfspaces over $\R^d$ and $C>0$ a sufficiently large universal constant. Then, \Cref{algorithm:improved for halfspaces}, given inputs $\Strain,\Stest$ of sizes $|\Strain| \ge C(\frac{d}{\eps})^{3/2}\log(\frac{1}{\eps\delta})$ and $|\Stest| \ge C(\frac{d}{\eps}+\frac{1}{\eps^2})\log(\frac{1}{\eps\delta})$ is a TDS learning algorithm for $\C$ w.r.t. the standard Gaussian distribution $\Gauss(0,I_d)$ with run-time $\poly(d,1/\eps)\log(\frac{1}{\delta})$, where $\eps$ is the accuracy parameter and $\delta$ is the failure probability.
\end{proposition}

\begin{algorithm}
	\caption{TDS Learning of Homogeneous Halfspaces}\label{algorithm:improved for halfspaces}
	\KwIn{Sets $\Strain$ from $\Dtrainjoint$, $\Stest$ from $\Dtestmarginal$, parameter $\eps>0$}
		Set $\eps' = \eps^{3/2}/(10d^{1/2})$.\\
        Run the Empirical Risk Minimization algorithm on $\Strain$ up to error $\eps'$, i.e., compute a vector $\hat\vv\in\Sphere^{d-1}$ with
        $
            \hat\vv = \argmin_{\vv'\in\Sphere^{d-1}} \pr_{(\x,y)\in\Strain}[y\neq \sign(\vv'\cdot \x)]
        $
        \\
        Let $\paramneighborhood = \{\vv'\in\Sphere^{d-1}: \|\vv'-\hat\vv\|_2 \le \eps'\}$.\\
        For each $\x\in\Stest$, compute the following quantities.
        \[\vv^+_\x = \argmax_{\vv'\in \paramneighborhood}\vv'\cdot \x \text{ and }\vv^-_\x = \argmin_{\vv'\in \paramneighborhood}\vv'\cdot \x\]\\
        \textbf{Reject} and terminate if $\pr_{\x\sim\Stest}[\sign(\vv_\x^+\cdot \x)\neq \sign(\vv_\x^-\cdot \x)] > 3\eps/4$.\\
        \textbf{Otherwise,} output $\hat{f}:\R^d\to\cube{}$ with $\hat f:\x\mapsto \sign(\hat{\vv} \cdot \x )$.
\end{algorithm}

We will use the following fact about the Gaussian distribution.

\begin{fact}\label{fact:gaussian-disagreement}
    For any $\vv_1,\vv_2\in\Sphere^{d-1}$ we have $\pr_{\x\sim \Gauss(0,I_d)}[\sign(\vv_1\cdot \x)\neq \sign(\vv_2\cdot \x)] = \measuredangle (\vv_1,\vv_2) / \pi$.
\end{fact}

\begin{proof}[Proof of \Cref{theorem:improved tds for halfspaces}]
    Suppose that $\Strain$ is a set of $\mtrain$ independent samples from $\Dtrainjoint$, where the marginal of $\Dtrainjoint$ on $\R^d$ is the standard Gaussian distribution. Let also $\Stest$ be a set of $\mtest$ independent unlabelled samples from $\Dtestmarginal$. In what follows, let $\eps'=\eps^{3/2}/(8d^{1/2})$.

    \paragraph{Soundness.} When the algorithm accepts, we have that $\pr_{\x\sim\Stest}[\sign(\vv_\x^+\cdot \x)\neq \sign(\vv_\x^-\cdot \x)] \le \frac{3\eps}{2}$. By standard VC dimension arguments and \Cref{fact:gaussian-disagreement}, after running the Empirical Risk Minimization algorithm on training data, as long as $\mtrain \ge C\frac{d}{\eps'}\log(1/(\delta\eps'))$, we have $\|\hat \vv- \vv\|_2 \le \eps'$. Therefore, both $\vv$ and $\hat \vv$ are within $\paramneighborhood = \{\vv'\in\Sphere^{d-1}:\|\vv'-\hat\vv\|_2\le \eps'\}$. By the definition of $\vv_\x^+$ and $\vv_\x^-$, we have the following.
    \begin{align}
        \pr_{\x\sim\Stest}[\sign(\hat\vv\cdot \x)\neq \sign(\vv\cdot \x)] \le \pr_{\x\sim\Stest}[\sign(\vv_\x^+\cdot \x)\neq \sign(\vv_\x^-\cdot \x)] \le 3\eps/4 \label{equation:disagreement-bound}
    \end{align}
    Moreover, we have $\error(\hat f;\Dtestjoint) = \E[\pr_{\x\sim\Stest}[\sign(\hat\vv\cdot \x)\neq \sign(\vv\cdot \x)]]$, where the expectation is over $\Stest\sim(\Dtestmarginal)^{\otimes \mtest}$. By standard VC dimension arguments, we have that, with probability at least $1-\delta/2$, $\error(\hat f;\Dtestjoint) = \pr_{\x\sim\Stest}[\sign(\hat\vv\cdot \x)\neq \sign(\vv\cdot \x)] + \eps/4$ whenever $\mtest \ge C\frac{d}{\eps}\log(1/(\delta\eps)) $. Therefore, with probability at least $1-\delta$ (union bound over two bad events), upon acceptance, we have $\error(\hat f; \Dtestjoint) \le \eps.$

    \paragraph{Completeness.} For completeness, we assume that $\Stest$ is drawn from $\Gauss(0,I_d)$. Observe that $\paramneighborhood$ does not depend on $\Stest$ (since it is formed only using training data). Therefore, we may apply a standard Hoeffding bound to ensure that with probability at least $1-\delta$, whenever $\mtest \ge C\frac{1}{\eps^2}\log(1/(\delta))$, we have 
    \[ 
        \pr_{\x\sim\Stest}[\sign(\vv_\x^+\cdot \x)\neq \sign(\vv_\x^-\cdot \x)] \le \pr_{\x\sim\Gauss(0,I_d)}[\sign(\vv_\x^+\cdot \x)\neq \sign(\vv_\x^-\cdot \x)]+\eps/4
    \]
    It remains to bound $\pr_{\x\sim\Gauss(0,I_d)}[\sign(\vv_\x^+\cdot \x)\neq \sign(\vv_\x^-\cdot \x)]$ by $\eps/2$. We observe that, since $\vv^+,\vv^-\in\paramneighborhood$, we have $\vv_\x^- \cdot \x \ge \vv_\x^+\cdot \x - \|\vv_\x^+-\vv_\x^-\|_2\|\x\|_2 \ge \vv_\x^+\cdot \x - \eps'\|\x\|_2 \ge \hat\vv\cdot \x - \eps'\|\x\|_2$ by the definition of $\vv^+_\x$ and $\vv^-_\x$. We similarly have $\vv_\x^+ \cdot \x \le \hat\vv\cdot \x+\eps'\|\x\|_2$. 
    
    Therefore the probability that $\sign(\vv_\x^+\cdot \x)\neq \sign(\vv_\x^-\cdot \x)$ is upper bounded by the probability that $|\hat\vv\cdot \x| \le \eps'\|\x\|_2$ (since, otherwise, both $\vv_\x^+\cdot \x$ and $\vv_\x^-\cdot \x$ have the same sign). In particular
    \begin{align*}
        \pr_{\x\sim\Gauss(0,I_d)}[\sign(\vv_\x^+\cdot \x)\neq \sign(\vv_\x^-\cdot \x)] &\le \pr_{\x\sim\Gauss(0,I_d)}[|\hat\vv\cdot \x| \le \eps'\|\x\|_2] \\
        &\le \pr_{\x\sim\Gauss(0,I_d)}[\|\x\|_2>\sqrt{4d/\eps}]+  \pr_{\x\sim\Gauss(0,I_d)}[|\hat\vv\cdot \x|\le \eps'\sqrt{4d/\eps}] \\
        &\le \frac{\E_{\x\sim\Gauss(0,I_d)}[\|\x\|_2^2]\eps}{4d} + \pr_{\x\sim\Gauss(0,I_d)}[|\hat\vv\cdot \x|\le \eps'\sqrt{4d/\eps}]
    \end{align*}
    We obtain the final inequality by applying Markov's inequality. Since $\E_{\x\sim\Gauss(0,I_d)}[\|\x\|_2^2] = d$ and the one-dimensional Gaussian density is upper bounded by $(2\pi)^{-1}$, we have the following bound.
    \[
        \pr_{\x\sim\Gauss(0,I_d)}[\sign(\vv_\x^+\cdot \x)\neq \sign(\vv_\x^-\cdot \x)] \le \frac{\eps}{4}+\frac{2}{\sqrt{2\pi}} \eps'\sqrt{4d/\eps} \le \eps/2\,,
    \]
    since $\eps' \le \eps^{3/2}/(8d^{1/2})$. This completes the proof.
\end{proof}

We now prove \Cref{theorem:disagreement-tds}, which we restate here for convenience.

\begin{theorem}[Disagreement-Based TDS learning]
    Let $\C$ be the class of concepts that map $\X\subseteq\R^d$ to $\cube{}$ with VC dimension $\vc(\C)$, let $\Dgeneric$ a distribution over $\X$ and $C>0$ a sufficiently large universal constant. Suppose that we have access to an ERM oracle for PAC learning $\C$ under $\Dgeneric$ and membership access to $\disagreementregion_{\eps'}(\concept;\Dgeneric)$ for any given $\concept\in\C$ and ${\eps'}>0$. Then, \Cref{algorithm:disagreement-tds}, given inputs of sizes $|\Strain| \ge C\frac{\vc(\C)}{\eps'}\log(\frac{1}{\eps'\delta})$ and $|\Stest| \ge C(\frac{\vc(\C)}{\eps}+\frac{1}{\eps^2})\log(\frac{1}{\eps\delta})$ is a TDS learning algorithm for $\C$ w.r.t. $\Dgeneric$ that calls the $\eps'$-ERM oracle once and the $\eps'$-membership oracle $|\Strain|$ times, where $\eps$ is the accuracy parameter, $\delta$ is the failure probability and $\eps'$ such that $\eps'\cdot \disagreementcoef(\eps',d) \le \eps/2$.
\end{theorem}

\begin{algorithm}
	\caption{Disagreement-Based TDS Learning}\label{algorithm:disagreement-tds}
	\KwIn{Sets $\Strain$ from $\Dtrainjoint$, $\Stest$ from $\Dtestmarginal$, parameter $\eps>0$}
		Set $\eps' >0$ such that $\eps'\cdot \disagreementcoef(\eps',d) \le \eps/2$.\\
        Run the Empirical Risk Minimization algorithm on $\Strain$ up to error $\eps'$, i.e., compute $\hat\concept \in \C$ with
        $
            \hat\concept = \argmin_{\concept'\in\C} \pr_{(\x,y)\in\Strain}[y\neq \concept'(\x)]
        $
        \\
        Let $\disagreementregion_{\eps'}(\hat\concept;\Dgeneric)$ be as in \Cref{definition:disagreement-region}.\\
        \textbf{Reject} and terminate if $\pr_{\x\sim\Stest}[\x\in \disagreementregion_{\eps'}(\hat\concept;\Dgeneric)] > \eps/2$.\\
        \textbf{Otherwise,} output $\hat{f}$.
\end{algorithm}

\begin{proof}[Proof of \Cref{theorem:disagreement-tds}]
    Suppose that $\Strain$ is a set of $\mtrain$ independent samples from $\Dtrainjoint$, where the marginal of $\Dtrainjoint$ on $\X$ is the distribution $\Dgeneric$. Let also $\Stest$ be a set of $\mtest$ independent unlabelled samples from $\Dtestmarginal$. In what follows, let $\eps'>0$ such that $\eps'\disagreementcoef(\eps',d)\le \eps/2$. The proof follows a similar recipe as the one of \Cref{theorem:improved tds for halfspaces}. For the following, let $\coptcommon\in\C$ be the label generating function.

    \paragraph{Soundness.} Suppose that the algorithm accepts. Then, $\pr_{\x\sim\Stest}[\x\in \disagreementregion_{\eps'}(\hat\concept;\Dgeneric)] \le \eps/2$. Since $\hat \concept$ is an minimizes the empirical error on training data, by standard VC arguments, we have that $\error(\hat\concept,\coptcommon;\Dgeneric) \le \eps/2$, whenever $\mtrain\ge C\frac{\vc(\C)}{\eps'}\log(\frac{1}{\eps'\delta})$, since $\Dtrainmarginal = \Dgeneric$ by assumption. Therefore, by the definition of $\disagreementregion_{\eps'}(\hat\concept;\Dgeneric)$, for any $\x \not\in \disagreementregion_{\eps'}(\hat\concept;\Dgeneric)$, we have $\hat\concept(\x) = \coptcommon(\x)$. Therefore, we have
    \[
        \pr_{\x\sim\Stest}[\hat\concept(\x) \neq \coptcommon(\x)] \le \pr_{\x\sim\Stest}[\x\in\disagreementregion_{\eps'}(\hat\concept;\Dgeneric)] \le \eps/2
    \]
    Whenever $\mtest\ge C\frac{\vc(\C)}{\eps}\log(\frac{1}{\eps\delta})$, we have $\pr_{\x\sim\Dtestjoint}[y \neq \coptcommon(\x)]\le \pr_{\x\sim\Stest}[\hat\concept(\x) \neq \coptcommon(\x)]+\eps/2 \le \eps$.

    \paragraph{Completeness.} Suppose that $\Dtestmarginal = \Dgeneric$. Then, by a standard Hoeffding bound, we have that whenever $\mtest \ge C\frac{1}{\eps}\log(1/\delta)$, we have $\pr_{\x\sim\Stest}[\x\in \disagreementregion_{\eps'}(\hat\concept;\Dgeneric)] \le  \pr_{\x\sim\Dgeneric}[\disagreementregion_{\eps'}(\hat\concept;\Dgeneric)] + \eps/2$ with probability at least $1-\delta$ and $\pr_{\x\sim\Dgeneric}[\disagreementregion_{\eps'}(\hat\concept;\Dgeneric)] \le \eps'\disagreementcoef(\eps',d) \le \eps/2$, by the choice of $\eps'$.
\end{proof}

\subsection{TDS Learner for General Halfspaces}\label{appendix:tds-general-halfspaces}

We now prove \Cref{thm: improved TDS learning for general halfspaces} which we restate here for convenience.

\begin{theorem}[TDS learning of General Halfspaces]
    Let $\C$ be the class of general halfspaces over $\R^d$ and $C>0$ a sufficiently large universal constant. Then, \Cref{algorithm:improved for general halfspaces}, given inputs of size $|\Strain|=|\Stest|=C d^{C\log {1}/{\epsilon}}$
	is a TDS learning algorithm for $\C$ w.r.t. the standard Gaussian distribution $\Gauss(0,I_d)$ with run-time $d^{O(\log {1}/{\epsilon})}$, where $\epsilon$ is the accuracy parameter, and the failure probability $\delta$ is at most $0.01$. 
\end{theorem}

\begin{algorithm}
	\caption{TDS Learning of General Halfspaces}\label{algorithm:improved for general halfspaces}
	\KwIn{Sets $\Strain$ from $\Dtrainjoint$, $\Stest$ from $\Dtestmarginal$, parameter $\eps>0$}
	\begin{algorithmic}[1]
\STATE		Set $T =2^{C_1^2 \log \frac{1}{\epsilon}+1}$, $\degbound = C_1 \log \frac{1}{\epsilon}$, $\Delta = \frac{\epsilon}{d^{C_2 \degbound}}$ and $\beta=\frac{\epsilon^2}{C_3 d^{C_3}}$.
	\IF{ $\pr_{(\x,y) \sim \Strain}[y\neq b] \leq \frac{1}{T}$ for some $b\in\cube{}$ (large bias case)}
		\STATE For each $\mindex\in \N^d$ with $\|\mindex\|_1 \le \degbound$, compute the quantity 
		$\momentempirical_\mindex = \E_{\x\sim\Stest} [\x^\mindex]$. \\
		\STATE 
   \label{line: used to be 4b}
   \textbf{Reject} and terminate if $|\momentempirical_\mindex-\E_{\x\sim \Gauss(0,I_d)}[\x^\mindex]|>\Delta$ for some $\mindex$ with $\|\mindex\|_1 \le \degbound$. 
   \STATE
  \label{line: used to be 4c}
  \textbf{Otherwise,} output 
  $\hat{f}:\R^d\to\cube{}$ and terminate, where $\hat{f}:\x\mapsto b$ ($\hat f$ constant).
    \ELSE
\STATE       Set $\hat{\vv}  =\frac{\E_{(\x,y)\sim \Strain}[y\x]}{ \|\E_{(\x,y)\sim \Strain}[y\x]\|_2}$.
\STATE      Let $\T = \{\hat{\vv}\cdot \x: (\x,y)\in \Strain\}$.
\STATE        Set
		$
		\hat{\tau}
		=
		\argmin_{\tau \in \T}
		\pr_{(\x,y) \in \Strain}[\coptcommon(\x)\neq \sign(\hat{\vv} \cdot \x - \tau')]
		$, 
         \STATE
        Let $\paramneighborhood = \{(\vv',\tau'): \|\vv'-\hat\vv\|_2 \le \beta, |\tau'-\hat\tau|\le \beta\}$.
        \STATE
        For each $\x\in\Stest$, compute the following quantities.
        \[(\vv^+_\x,\tau^+_\x) = \argmax_{(\vv',\tau')\in \paramneighborhood}\vv'\cdot \x-\tau' \text{ and }(\vv^-_\x,\tau^-_\x) = \argmin_{(\vv',\tau')\in \paramneighborhood}\vv'\cdot \x-\tau'\]
        \STATE
        \label{line: used to be 5c}
        \textbf{Reject} and terminate if $\pr_{\x\sim\Stest}[\sign(\vv_\x^+\cdot \x-\tau^+_\x)\neq \sign(\vv_\x^-\cdot \x-\tau^-_\x)] > 10\eps$.
        \STATE
          \label{line: used to be 5d}
        \textbf{Otherwise,} output $\hat{f}:\R^d\to\cube{}$ with $\hat f:\x\mapsto \sign(\hat{\vv} \cdot \x -\hat{\tau})$.
        \ENDIF
        \end{algorithmic}
\end{algorithm}

Suppose the ground-truth halfspace $\coptcommon(\x)=\sign(\x \cdot \vv -\tau)$ is determined by a unit vector $\vv \in \R^d$ and a value $\tau \in \R$. We will need the following showing that if a halfspace not too biased under the standard Gaussian distribution, then it is possible to recover the parameters of the halfspace up to a very high accuracy. See Subsection \ref{sec: parameter recovery} for the proof.

\begin{proposition}
	[Parameter recovery for halfspaces]
	\label{prop: parameter recovery} For a sufficiently large absolute constant $C>0$, the following is true.
	For every $\beta,\gamma \in (0,1)$ and integer $d$, let $\Strain$ be a set of $C(\frac{d}{ \beta \gamma })^C$ i.i.d samples from a distribution $\Dtrainjoint$ such that $\Dtrainmarginal = \Gauss(0,I_d)$ and the labels are given by an unknown halfspace $f:\x\mapsto\sign(\vv \cdot \x - \tau)$. Additionally, assume that the halfspace $f$ satisfies $\pr_{\x \in \Gauss(0,I_d)}[\coptcommon(\x)=-1]\ge \gamma$ and $ \pr_{\x \in \Gauss(0,I_d)}[\coptcommon(\x)=1]\geq \gamma$. Let $\T = \{\hat{\vv}\cdot \x: (\x,y)\in \Strain\}$ and set \[\hat{\vv}  = \frac{\sum_{(\x,y) \in \Strain} \x y}{\norm{\sum_{(\x,y) \in \Strain} \x y}_2} \text{ and }~
	\hat{\tau}
	=
	\argmin_{\tau' \in \T}
	\pr_{(\x,y) \in \Strain}[\coptcommon(\x)\neq \sign(\hat{\vv} \cdot \x - \tau')].
	\]
	Then, with probability at least $1-1/1000$ we have 
	$
	\norm{\vv-\hat{\vv}}_2
	\leq \beta \text{ and }
	|\tau-\hat{\tau}|
	\leq \beta.
	$
\end{proposition}
We also highlight two technical lemmas that we use for the analysis of Algorithm \ref{algorithm:improved for general halfspaces}. 
Our first technical lemma insures that if $f$ is a halfspace that very likely assigns the same label to samples from the Gaussian distribution, then $f$ also very likely assigns the same label to samples form a distribution whose low-degree moments match those of a Gaussian. This lemma will be useful for proving the soundness of Algorithm \ref{algorithm:improved for general halfspaces}, and is proven in Section \ref{sec: proof of Completeness}. (Recall that for $\x\in \R^d$ we denote $\prod_{i=1}^{n} x_i^{\mindex_i}$ as $\x^{\mindex}$.) 
\begin{lemma}
	\label{prop: if keep seeing the same label moment matching guarantees generalization}
	When $C_1$ and $C_2$ both exceed some specific absolute constant, the following holds.
	Let $k$ and $T$ be defined as in Algorithm \ref{algorithm:improved for general halfspaces}.
	Suppose, the set $\Stest$ is such that for every collection of non-negative integers $(\mindex_1, \cdots, \mindex_d)$ satisfying $\sum_i \mindex_i \leq k$ we have
	\begin{equation}
		\label{eq: moment matching}
		\left\lvert
		\ex_{\x \sim \Stest} 
		\left[\x^\alpha 
		\right]
		-
		\ex_{\x \sim \Gauss(0,I_d)} 
		\left[\x^\alpha 
		\right]
		\right\rvert
		\leq\frac{\epsilon}{d^{C_2 k}}.
	\end{equation}
	Also, suppose the function $\coptcommon(\x)=\sign(\x \cdot \vv -\tau)$ and the value $L\in \{\pm 1\}$ are such that 
	\begin{equation}
		\label{eq: halfspace is small}
		\pr_{\x \sim \mathcal{N}(0,1)}[\coptcommon(\x)\neq L] \leq \frac{2}{T}.
	\end{equation}
	Then, it is the case that 
	\begin{equation}
		\pr_{\x \sim \Stest}[\coptcommon(\x)\neq L]\leq O(\epsilon).
	\end{equation}
\end{lemma}

Our second technical lemma bounds, for $\x$ chosen from the standard Gaussian, the probability that one is unsure about $\coptcommon(\x)=\sign(\vv\cdot \x -\tau)$ when one only has approximate estimates for $\hat{\vv}$ and $\hat{\tau}$ for $\vv$ and $\tau$ respectively. This lemma will be useful for proving the completeness of Algorithm \ref{algorithm:improved for general halfspaces}, and is proven in Section \ref{sec: proof of soundness}.
\begin{lemma}
	\label{prop: gaussian is unlikely to fall into disagreement region}
	There is some absolute constant $K_1$, such that for every positive integer $d$ and $\beta\in (0,1)$, the following holds.    Let $\hat{\vv}$ be any unit vector in $\R^d$ and $\hat{\tau}$ be in $\R$. Then, we have for $\paramneighborhood = \{(\vv',\tau'): \|\vv'-\hat\vv\|_2 \le \beta, |\tau'-\hat\tau|\le \beta\}$
	\begin{equation}
		\label{eq: size of disagreement region small}
		\pr_{\x \sim \Gauss(0,I_d)}
		\left[
		\sign\left(
		\max_{
			(\vv',\tau')\in \paramneighborhood}
		\vv'\cdot \x -\tau'
		\right)
		\neq \sign\left(
		\min_{
			(\vv',\tau')\in \paramneighborhood}
		\vv'\cdot \x -\tau'
		\right)
		\right]
		\leq
		K_1  d^{K_1} \sqrt{\beta}
	\end{equation}
\end{lemma}

\subsubsection{Proof of Soundness.}
\label{sec: proof of soundness}
In this subsection we show that if Algorithm \ref{algorithm:improved for general halfspaces} accepts then the output $\hat{f}$ of our algorithm will generalize on the distribution $\Dtestmarginal$. 
\begin{proposition}[Soundness]
	\label{prop: if accepts then generalizes}
	For any sufficiently large absolute constant $C$, the following is true. For any distribution $\Dtestmarginal$ and any halfspace $f=\sign(\hat{\vv} \cdot \x -\hat{\tau})$, the following is true.
	It can happen with probability only at most $\frac{1}{100}$ that \Cref{algorithm:improved for general halfspaces} gives an output (ACCEPT, $\hat{f}$) for some predictor $\hat{f}$, but it is not the case that
	\[
	\pr_{\x \sim \Dtestmarginal}
	[\coptcommon(\x)\neq \hat{f}(\x)]
	\leq O(\epsilon).
	\]
\end{proposition}
To prove this proposition, we first need to prove \Cref{prop: if keep seeing the same label moment matching guarantees generalization}.
\begin{proof}[Proof of \Cref{prop: if keep seeing the same label moment matching guarantees generalization}]
	First of all, we claim that Equation \ref{eq: halfspace is small} implies that
	\begin{equation}
		\label{eq: tau is large}
		|\tau|
		\geq
		\sqrt{\frac{1}{2} \log \frac{T}{2}} 
	\end{equation}
	Indeed, we have
	\[
	\frac{2}{T}
	\geq
	\frac{1}{\sqrt{2\pi}}
	\int_{|\tau|}^\infty
	e^{-z^2/2} d z
	\geq
	|\tau|
	e^{-2|\tau|^2}
	\geq
	e^{-2|\tau|^2}, 
	\]
	where the last inequality holds because for sufficiently large $C_1$ the value of $T$ and therefore $|\tau|$ is sufficiently large and exceeds $1$. 
	
	Recall that $\vv$ is assumed to be a unit vector in $\R^d$. Assume, without loss of generality, that $L=-1$, and therefore $\tau>0$. We have  
	\begin{equation}
		\label{eq: degree-k Chabychev}
		\pr_{\x \sim \Stest}[\sign(\x \cdot \vv -\tau)\neq -1] =
		\pr_{\x \sim \Stest}[\x \cdot \vv \geq \tau]
		\leq
		\frac{\ex_{\x \sim \Stest}[(\x \cdot \vv)^k]}{\tau^k}.
	\end{equation}
	To use this inequality, we need to upper-bound $\ex_{\x \sim \Stest}[(\x \cdot \vv)^k]$. Since $\vv$ is a unit vector, every (of at most $d^k$) terms of the polynomial mapping $\x \in \R^d$ to $(\x \cdot \vv)^k$ has coefficient at most $1$. This, together with Equation \ref{eq: moment matching} and the triangle inequality, allows us to conclude that
	\[
	\left\lvert
	\ex_{\x \sim \Stest}[(\x \cdot \vv)^k]
	-
	\ex_{\x \sim \Gauss(0,I_d)}[(\x \cdot \vv)^k]
	\right \rvert
	\leq
	d^k
	\frac{\epsilon}{d^{C_2 k}}.
	\]
	Now, since $\vv$ is a unit vector, we have $\ex_{\x \sim \Gauss(0,I_d)}[(\x \cdot \vv)^k] = k!!\leq k^{k}$. Combining this with the equation above, and Equation \ref{eq: degree-k Chabychev} and then substituting Equation \ref{eq: tau is large} and the values of $k$ and $T$ we get:
	\begin{multline*}
		\pr_{\x \sim \Stest}[\sign(\x \cdot \vv -\tau)\neq -1]
		\leq
		\frac{1}{|\tau|^k} \left( 
		k^{k/2}
		+
		d^k
		\frac{\epsilon}{d^{C_2 k}}
		\right)
		\leq\\
		\frac{1}{\left(\frac{1}{2} C_1^2 \log \frac{1}{\epsilon}\right)^{C_1 \log \frac{1}{\epsilon}}} \left( 
		\left( C_1 \log \frac{1}{\epsilon}
		\right)^{C_1\log \frac{1}{\epsilon}}
		+
		d^k
		\frac{\epsilon}{d^{C_2 k}}
		\right)
	\end{multline*}
	We see that when $C_1$ and $C_2$ both exceed some absolute constant, the above expression is at most $\epsilon$, which completes the proof.
\end{proof}

Having proven \Cref{prop: if keep seeing the same label moment matching guarantees generalization}, we are now ready to prove \Cref{prop: if accepts then generalizes}.
\begin{proof}[Proof of \Cref{prop: if accepts then generalizes}.]
	
	First, suppose the algorithm outputs (ACCEPT, $L$) for some $L\in \{\pm 1\}$ via Step \ref{line: used to be 4c}. For the algorithm to reach this step, it has to be that 
	\[
	\pr_{\x \in S}[\coptcommon(\x) \neq L] \leq \frac{1}{T},
	\]
	Via Hoeffding's inequality, if $C$ is sufficiently large then with probability at least $1-\frac{1}{1000}$ it holds that 
 \begin{equation}
 \label{eq: label estimates correct}
 |\pr_{\x \in S}[\coptcommon(\x) \neq L]-\pr_{\x \in S}[\coptcommon(\x) \neq L]|\leq \frac{1}{2T},
 \end{equation} and combining the two equations above
	\[
	\pr_{\x \in \Gauss(0,I_d)}[\coptcommon(\x)\neq L] \leq \frac{2}{T}.
	\]
	Furthermore, for the algorithm not output REJECT in Step \ref{line: used to be 4b}, it has to be the case that for every collection of non-negative integers $(\mindex_1, \cdots, \mindex_d)$ satisfying $\sum_i \mindex_i \leq k$ we have 
	\[
	\left\lvert
	\ex_{\x \sim \Stest} 
	\left[\x^\alpha 
	\right]
	-
	\ex_{\x \sim \Gauss(0,I_d)} 
	\left[\x^\alpha 
	\right]
	\right\rvert
	>\frac{\epsilon}{d^{C_2 k}}.
	\]
	Overall, this allows us to apply \Cref{prop: if keep seeing the same label moment matching guarantees generalization} to conclude that
	\[
	\pr_{\x \sim \Stest}[\coptcommon(\x)\neq L]\leq O(\epsilon),
	\]
	and, for a sufficiently large absolute constant $C$, with probability at least $1-\frac{1}{1000}$, this is only possible if
	\[
	\pr_{\x \sim \Dtestmarginal}[\coptcommon(\x)\neq L]\leq O(\epsilon),
	\]
	which finishes the proof for the case when the algorithm accepts in Step \ref{line: used to be 4c}.
	
	Now, suppose the algorithm accepts in Step \ref{line: used to be 5d}. For the algorithm to reach this step, it has to be that 
	\[
	\pr_{\x \in S}[\coptcommon(\x) \neq L] > \frac{1}{T},
	\]
	And together with Equation \ref{eq: label estimates correct}, this implies that
	\[
	\pr_{\x \in \Gauss(0,I_d)}[\coptcommon(\x)\neq L] > \frac{1}{2T}.
	\]
	For such $f^*$ we can apply \Cref{prop: parameter recovery} and conclude that with probability at least $1-1/1000$ the values of $\hat{\vv}$ and $\tau$ obtained in \Cref{algorithm:improved for general halfspaces} satisfy  
	\begin{equation}
		\label{eq: v hat is good}
		\norm{\vv-\hat{\vv}}_2
		\leq 
		\left(\frac{\epsilon}{C_3 d^{C_3}}
		\right)^2
		=\beta,
	\end{equation}
	\begin{equation}
		\label{eq: tau hat is good}
		|\tau-\hat{\tau}|
		\leq \left(\frac{\epsilon}{C_3 d^{C_3}}\right)^2 = \beta
		,
	\end{equation}
	where the last equality is by the definition of $\beta$. Now, since the algorithm did not reject in Step \ref{line: used to be 5c}, it must be the case that the fraction of elements in $\Stest$ that satisfy $\sign(\vv_\x^+\cdot \x-\tau^+_\x)\neq \sign(\vv_\x^-\cdot \x-\tau^-_\x)$ is at most $10\epsilon$. If $C$ is a sufficiently large absolute constant, the standard Hoeffding inequality tells us that for this to happen with probability larger than $1/1000$ it has to be the case that
	\[
	\pr_{\x \sim \Dtestmarginal} \left[ \sign\left(
	\max_{
		(\vv',\tau')\in \paramneighborhood}
	\vv'\cdot \x -\tau'
	\right)
	\neq \sign\left(
	\min_{
		(\vv',\tau')\in \paramneighborhood}
	\vv'\cdot \x -\tau'
	\right) \right]\leq 11 \epsilon.
	\]
	Whenever the event above occurs, since $\paramneighborhood = \{(\vv',\tau'): \|\vv'-\hat\vv\|_2 \le \beta, |\tau'-\hat\tau|\le \beta\}$ we can use Equations \ref{eq: v hat is good} and \ref{eq: tau hat is good} to conclude $\sign(\vv \cdot \x-\tau)=\sign(\hat{\vv} \cdot \x-\hat{\tau})$. Therefore,
	\[
	\pr_{\x \sim \Dtestmarginal} \left[ 
	\sign(\vv \cdot \x-\tau)\neq \sign(\hat{\vv} \cdot \x-\hat{\tau})
	\right] \leq 11\epsilon
	\]
 This completes the proof of soundness of \Cref{algorithm:improved for general halfspaces}.
\end{proof}

\subsubsection{Proof of Completeness.}
\label{sec: proof of Completeness}
The second proposition shows that if the testing distribution is the standard Gaussian, then the algorithm will likely accept. Together, propositions \ref{prop: if accepts then generalizes} and \ref{prop: if Gaussian then accept} yield Theorem \ref{thm: improved TDS learning for general halfspaces}.
\begin{proposition}[Completeness]
	\label{prop: if Gaussian then accept}
	For sufficiently large value of the absolute constants $C$ and $C_3$ and for any halfspace $f=\sign(\hat{\vv} \cdot \x -\hat{\tau})$, suppose the testing distribution $\Dtestmarginal$ is the standard Gaussian distribution. Then, with probability at least $1-\frac{1}{100}$ \Cref{algorithm:improved for general halfspaces} will accept, i.e. output (ACCEPT, $\hat{f}$) for some $\hat{f}$.
\end{proposition}
To prove this proposition, we first need to prove \Cref{prop: gaussian is unlikely to fall into disagreement region}.
\begin{proof}[Proof of \Cref{prop: gaussian is unlikely to fall into disagreement region}]
	We have $\ex_{\x \sim \Gauss(0,I_d)}[\norm{\x}_2^2]=d$. Therefore, by Markov's inequality, we have
	\begin{equation}
		\label{eq: norm of x is small}
		\pr_{\x \sim \Gauss(0,I_d)}
		\left[\norm{\x}_2
		>
		\frac{\sqrt{d}}{\sqrt{\beta}}
		\right]
		=
		\pr_{\x \sim \Gauss(0,I_d)}
		\left[\norm{\x}_2^2
		>
		\frac{d}{\beta}
		\right]
		\leq
		\beta
	\end{equation}
	Additionally, from the bound of $\frac{1}{\sqrt{2\pi}}$ on the density of standard Gaussian in one dimension, we get:
	\begin{equation}
		\label{eq: halfspace is satisfied with a margin}
		\pr_{\x \sim \Gauss(0,I_d)}
		\left[
		|\hat{\vv} \cdot \x - \hat{\tau}|
		\le
		100\sqrt{\beta d} +\beta
		\right]
		\le
		\frac{200\sqrt{\beta d} +2 \beta}{\sqrt{2 \pi}}
	\end{equation}
	If it holds that $\norm{\x}_2
	\leq
	\frac{\sqrt{d}}{\sqrt{\beta}}$, we have
	for every $\vv'$ satisfying $\norm{\vv'-\hat{\vv}}_2\leq \beta$ and any
	$\tau'$ satisfying $|\tau'-\hat{\tau}|\leq \beta$ that
	\[
	|
	\vv'\cdot \x -\tau' -(\hat{\vv} \cdot \x - \hat{\tau})
	|
	\leq
	\sqrt{d \beta}
	+
	\beta
	\]
	Therefore, if it is also the case that $|\hat{\vv} \cdot \x - \hat{\tau}|
	>
	100\sqrt{\beta d} +\beta$, then we have 
	\[
	\sign\left(
	\vv'\cdot \x -\tau'
	\right)
	=
	\sign\left(
	\hat{\vv}\cdot \x -\hat{\tau}
	\right)
	\]
	This allows us to conclude that 
	\begin{multline*}
		\pr_{\x \sim \Gauss(0,I_d)}
		\left[
		\sign\left(
		\max_{
			(\vv',\tau')\in \paramneighborhood}
		\vv'\cdot \x -\tau'
		\right)
		\neq \sign\left(
		\min_{
			(\vv',\tau')\in \paramneighborhood}
		\vv'\cdot \x -\tau'
		\right)
		\right]
		\leq\\
		\pr_{\x \sim \Gauss(0,I_d)}
		\left[\norm{\x}_2
		>
		\frac{\sqrt{d}}{\sqrt{\beta}}
		\right]
		+
		\pr_{\x \sim \Gauss(0,I_d)}
		\left[
		|\hat{\vv} \cdot \x - \hat{\tau}|
		\le
		100\sqrt{\beta d} +\beta
		\right]
		\leq
		\beta
		+\frac{200\sqrt{\beta d} +2 \beta}{\sqrt{2 \pi}},
	\end{multline*}
	where in the end we substituted Equation \ref{eq: norm of x is small} and Equation \ref{eq: halfspace is satisfied with a margin}. Recalling that for $\beta \in (0,1)$ we have $\beta<\sqrt{\beta}$ and picking $K_1$ to be a sufficiently large absolute constant, our proposition follows from the inequality above.
\end{proof}

Having proven \Cref{prop: gaussian is unlikely to fall into disagreement region}, we are now ready to prove \Cref{prop: if Gaussian then accept}.

\begin{proof}[Proof of \Cref{prop: if Gaussian then accept}.]
	There are two ways for the algorithm to output REJECT: through Step \ref{line: used to be 4b} and through Step \ref{line: used to be 5c}. We will argue neither takes place. From standard Gaussian concentration, if $C$ is a sufficiently large absolute const ant, with probability at least $1-\frac{1}{1000}$ the algorithm will not output REJECT in Step \ref{line: used to be 4b}.
	
	We now proceed to ruling out the possibility that the algorithm outputs REJECT in Step \ref{line: used to be 5c}. For the algorithm to reach step Step \ref{line: used to be 5c}, it is necessary that     
	\[
	\pr_{\x \in S}[\coptcommon(\x) \neq L] > \frac{1}{T},
	\]
 Via Hoeffding's inequality, if $C$ is sufficiently large then with probability at least $1-\frac{1}{1000}$ it holds that 
 $
 |\pr_{\x \in S}[\coptcommon(\x) \neq L]-\pr_{\x \in S}[\coptcommon(\x) \neq L]|\leq \frac{1}{2T}$, which together with the equation above implies that
	\[
	\pr_{\x \in \Gauss(0,I_d)}[\coptcommon(\x)\neq L] > \frac{1}{2T}.
	\]
	
	For such $f^*$ we can apply \Cref{prop: parameter recovery} and conclude that with probability at least $1-1/1000$ the values of $\hat{\vv}$ and $\tau$ obtained in \Cref{algorithm:improved for general halfspaces} satisfy  
	\begin{equation}
		\label{eq: v hat is good new}
		\norm{\vv-\hat{\vv}}_2
		\leq \left(\frac{\epsilon}{C_3 d^{C_3}}\right)^2
		=\beta
		,
	\end{equation}
	\begin{equation}
		\label{eq: tau hat is good new}
		|\tau-\hat{\tau}|
		\leq \left(\frac{\epsilon}{C_3 d^{C_3}} \right)^2=\beta
		,
	\end{equation}
 Recall that $\paramneighborhood = \{(\vv',\tau'): \|\vv'-\hat\vv\|_2 \le \beta, |\tau'-\hat\tau|\le \beta\}$.
	The equation above together with \Cref{prop: gaussian is unlikely to fall into disagreement region} implies that
	\[
	\pr_{\x \sim \Gauss(0,I_d)}
	\left[
	\sign\left(
	\max_{
		(\vv',\tau')\in \paramneighborhood}
	\vv'\cdot \x -\tau'
	\right)
	\neq \sign\left(
	\min_{
		(\vv',\tau')\in \paramneighborhood}
	\vv'\cdot \x -\tau'
	\right)
	\right]
	\leq
	K_1  d^{K_1} \frac{\epsilon}{C_3 d^{C_3}} \leq \epsilon,
	\]
	where the last inequality holds for sufficiently large value of $C_3$. Combining the inequality above with the standard Hoeffding bound and recalling that $\Dtestmarginal=\Gauss(0,I_d)$, we see that with probability at least $1-\frac{1}{1000}$,
	\[
	\pr_{\x \sim \Stest}
	\left[
	\sign\left(
	\max_{
		(\vv',\tau')\in \paramneighborhood}
	\vv'\cdot \x -\tau'
	\right)
	\neq \sign\left(
	\min_{
		(\vv',\tau')\in \paramneighborhood}
	\vv'\cdot \x -\tau'
	\right)
	\right]
	\leq 2\epsilon,
	\]
	In conclusion, we see that the inequality above implies that the algorithm does not output REJECT in Step \ref{line: used to be 5c}. This completes our proof.
\end{proof}
\subsubsection{Parameter recovery.}
	\label{sec: parameter recovery}
Here we prove \Cref{prop: parameter recovery}, which was used in the proofs of \Cref{prop: if accepts then generalizes} and \Cref{prop: if Gaussian then accept}, thereby finishing the proof of \Cref{thm: improved TDS learning for general halfspaces}. Let us first recall the setting of \Cref{prop: parameter recovery}. For a unit vector $\vv$ in $\R^d$ and $\tau \in \R$ satisfying
\[
\min\left(\pr_{x \in \Gauss(0,I_d)}[\vv \cdot \x -\tau > 0], \pr_{x \in \Gauss(0,I_d)}[\vv \cdot \x -\tau < 0]\right)\geq \eta,
\]
$\Strain$ is a set of $C\left(\frac{d}{\eta \beta }\right)^C$ i.i.d samples from a distribution $\Dtrainjoint$ with $\X$-marginal distributed as standard Gaussian and $\Y$-marginal given by the halfspace $f=\sign(\vv \cdot \x - \tau)$. The absolute constant $C$ is assumed to be sufficiently large. We let $\T = \{\hat{\vv}\cdot \x: (\x,y)\in \Strain\}$ and set \[\hat{\vv}  = \frac{\sum_{(\x,y) \in \Strain} \x y}{\norm{\sum_{(\x,y) \in \Strain} \x y}_2}\]
	\[
	\hat{\tau}
	=
	\argmin_{\tau' \in \T}
	\pr_{(\x,y) \in \Strain}[\coptcommon(\x)\neq \sign(\hat{\vv} \cdot \x - \tau')].
	\]
	We would like to prove that with probability at least $29/30$ we have 
	\[
	\norm{\vv-\hat{\vv}}_2
	\leq \beta,
	\]
	\[
	|\tau-\hat{\tau}|
	\leq \beta.
	\]

The following proposition tells us that the first inequality above is likely to hold:
\begin{proposition}
	[Recovery of normal vector for halfspaces]
	\label{prop: recover normal vector}
	For a sufficiently large absolute constant $C$, and every $\eta, \beta \in (0,1)$ and integer $d$, the following holds. Let
	$\Strain$ is a set of at least $C\left(\frac{d}{\eta \beta }\right)^C$ i.i.d samples from a distribution $\Dtrainjoint$ with $\X$-marginal distributed as standard Gaussian and $\Y$-marginal given by the halfspace $f=\sign(\vv \cdot \x - \tau)$.
 For every unit vector $\vv$ in $\R^d$ and $\tau \in \R$ satisfying
	\[
	\min\left(\pr_{\x \in \Gauss(0,I_d)}[\vv \cdot \x -\tau > 0], \pr_{\x \in \Gauss(0,I_d)}[\vv \cdot \x -\tau < 0]\right)\geq \eta,
	\]
	The vector $\hat{\vv}  = \frac{\sum_{(\x,y) \in \Strain} \x y}{\norm{\sum_{(\x,y) \in \Strain} \x y}_2}$ with probability at least $1-\frac{1}{2000}$ satisfies: 
	\[
	\norm{\vv-\hat{\vv}}_2
	\leq \beta,
	\]
\end{proposition}

Once this stage is accomplished, the next proposition tells us that we can recover the offset $\tau$. 
\begin{proposition}
	[Offset recovery for halfspaces]
	\label{prop: offset recovery}
	For a sufficiently large absolute constant $C$, and every $\eta, \gamma \in (0,1)$ and integer $d$, the following holds. Let
	$\Strain$ is a set of at least $C\left(\frac{d}{\eta \gamma }\right)^C$ i.i.d samples from a distribution $\Dtrainjoint$ with $\X$-marginal distributed as standard Gaussian and $\Y$-marginal given by the halfspace $f=\sign(v \cdot \x - \tau)$.
 For every unit vector $\vv$ in $\R^d$ and $\tau \in \R$ satisfying
	\[
	\min\left(\pr_{\x \in \Gauss(0,I_d)}[\vv \cdot \x -\tau > 0], \pr_{\x \in \Gauss(0,I_d)}[\vv \cdot \x -\tau < 0]\right)\geq \eta,
	\]
	Then,  with probability at least $1-\frac{1}{2000}$, for every unit vector $\hat{\vv}$ that forms an angle of at most $\gamma$ with $\vv$  the value \[
	\hat{\tau}
	=
	\argmin_{\tau' \in \R}
	\pr_{(\x,y) \in \Strain}[\coptcommon(\x)\neq \sign(\hat{\vv} \cdot \x - \tau')].
	\] satisfies 
	\[
	|\tau-\hat{\tau}|
	\leq O\left(\frac{1}{\eta^{50}} \sqrt{\gamma}\right).
	\]
\end{proposition}
Formally, \Cref{prop: parameter recovery} follows from the two propositions above as follows. One first uses \Cref{prop: recover normal vector} to conclude that, for any absolute constant $C_5$, there is a value of the absolute constant $C$ for which with probability $1-\frac{1}{2000}$ a vector $\hat{\vv}$ that satisfies $\norm{\vv - \hat{\vv}}\leq \frac{1}{C_5} \beta^2 \eta^{100}$. This implies that the angle between $\vv$ and $\hat{\vv}$ is upper-bounded by $\frac{10}{C_5} \beta^2 \eta^{100}$. Then, if the absolute constant $C_5$ is large enough, if we use \Cref{prop: offset recovery}, then  with probability $1-\frac{1}{2000}$ the value $\hat{\tau}$ satisfies $|\tau-\hat{\tau}|\leq \beta$, finishing the proof of \Cref{prop: parameter recovery}.

Now, proceed to prove the two propositions above. We start with  \Cref{prop: recover normal vector}.
\begin{proof}[Proof of \Cref{prop: recover normal vector}]
	
	Let $\{\e_1, \cdots \e_{d-1}\}$ form an orthonormal basis for the subspace orthogonal to $\vv$. Since all the projections $\{\vv\cdot \x, \e_1 \cdot \x, \cdots, \e_{d-1} \cdot \x\}$ are independent standard Gaussians and $\coptcommon(\x)=\sign(\vv \cdot \x -\tau)$ we have for all $i$ 
	\[
	\ex_{\x \in \Gauss(0,I_d)} [\e_i \cdot \x \coptcommon(\x)] =0. 
	\]
	At the same time
	\begin{multline*}
		\ex_{\x \in \Gauss(0,I_d)} [\vv \cdot \x f(\x)] = 
		\int_{t=-\infty}^{+\infty}
		t \sign(t-\tau) \frac{1}{\sqrt{2 \pi}} ~dt
		=\\\int_{t\in[-|\tau|,|\tau|]}
		t \sign(t-\tau) \frac{1}{\sqrt{2 \pi}} ~dt
		+\int_{t\in[-\infty,-|\tau|]\cup [|\tau|, +\infty]}
		t \sign(t-\tau) \frac{1}{\sqrt{2 \pi}} ~dt
		=\frac{2}{\sqrt{2 \pi}}\int_{t=|\tau|}^{\infty}
		t  ~dt
	\end{multline*}
	For some positive absolute constant $K_2$,
	the final expression above is at least $K_2 \pr_{t \sim N(0,1)}[t>\tau]$, because if $|\tau|>1$, then one can lower-bound the expression above by $\frac{2}{\sqrt{2 \pi}}\int_{t=|\tau|}^{\infty} ~dt$. On the other hand, if $|\tau|\in [0,1]$, then the expression on the right side is at least $\frac{2}{\sqrt{2 \pi}}\int_{t=1}^{\infty} ~dt$ which is a positive absolute constant, while $\pr_{t \sim N(0,1)}[t>\tau]$ is always upper-bounded by $1$. Overall, we have 
	\begin{align*}
	\ex_{\x \in \Gauss(0,I_d)} [\vv \cdot \x \coptcommon(\x)] &\geq
	K_2  \pr_{t \sim N(0,1)}[t>\tau] \\
	&=
	K_2
	\min\left(\pr_{\x \in \Gauss(0,I_d)}[\vv \cdot \x -\tau > 0], \pr_{\x \in \Gauss(0,I_d)}[\vv \cdot \x -\tau < 0]\right)\\
	&\geq
	K_2 \eta.
	\end{align*}
	Now, we bound the variance of $\x \coptcommon(\x)$. Since $\coptcommon(\x)\in \{\pm 1\}$, we have 
	\[
	\ex_{\x\in \Gauss(0,I_d)}
	\left[( \e_i\cdot \x\coptcommon(\x) )^2\right]
	=
	\ex_{\x\in \Gauss(0,I_d)}
	\left[(\e_i\cdot \x )^2\right]
	=1,
	\]
	\[
	\ex_{\x\in \Gauss(0,I_d)}
	\left[(  \vv\cdot \x \coptcommon(\x))^2\right]
	=\ex_{\x\in \Gauss(0,I_d)}
	\left[(  \vv\cdot \x)^2\right]
	=1.
	\]
	This allows us to use the Chebychev's inequality together with the union bound to conclude that with probability at least $1-\frac{1}{2000}$ we have for all $i$
	\[
	\left\lvert E_{\x \in S}
	[\e_i \cdot \x \coptcommon(\x) ]
	\right \rvert
	\leq \sqrt{\frac{60d}{N}},
	\]
	and also 
	\[
	E_{\x \in S}
	[\vv \cdot \x \coptcommon(\x) ]
	\geq K_2 \eta-\sqrt{\frac{60d}{N}},
	\]
Recalling that $\hat{\vv}=\frac{\sum_{\x \in S_1} \x \coptcommon(\x)}{\norm{\sum_{\x \in S_1} \x \coptcommon(\x)}_2}=\frac{\ex_{\x \in S_1} \x \coptcommon(\x)}{\norm{\ex_{\x \in S_1} \x \coptcommon(\x)}_2}$, we see that
\[
|\hat{\vv} \cdot \e_i|
\leq \frac{\sqrt{\frac{60d}{N}}}{K_2\eta-\sqrt{\frac{60d}{N}}}
\]
Substituting $N=C (\frac{d}{\eta \beta})^{C} $, and letting $C$ be a sufficiently large absolute constant, we obtain from above implies that $|\hat{\vv} \cdot \e_i|\leq \frac{\beta}{10\sqrt{d}}$. Since $\norm{\hat{\vv}}=1$ we have 
\[
1\geq|\hat{\vv} \cdot \vv|\geq
\sqrt{1-\frac{\beta}{10}}
\geq 1-\frac{\beta}{10},
\]
we also see that taking  $C$ to be a sufficiently large absolute constant also ensures that $\hat{\vv} \cdot \vv>0$, so overall we get
\[
\norm{\hat{\vv}-\vv}
\leq \beta,
\]
which finishes the proof.
\end{proof}



In order to prove \Cref{prop: offset recovery}, we will need  a proposition that relates the following two quantities: (1) the difference in offsets $\tau_1$ and $\tau_2$ of two halfspaces (2) The probability that these two hafspaces disagree on a point drawn from the standard Gaussian. 
\begin{proposition}
	\label{prop: error in offset yields error in halfspaces}
	There is some absolute constant $K_1$ such that for any pair of unit vectors $\vv_1, \vv_2 \in \R^d$ and a pair of real numbers $\tau_1, \tau_2$, letting $\gamma$ denote the angle between $\vv_1$ and $\vv_2$, the following holds. Suppose $\gamma< \pi/4$, then
	\begin{equation}\label{eq: wrong offset bad}
		\pr_{\x \in \Gauss(0,I_d)}\left[\sign\left(\vv_1 \cdot \x - \tau_1\right)
		\neq \sign\left(
		\vv_2 \cdot \x - \tau_2
		\right)
		\right]
		\geq 
		\frac{1}{K_1}
		e^{-
			\tau_1^2/2}
		\min\left(
		\left\lvert
		\tau_1 - \frac{\tau_2}{\cos \gamma}
		\right\rvert
		, \frac{1}{|\tau_1|+1}\right)
	\end{equation}
	It is also the case that
	\begin{equation}
		\label{eq: right offset good}
		\pr_{\x \in \Gauss(0,I_d)}\left[\sign\left(\vv_1 \cdot \x - \tau_1\right)
		\neq \sign\left(
		\vv_2 \cdot \x - \tau_1 \cos \gamma
		\right)
		\right]
		\leq
		K_1 \sqrt{\gamma}
	\end{equation}
\end{proposition}
\begin{proof}
	To prove this, we first show that for any $z \in \R$, conditioned on $\vv_1 \cdot \x=z_1$ the distribution of $\vv_2 \cdot \x$ is $\mathcal{N}(z_1 \cos \gamma, \sin \gamma)$. Indeed, let $\vv_3$ be the unit vector that one obtains by first projecting $\vv_2$ into the subspace perpendicular to $\vv_1$, and then normalizing the resulting vector to have unit norm. This means $\vv_3$ is orthogonal to $\vv_1$ and we have 
	\[
	\vv_2=\vv_1 \cos \gamma + \vv_3 \sin \gamma.
	\]
	Therefore
	\[
	\x \cdot \vv_2 = \x \cdot \vv_1 \cos \gamma + \x \cdot \vv_3 \sin \gamma
	\]
	Now, since $\x \cdot \vv_1$ and  $\x \cdot \vv_3$ are distributed as i.i.d. one-dimensional standard Gaussians. Thus, conditioning on $\x \cdot \vv_1=z_1$ we get that $ \x \cdot \vv_2 $ is distributed as $\mathcal{N}(z \cos \gamma, \sin \gamma)$. 
	
	Our observation allows us to write:
	\begin{multline}\label{eq: first step to bound disagreement for halfspaces with differing thresholds}
		\pr_{\x \in \Gauss(0,I_d)}\left[\sign\left(\vv_1 \cdot \x - \tau_1\right)
		\neq \sign\left(
		\vv_2 \cdot \x - \tau_2
		\right)
		\right]
		=\\
		\pr_{z_1, z_2 \in \mathcal{N}(0, 1)}\left[\sign\left(z_1 - \tau_1\right)
		\neq \sign\left(
		z_1 \cos \gamma + z_2 \sin \gamma - \tau_2
		\right)
		\right]
		=\\
		\pr_{z_1, z_2 \in \mathcal{N}(0, 1)}\left[\sign\left(z_1 - \tau_1\right)
		\neq \sign\left(
		z_1   + z_2 \tan \gamma - \tau_2/ \cos \gamma
		\right)
		\right]
\end{multline}

Let us first focus on the case when $\gamma\in [0, \pi/2)$. We see that
\begin{multline}
	\label{eq: second step to bound disagreement for halfspaces with differing thresholds}
	\pr_{z_1, z_2 \in \mathcal{N}(0, 1)}\left[\sign\left(z_1 - \tau_1\right)
	\neq \sign\left(
	z_1   + z_2 \tan \gamma - \tau_2/ \cos \gamma
	\right)
	\right]
	\geq\\
	\frac{1}{2}
	\pr_{z_1\in \mathcal{N}(0, 1)}\left[\sign\left(z_1 - \tau_1\right)
	\neq \sign\left(
	z_1    -\tau_2/ \cos \gamma
	\right)
	\right]
\end{multline}
The reason that inequality above is true is that, conditioned on a specific value of $z_1$, if $z_1> \tau_2/ \cos \gamma $, then $z_1   + z_2 \tan \gamma - \tau_2$ is more likely to be positive than negative. At the same time, if  $z_1< \tau_2/ \cos \gamma $, then $z_1   + z_2 \tan \gamma - \tau_2$ is more likely to be negative than positive. 

We lower-bound the probability above as follows. Let $A$ be the interval of $\R$ defined as follows:
\[
A:=
\left\{
z\in \R: ~~
\sign(z-\tau_1) \neq \sign(z-\tau_2/\cos \gamma)
\And
|z-\tau_1| \leq \frac{1}{|\tau_1|+1}  
\right\}
\]
We have
\begin{align}
	\pr_{z_1\in \mathcal{N}(0, 1)}&\left[\sign\left(z_1 - \tau_1\right)
	\neq \sign\left(
	z_1    -\tau_2/ \cos \gamma
	\right)
	\right]
	\geq
	\pr_{z_1\in \mathcal{N}(0, 1)}\left[
	z_1 \in A
	\right]
	\geq\nonumber
 \\
	&\ge \min\left(
	\left\lvert
	\tau_1 - \frac{\tau_2}{\cos \gamma}
	\right\rvert
	, \frac{1}{|\tau_1|+1}\right)
	\frac{1}{\sqrt{2 \pi}}
	e^{-\frac{1}{2}
		\left(
		|\tau_1|
		-\frac{1}{|\tau_1|+1}
		\right)^2
	}\nonumber
 \\
	&\geq
	\Omega(1) \cdot \min\left(
	\left\lvert
	\tau_1 - \frac{\tau_2}{\cos \gamma}
	\right\rvert
	, \frac{1}{|\tau_1|+1}\right)
	e^{-\tau_1^2/2},
\end{align}
which, combined with Equations \ref{eq: first step to bound disagreement for halfspaces with differing thresholds} and \ref{eq: second step to bound disagreement for halfspaces with differing thresholds}, finishes the proof of Equation \ref{eq: wrong offset bad}.

Now, we proceed to proving Equation \ref{eq: right offset good}.
We proceed as follows:
\begin{align*}
	\pr_{z_1, z_2 \in \mathcal{N}(0, 1)}&\left[\sign\left(z_1 - \tau_1
	\right)
	= \sign\left(
	z_1   + z_2 \tan \gamma - \tau_1
	\right)
	\right] \\
    &\geq
	\pr_{z_1, z_2 \in \mathcal{N}(0, 1)}\left[
	|z_1-\tau_1| > \sqrt{\tan \gamma} 
	\And
	|z_2|<\frac{1}{\sqrt{\tan \gamma}}
	\right]\\
	&\geq
	1-O(1) \cdot \sqrt{\tan \gamma}-
	O(1)
	\int_{\frac{1}{\sqrt{\tan \gamma}}}^{\infty}
	e^{-z^2/2} ~ dz
	\\
    &=
	1-O(\sqrt{\tan \gamma})=1-O(\sqrt{ \gamma}),
\end{align*}
which, when combining with with Equation \ref{eq: first step to bound disagreement for halfspaces with differing thresholds} and substituting $\tau_2=\tau_1 \cos \gamma$, proves Equation \ref{eq: right offset good}.
\end{proof}
Having proven \Cref{prop: error in offset yields error in halfspaces}, we are now ready to prove \Cref{prop: offset recovery}.
\begin{proof}[Proof of \Cref{prop: offset recovery}.]
Recall that $\T = \{\hat{\vv}\cdot \x: (\x,y)\in \Strain\}$. We see for $\tau'$ between two neighboring elements of $\T$ the value of $\pr_{\x \in \N(0,I)}[\coptcommon(\x)\neq \sign(\hat{\vv} \cdot \x - \tau')]$ stays the same. Therefore
\begin{equation}
\label{eq: discretization is ok}
\pr_{\x \in \T}[\coptcommon(\x)\neq \sign(\hat{\vv} \cdot \x - \hat{\tau})]=\min_{\tau' \in \T}
	\pr_{\x \in \T}[\coptcommon(\x)\neq \sign(\hat{\vv} \cdot \x - \tau')]
 =
 \min_{\tau' \in \R}
	\pr_{\x \in \T}[\coptcommon(\x)\neq \sign(\hat{\vv} \cdot \x - \tau')].
\end{equation}
Since the function class $\{\sign(\vv' \cdot \x-\tau':~\vv'\in \R^d,~\tau'\in \R \}$ has a VC dimension of $d+1$, the standard VC bound tells us that for sufficiently large absolute constant $C$ with probability at least $1-\frac{1}{2000}$  we have for every $\tau'\in \R$ and unit vector $\hat{\vv}$ that\begin{equation}
	\label{eq: empirical error is close to pop error for halfspaces different offsets}
	\left\lvert
	\pr_{\x \in \N(0,I)}[\coptcommon(\x)\neq \sign(\hat{\vv} \cdot \x - \tau']
	-
	\pr_{\x \in \T}[\coptcommon(\x)\neq \sign(\hat{\vv} \cdot \x - \tau')]
	\right \rvert
	\leq\sqrt{\gamma}
\end{equation}

From Equation \ref{eq: right offset good} in \Cref{prop: error in offset yields error in halfspaces} we have that
\begin{equation}
	\label{eq: there is some offset that does really good}
	\min_{\tau' \in \R}
	\pr_{\x \in \N(0,I)}[\coptcommon(\x)\neq \sign(\hat{\vv} \cdot \x - \tau')]
	\leq
	K_1 \sqrt{\gamma}
	\leq O( \sqrt{\gamma})
\end{equation}
We now upper-bound $|\tau|$	 in terms $\eta$ as follows:
\begin{equation}
	\label{eq: tau is not too big}
	|\tau|
	\leq 10 \sqrt{\log \frac{1}{\eta}},
\end{equation}
For $|\tau|<1$, this is immediate, because the probability that the Gaussian exceeds one standard deviation in a given direction is at least $1/10$. For $|\tau|\geq 1$, we write 
\[
\eta
\geq
\int_{|\tau|}^{\infty}
e^{-t^2/2}~dt
\geq
\frac{1}{|\tau|}
e^{-(|\tau|+1/|\tau|)^2/2}
\geq\frac{1}{e^2} \cdot \frac{1}{|\tau|} e^{-|\tau|^2/2},
\]
which proves Equation \ref{eq: tau is not too big}. 

Taking Equation \ref{eq: wrong offset bad} in \Cref{prop: error in offset yields error in halfspaces} and substituting Equation \ref{eq: tau is not too big} we get 
\begin{multline*}
	\pr_{x \in \Gauss(0,I_d)}\left[f(x) 
	\neq \sign\left(
	\hat{\vv} \cdot x - \hat{\tau}
	\right)
	\right]
	\geq 
	\frac{1}{K_1}
	e^{-
		\tau_1^2/2}
	\min\left(
	\left\lvert
	\tau - \frac{\hat{\tau}}{\cos \gamma}
	\right\rvert
	, \frac{1}{|\tau|+1}\right)\geq\\
	\frac{\eta^{50}}{K_1 }
	\min\left(\left\lvert
	\tau - \frac{\hat{\tau}}{\cos \gamma}
	\right\rvert, 1\right)
\end{multline*}
Combining the above with Equation \ref{eq: discretization is ok}, Equation \ref{eq: empirical error is close to pop error for halfspaces different offsets} and Equation \ref{eq: there is some offset that does really good} we get 
\[
\left\lvert
\tau - \frac{\hat{\tau}}{\cos \gamma}
\right\rvert
\leq
\frac{K_1}{\eta^{50}}
(O(\sqrt{\gamma})
+
\sqrt{\gamma}
)
\leq 
O(\sqrt{\gamma} /\eta^{50}).
\]
Finally, we see that
\[
|\tau - \hat{\tau}|
\leq
\left\lvert
\tau - \frac{\hat{\tau}}{\cos \gamma}
\right\rvert+
\left\lvert
\hat{\tau} - \frac{\hat{\tau}}{\cos \gamma}
\right\rvert
\leq
O(\sqrt{\gamma} /\eta^{50})
+O(\sqrt{\log ({1}/{\eta})} \gamma^2)
= O(\sqrt{\gamma} /\eta^{50}).
\]

This completes the proof of \Cref{prop: offset recovery}.
\end{proof}

\section{TDS Learning Through Moment Matching}\label{appendix:tds-moment-matching}

\subsection{$\L_2$-Sandwiching Implies TDS Learning}

We now prove \Cref{theorem:l2-sandwiching-implies-tds} which we restate here for convenience.

\begin{theorem}[$\L_2$-sandwiching implies TDS Learning]
    Let $\Dgeneric$ be a distribution over a set $\X\subseteq\R^d$ and let $\C\subseteq\{\X\to \cube{}\}$ be a concept class. Let $\eps,\delta\in(0,1)$, $\eps'=\eps/100$ $\delta'=\delta/2$ and assume that the following are true.
    \begin{enumerate}[label=\textnormal{(}\roman*\textnormal{)}]
        \item\label{item:sandwiching} ($\L_2$-Sandwiching) The $\eps'$-approximate $\L_2$-sandwiching degree of $\C$ under $\Dgeneric$ is at most $\degbound$ with coefficient bound $\pbound$. 
        \item\label{item:concentration} (Moment Concentration) If $\Sunlabelled\sim\Dgeneric^{\otimes m}$ and $m\ge \mconc$ then, with probability at least $1-\delta'$, we have that for any $\mindex\in \N^d$ with $\|\mindex\|_1\le k$ it holds $|\E_\Sunlabelled[\x^\mindex] - \E_\Dgeneric[\x^\mindex]| \le \frac{\eps'}{\pbound^2d^{4\degbound}}$.
        \item\label{item:generalization} (Generalization) 
        If $\Slabelled\sim\Dgenericjoint^{\otimes m}$ where $\Dgenericjoint$ is any distribution over $\X\times\cube{}$ such that $\Dgenericmarginal = \Dgeneric$ and $m\ge \mgen$ then, with probability at least $1-\delta'$ we have that for any degree-$k$ polynomial $p$ with coefficient bound $\pbound$ it holds $|\E_{\Dgenericjoint}[(y-p(\x))^2]- \E_\Slabelled[(y-p(\x))^2]|\le {\eps'}$.
    \end{enumerate}

    Then, \Cref{algorithm:l2-sandwiching}, upon receiving $\mtrain \ge  \mgen$ labelled samples $\Strain$ from the training distribution and $\mtest\ge C\cdot \frac{d^\degbound+\log(1/\delta)}{\eps^2} + \mconc$ unlabelled samples $\Stest$ from the test distribution (where $C>0$ is a sufficiently large universal constant), runs in time $\poly(|\Strain|,|\Stest|, d^{\degbound})$ and TDS learns $\C$ with respect to $\Dgeneric$ up to error $32\optcommon+\eps$ and with failure probability $\delta$.
\end{theorem}

\begin{algorithm}
\caption{TDS Learning through Moment Matching}\label{algorithm:l2-sandwiching}
\KwIn{Sets $\Strain$ from $\Dtrainjoint$, $\Stest$ from $\Dtestmarginal$, parameters $\eps>0, \delta\in(0,1)$, $\degbound\in\N, \pbound > 0$}
Set $\eps' = \eps/100$, $\delta' = \delta/2$ and $\mslack=\frac{\eps'}{\pbound^2d^{4\degbound}}$ \\
For each $\mindex\in \N^d$ with $\|\mindex\|_1 \le 2\degbound$, compute the quantity 
$\momentempirical_\mindex = \E_{\x\sim\Stest} [\x^\mindex] = \E_{\x\sim\Stest} \Bigr[\prod_{i\in[d]}\x_i^{\mindex_i}\Bigr]$ \\
\textbf{Reject} and terminate if $|\momentempirical_\mindex-\E_{\x\sim \Dgeneric}[\x^\mindex]|>\mslack$ for some $\mindex$ with $\|\mindex\|_1 \le 2\degbound$. \\
\textbf{Otherwise}, solve the following least squares problem on $\Strain$ up to error $\eps'$
\begin{align*}
    \min_p &\E_{(\x,y)\sim\Strain}\left[(y-p(\x))^2\right] \\
    \text{s.t. } &p \text{ is a polynomial with degree at most }\degbound \\
    &\text{each coefficient of }p \text{ is absolutely bounded by }\pbound
\end{align*}
\\
Let $\hat{p}$ be an $\eps'$-approximate solution to the above optimization problem. \\
\textbf{Accept} and output $h:\X\to\cube{}$ where $h: \x \mapsto \sign(\hat{p}(\x))$.
\end{algorithm}

One key ingredient of the proof of \Cref{theorem:l2-sandwiching-implies-tds} is the following transfer lemma which states that moment matching implies that the empirical squared loss between two polynomials on the test distribution is close to their expected squared loss under the target distribution.

\begin{lemma}[Transfer Lemma for Square Loss]\label{lemma:transfer-lemma-formal}
    Let $\Dgeneric$ be a distribution over $\X\subseteq\R^d$ and $\Stest$ a (multi)set of points in $\R^d$. If $|\E_{\x\sim\Stest}[\x^\mindex]- \E_{\x\sim\Dgeneric}[\x^\mindex]|\le \mslack$ for all $\mindex\in\N^d$ with $\|\mindex\|_1\le 2\degbound$, then for any degree $\degbound$ polynomials $p_1,p_2$ with coefficients that are absolutely bounded by $\pbound$, it holds
    \[
        \Bigr|\E_{\x\sim \Stest}[(p_1(\x)-p_2(\x))^2] - \E_{\x\sim\Dgeneric}[(p_1(\x)-p_2(\x))^2]\Bigr| \le B^2\cdot d^{4\degbound}\cdot \mslack
    \]
    \end{lemma}
    \begin{proof}
        The polynomials $p_1,p_2$ all have degree at most $\degbound$ and coefficients that are absolutely bounded by $\pbound$. Therefore, the polynomial $(p_1-p_2)^2$ has degree at most $2\degbound$ and coefficients that are absolutely bounded by $\pbound^2d^{2\degbound}$. Let $p' = (p_1-p_2)^2 = \sum_{\mindex:\|\mindex\|_1\le 2\degbound} p'_\mindex \x^\mindex$ (with $|p'_\mindex| \le \pbound^2d^{2\degbound}$ as argued above) which gives the following.
    \begin{align*}
        \|p_1-p_2\|_{\L_2(\Stest)}^2 &= \E_{\x\sim\Stest}\left[ (p_1(\x)-p_2(\x))^2\right] = \E_{\x\sim\Stest}\left[ p'(\x)\right]
    \end{align*}
    It remains to relate $\E_{\x\sim\Stest}\left[ p'(\x)\right]$ to $\E_{\x\sim\Dgeneric}\left[ p'(\x)\right]$, which follows by the moment-matching assumption.
    \begin{align*}
        \Bigr| \E_{\x\sim\Stest}\left[ p'(\x)\right] - \E_{\x\sim\Dgeneric}\left[ p'(\x)\right] \Bigr| &= \biggr| \sum_{\mindex:\|\mindex\|_1\le 2\degbound}p'_\mindex\left(\E_{\x\sim\Stest}\left[ \x^\mindex\right] - \E_{\x\sim\Dgeneric}\left[ \x^\mindex\right]\right) \biggr| \\
        &\le \sum_{\mindex:\|\mindex\|_1\le 2\degbound}|p'_\mindex|\cdot \left| \E_{\x\sim\Stest}\left[ \x^\mindex\right] - \E_{\x\sim\Dgeneric}\left[ \x^\mindex\right]\right| \\
        &=\sum_{\mindex:\|\mindex\|_1\le 2\degbound}|p'_\mindex|\cdot \left| \momentempirical_\mindex - \moment_\mindex\right| \\
        &\le d^{2\degbound}\cdot \pbound^2\cdot d^{2\degbound}\cdot \mslack\,,
    \end{align*}
    which concludes the proof of the lemma.
    \end{proof}

We are now ready to prove \Cref{theorem:l2-sandwiching-implies-tds}.

\begin{proof}[Proof of \Cref{theorem:l2-sandwiching-implies-tds}]
    For the following, let $\Dtrainjoint$ be the training distribution, $\Dtestjoint$ the test distribution (both over $\X\times \cube{}$) and $\Dtrainmarginal,\Dtestmarginal$ the corresponding marginal distributions over $\X$. We assume that $\Dtrainmarginal = \Dgeneric$. Let $\mtrain = |\Strain|$ and $\mtest = |\Stest|$, $\epsilon' = \epsilon/100$, $\delta'=\delta/2$, $k$, $\pbound$ as defined in condition \ref{item:sandwiching}. We also set $\Delta = \frac{\eps'}{\pbound^2d^{4\degbound}}$ and $\mconc$ as defined in condition \ref{item:concentration}, as well as $\mgen$ as defined in \ref{item:generalization}.

    \paragraph{Soundness.} Suppose that \Cref{algorithm:l2-sandwiching} accepts and outputs $h = \sign(\hat{p})$. For the following, let $\optcommontrain = \error(\coptcommon;\Dtrainjoint)$ and $\optcommontest = \error(\coptcommon;\Dtestjoint)$ (where we have $\optcommon = \optcommontrain+ \optcommontest$).
    We can bound the error of the hypothesis $h$ on $\Dtestjoint$ as follows
    \begin{align*}
        \error(h;\Dtestjoint) &\le \error(\coptcommon;\Dtestjoint) + \error(\coptcommon,h;\Dtestmarginal) \\
        &= \optcommontest + \E[\error(\coptcommon,h;\Stest)] \,,
    \end{align*}
    where the expectation above is over ${\Stest\sim (\Dtestmarginal)^{\otimes \mtest}}$. Denote $\error(h;\Dtestjoint) = \pr_{\Dtestjoint}[y\neq h(\x)]$ and $\error(h_1,h_2;\Dtestmarginal)=\pr_{\Dtestmarginal}[h_1(\x)\neq h_2(\x)]$ and use the fact that for random variables $y_1,y_2,y_3\in\cube{}$, it holds $\pr[y_1\neq y_2] \le \pr[y_1\neq y_3] + \pr[y_2\neq y_3]$. Since $h$ is the sign of a polynomial with degree at most $\degbound = \degbound(\eps')$ (see \Cref{algorithm:l2-sandwiching}) and the class of functions of this form has VC dimension at most $d^\degbound$ (e.g., by viewing it as the class of halfspaces in $d^\degbound$ dimensions) we have that whenever $\mtest \ge C\cdot \frac{d^\degbound + \log(1/\delta')}{\epsilon'^2}$ for some sufficiently large universal constant $C>0$ the following is true with probability at least $1-\delta'$ over the distribution of $\Stest$.
    \[  
        \E[\error(\coptcommon,h;\Stest)] \le \error(\coptcommon,h;\Stest) + \epsilon'
    \]
    Therefore, it is sufficient to bound the quantity $\error(\coptcommon,h;\Stest)$. We now observe the following simple fact.
    \begin{align*}
        \E_{\x\sim\Stest}[(\coptcommon(\x) - \hat{p}(\x))^2] &\ge \pr_{\Stest}[\coptcommon(\x)=1,\hat p(\x) < 0]+\pr_{\Stest}[\coptcommon(\x)=-1,\hat p(\x) \ge 0] \\
        &= \pr_{\Stest}[\coptcommon(\x)\neq \sign{\hat{p}(\x)}] \\
        &= \error(\coptcommon,h;\Stest)
    \end{align*}
    Therefore, we have $\error(\coptcommon,h;\Stest) \le \|\coptcommon-\hat p\|_{\L_2(\Stest)}^2$. Let $\pup,\pdown$ be $\epsilon'$-approximate $\L_2$ sandwiching polynomials for $\coptcommon$ of degree at most $\degbound = k(\eps')$ and with coefficient bound $\pbound = \pbound(\eps')$. The right hand side can be bounded as follows.
    \begin{align*}
        \|\coptcommon-\hat p\|_{\L_2(\Stest)} &\le \|\coptcommon-\pdown\|_{\L_2(\Stest)} + \|\pdown-\hat p\|_{\L_2(\Stest)} \\
        &\le \|\pup-\pdown\|_{\L_2(\Stest)} + \|\pdown-\hat p\|_{\L_2(\Stest)} 
    \end{align*}
    In the last inequality, we used the fact that $\pdown(\x)\le\coptcommon(\x)\le \pup(\x)$ for any $\x\in \X$. We will now compare $\|\pup-\pdown\|_{\L_2(\Stest)}$ to $\|\pup-\pdown\|_{\L_2(\Dgeneric)}$ (and, similarly, $\|\pdown-\hat p\|_{\L_2(\Stest)}$ to $\|\pdown-\hat p\|_{\L_2(\Dgeneric)}$) using the transfer lemma (\Cref{lemma:transfer-lemma-formal}). The polynomials $\pup,\pdown,\hat{p}$ all have degree at most $\degbound$ and coefficients that are absolutely bounded by $\pbound$. 
    Moreover, since \Cref{algorithm:l2-sandwiching} has accepted, we have that for any $\mindex\in\N^d$ with $\|\mindex\|_1 \le 2\degbound$, the following is true
    \begin{equation}\label{equation:moment-matching-guarantee}
        \left|\momentempirical_\mindex - \moment_\mindex \right| \le \mslack\,,
    \end{equation}
    where $\momentempirical = \E_{\x\sim\Stest}[\x^\mindex]$ (recall that $\x^\mindex = \prod_{i\in[d]}\x_i^{\mindex_i}$), $\moment = \E_{\x\sim \Dgeneric}[\x^\mindex]$ and $\mslack = \frac{\eps'}{\pbound^2d^{4\degbound}}$. 
    Therefore, by applying \Cref{lemma:transfer-lemma-formal}, we obtain that $\|\pup-\pdown\|_{\L_2(\Stest)} \le \|\pup-\pdown\|_{\L_2(\Dgeneric)} + \sqrt{\eps'}$ and, similarly, $\|\pdown-\hat p\|_{\L_2(\Stest)} \le \|\pdown-\hat p\|_{\L_2(\Dgeneric)} + \sqrt{\eps'}$. 
    
    We have assumed that $\pup,\pdown$ are $\eps'$-approximate $\L_2$ sandwiching polynomials for $\coptcommon$ and, therefore $\|\pup-\pdown\|_{\L_2(\Dgeneric)} = \sqrt{\|\pup-\pdown\|_{\L_2(\Dgeneric)}^2} \le \sqrt{\eps'}$ (see \Cref{definition:l2-sandwiching}). We bound the quantity $\|\pdown-\hat p\|_{\L_2(\Dgeneric)}$ as follows.
    \begin{align}
        \|\pdown-\hat p\|_{\L_2(\Dgeneric)} &\le \|\pdown-\coptcommon\|_{\L_2(\Dgeneric)} + \|\coptcommon-\hat p\|_{\L_2(\Dgeneric)} \nonumber\\
        &\le \|\pup-\pdown\|_{\L_2(\Dgeneric)} + \|\coptcommon-\hat p\|_{\L_2(\Dgeneric)} \tag{since $\pdown\le \coptcommon \le \pup$} \\
        &\le \sqrt{\eps'} + \|\coptcommon-\hat p\|_{\L_2(\Dgeneric)}\label{equation:pdown-phat-dgeneric}
    \end{align}
    Recall that $\|\coptcommon-\hat p\|_{\L_2(\Dgeneric)}^2 = \E_{\x\sim\Dgeneric}[(\hat p(\x)-\coptcommon(\x))^2]$. By assumption, $\Dtrainmarginal = \Dgeneric$ and therefore $\E_{\x\sim\Dgeneric}[(\hat p(\x)-\coptcommon(\x))^2] = \E_{\x\sim\Dtrainmarginal}[(\hat p(\x)-\coptcommon(\x))^2]$. Moreover, we can view the expectation to be over the joint distribution $(\x,y)\sim\Dtrainjoint$ (coupling of $\x$ and $y$), but the variable $y$ is ignored, i.e.,  $\E_{\x\sim\Dtrainmarginal}[(\hat p(\x)-\coptcommon(\x))^2] = \E_{(\x,y)\sim\Dtrainjoint}[(\hat p(\x)-\coptcommon(\x))^2]$. We can bound the latter term as follows.
    \begin{align*}
        \E_{(\x,y)\sim\Dtrainjoint}[(\hat p(\x)-\coptcommon(\x))^2]^{1/2} &= \E_{(\x,y)\sim\Dtrainjoint}[(\hat p(\x)-y+y-\coptcommon(\x))^2]^{1/2} \\
        &\le {\E_{\Dtrainjoint}[(\hat p(\x)-y)^2]}^{1/2} + {\E_{\Dtrainjoint}[(y-\coptcommon(\x))^2]}^{1/2}
    \end{align*}
    For the term $\E_{(\x,y)\sim\Dtrainjoint}[(\hat p(\x)-y)^2]$, we use condition \ref{item:generalization} to have with probability at least $1-\delta'$, $|\E_{(\x,y)\sim\Dtrainjoint}[(\hat p(\x)-y)^2] - \E_{(\x,y)\sim\Strain}[(\hat p(\x)-y)^2]|\le\eps'$ whenever $\mtrain \ge \mgen$. We now use the fact that $\hat p$ is an $\eps'$-approximate solution to the least squares problem defined in \Cref{algorithm:l2-sandwiching} and have the following bound
    \begin{align*}
        \E_{(\x,y)\sim\Strain}[(\hat p(\x)-y)^2]^{1/2} &\le \E_{(\x,y)\sim\Strain}[(\pdown(\x)-y)^2]^{1/2} + \sqrt{\eps'}
    \end{align*}
    Therefore, due to the generalization condition we have   
    \begin{align*}
        \E_{(\x,y)\sim\Dtrainjoint}[(\hat p(\x)-y)^2]^{1/2} &\le \E_{(\x,y)\sim\Dtrainjoint}[(\pdown(\x)-y)^2]^{1/2} + 3\sqrt{\eps'}  \\
        &\le \|\pdown - \coptcommon\|_{\L_2(\Dtrainmarginal)} + \E_{(\x,y)\sim\Dtrainjoint}[(y-\coptcommon(\x))^2]^{1/2}+ 3\sqrt{\eps'} \\
        &\le \|\pdown - \pup\|_{\L_2(\Dgeneric)} + \E_{(\x,y)\sim\Dtrainjoint}[(y-\coptcommon(\x))^2]^{1/2}+ 3\sqrt{\eps'} \\
        &\le \E_{(\x,y)\sim\Dtrainjoint}[(y-\coptcommon(\x))^2]^{1/2}+ 4\sqrt{\eps'}
    \end{align*}
    Therefore, we have shown that $\|\coptcommon-\hat p\|_{\L_2(\Dgeneric)} \le 4\E_{\Dtrainjoint}[(y-\coptcommon(\x))^2]^{1/2}+2\sqrt{\eps'}$. Note that $\E_{\Dtrainjoint}[(y-\coptcommon(\x))^2] = 4\pr_{\Dtrainjoint}[y\neq \coptcommon(\x)] = 4\optcommontrain$. Therefore, $\|\coptcommon-\hat p\|_{\L_2(\Dgeneric)} \le 4\sqrt{\optcommontrain}+4\sqrt{\eps'}$. By Equation \eqref{equation:pdown-phat-dgeneric}, this implies $\|\pdown-\hat p\|_{\L_2(\Dgeneric)} \le 4\sqrt{\optcommontrain}+5\sqrt{\eps'}$, which in turn implies $\|\pdown-\hat p\|_{\L_2(\Stest)} \le 4\sqrt{\optcommontrain} + 7\sqrt{\eps'}$. We overall obtain the following bound.
    \begin{align*}
        \error(h;\Dtestjoint) &\le \optcommontest + (4\optcommontrain^{1/2}+7\sqrt{\eps'})^2 \\
        &\le \optcommontest + 32 \optcommontrain + 100\eps' \\
        &\le 32 \optcommon + \eps \tag{since $\eps' = \eps/100$ and $\optcommontest\ge 0$}
    \end{align*}
    Note that, in fact, we have also demonstrated that upon acceptance, the following is true.
    \begin{align*}
        \error(\coptcommon,h;\Dtestmarginal) \le 32 \optcommontrain + \eps
    \end{align*}
    The results above holds with probability at least $1-3\delta' = 1-\delta$ (union bound over two bad events).
    
    \paragraph{Completeness.} For completeness, it is sufficient to ensure that $\mtest \ge \mconc$, because then, the probability of acceptance is at least $1-\delta$, due to condition \ref{item:concentration}, as required.
\end{proof}

\subsection{Applications}

In this section, we apply our main result in \Cref{theorem:l2-sandwiching-implies-tds} to obtain a number of TDS learners for important concept classes with respect to Gaussian and Uniform target marginals. In particular, we will use the following corollary, which follows by \Cref{theorem:l2-sandwiching-implies-tds} and some simple properties of the Gaussian and Uniform distributions (see \Cref{lemma:gaussian-properties,lemma:uniform-properties}).

\begin{corollary}
    Let $\Dgeneric$ be either the standard Gaussian in $d$ dimensions or the uniform distribution over the $d$-dimensional hypercube. Let $\C$ be a concept class whose $\eps$-approximate sandwiching degree with respect to $\Dgeneric$ is $\degbound$. Then, there is an algorithm that runs in time $d^{O(\degbound)}$
    and TDS learns $\C$ up to error $32\optcommon +O(\eps)$ and failure probability at most $0.1$.
\end{corollary}

\paragraph{Boolean Classes.} We now bound the $\L_2$ sandwiching degree of bounded size Decision trees and bounded size and depth Boolean Formulas.

\begin{lemma}[$\L_2$ sandwiching degree of Decision Trees]\label{lemma:l2-sandwiching-dts}
    Let $\Dgeneric$ be the uniform distribution over the hypercube $\X = {\cube{}}^d$. For ${\dtsize}\in\N$, let $\C$ be the class of Decision Trees of size $\dtsize$. Then, for any $\eps>0$ the $\L_2$ sandwiching degree of $\C$ is at most $\degbound = O(\log(\dtsize/\eps))$.
\end{lemma}

\begin{proof}
    Let $\concept\in\C$ be a decision tree of size $\dtsize$. Consider the polynomials $\pup,\pdown$ over $\cube{d}$ which correspond to the following truncated decision trees. For $\pup$, we truncate $\concept$ at depth $\degbound$ and substitute the internal nodes at depth $\degbound$ with leaf nodes labelled $1$. Then, $\pup$ corresponds to a sum of polynomials of degree at most $\degbound$, each corresponding to a root-to-leaf path in the truncated decision tree. Clearly, $\pup\ge \concept$ and $\pup$ has degree $\degbound$. We have that $\E_\Dgeneric[(\pup(\x)-\concept(\x))^2]$ is upper bounded by a constant multiple of the probability that $\pup$ takes the value $1$, while $\concept(\x)$ takes the value $-1$, since $\pup$ is itself a Boolean-valued function (it is a decision tree). The probability that this happens is at most $\dtsize\cdot 2^{-k} = O(\eps)$ for $k=O(\log(\dtsize/\eps))$. We obtain $\pdown$ by a symmetric argument. 
\end{proof}

For the following lemma, we make use of an upper bound for the pointwise distance between a Boolean formula and the best approximating low-degree polynomial from \cite{o2003newdegreebounds} (which readily implies the existence of low-degree $\L_2$ sandwiching polynomials).  

\begin{lemma}[$\L_2$ sandwiching degree of Boolean Formulas, modification of Theorem 6 in \cite{o2003newdegreebounds}]\label{lemma:l2-sandwiching-ac0}
    Let $\Dgeneric$ be the uniform distribution over the hypercube $\X = {\cube{}}^d$. For ${\aczsize,\aczdepth}\in\N$, let $\C$ be the class of Boolean formulas of size at most $\aczsize$ and depth at most $\aczdepth$. Then, for any $\eps>0$ the $\L_2$ sandwiching degree of $\C$ is at most $\degbound = (C\log(\aczsize/\eps))^{5\aczdepth/2}\sqrt{{\aczsize}}$, for some sufficiently large universal constant $C>0$.
\end{lemma}

\begin{proof}
    Let $\concept\in \C$ be an formula of size $\aczsize$ and depth $\aczdepth$. We first construct a polynomial $p$ that satisfies $|p(\x)-\concept(\x)|\le \sqrt{\eps}/2$ for any $\x\in \cube{d}$. This corresponds to a slight modification of the proof of Theorem 6 in \cite{o2003newdegreebounds}, where the basis of the inductive construction of $p$ (see Lemma 10 in \cite{o2003newdegreebounds}) is an $O(\sqrt{\eps}/\aczsize^3)$ bound (instead of the original $1/\aczsize^3$ bound) for the (trivial) approximation of a single variable $\x_i$ by itself. The degree of $p$ is indeed upper bounded by $(C\log(\aczsize/\eps))^{5\aczdepth/2}\sqrt{{\aczsize}}$ and we may obtain $\pup,\pdown$ by setting $\pup(\x) = p(\x)+ \sqrt{\eps}/2$ and $\pdown = p(\x)-\sqrt{\eps}/2$. Clearly, $\pdown(\x)\le \concept(\x)\le \pup(\x)$ and $|\pup(\x)-\pdown(\x)| = \sqrt{\eps}$ for all $\x\in\cube{d}$. Therefore $\|\pup-\pdown\|_{\L_2(\Dgeneric)}^2 \le \eps$.
\end{proof}


We obtain the following results for agnostic TDS learning of boolean concept classes.

\begin{corollary}[TDS Learner for Decision Trees]\label{corollary:tds-dtrees}
    Let $\Dgeneric$ be the uniform distribution over the hypercube in $d$ dimensions. Then, there is an algorithm that runs in time $d^{O(\log(\dtsize/\eps))}$ and TDS learns Decision Trees of size $\dtsize$ with respect to $\Unif(\cube{d})$ up to error $32\optcommon+O(\eps)$.
\end{corollary}

\begin{corollary}[TDS Learner for Boolean Formulas]\label{corollary:tds-acz}
    Let $\Dgeneric$ be the uniform distribution over the hypercube in $d$ dimensions and $C>0$ some sufficiently large universal constant. Then, there is an algorithm that runs in time $d^{\sqrt{{\aczsize}}(C\log(\aczsize/\eps))^{5\aczdepth/2}}$ and TDS learns Boolean formulas of size at most $\aczsize$ and depth at most $\aczdepth$ with respect to $\Unif(\cube{d})$ up to error $32\optcommon+O(\eps)$.
\end{corollary}

\paragraph{Intersections and Decision Trees of Halfspaces.} We now provide an upper bound for the $\L_2$-sandwiching degree of Decision Trees of halfspaces, which does not merely follow from a bound on the $\L_\infty$ approximate degree and, in particular, holds under both the Gaussian distribution and the Uniform over the hypercube. The following lemma is based on a powerful result from pseudorandomness literature (Theorem 10.4 from \cite{gopalan2010fooling}) which was originally used to provide a bound for the $\L_1$-sandwiching degree of decision trees of halfspaces, but, as we show, also provides a bound on the $\L_2$-sandwiching degree with careful manipulation. 

\begin{lemma}[$\L_2$-sandwiching degree of Intersections and Decision Trees of Halfspaces]\label{lemma:l2-sandwiching-intersections}
    Let $\Dgeneric$ be either the uniform distribution over the hypercube $\X = {\cube{}}^d$ or the multivariate Gaussian distribution $\Gauss(0,I_d)$ over $\X = \R^d$. For ${\nhalfspaces}\in\N$, let also $\C$ be the class of concepts that can be expressed as an intersection of ${\nhalfspaces}$ halfspaces on $\X$. Then, for any $\eps>0$ the $\L_2$ sandwiching degree of $\C$ is at most $k = \widetilde{O}(\frac{{\nhalfspaces}^6}{\eps^2})$. For Decision Trees of halfspaces of size $\dtsize$ and depth $\aczdepth$, the bound is $\degbound = \widetilde{O}(\frac{\dtsize^2 \aczdepth^6}{\eps^2})$.
\end{lemma}

The above result implies the following corollary.

\begin{corollary}[TDS Learner for Intersections and Decision Trees of Halfspaces]\label{corollary:tds-dts-intersections-of-halfspaces}
    Let $\Dgeneric$ be either the standard Gaussian in $\R^d$ or the uniform distribution over the hypercube in $d$ dimensions. Then, there is an algorithm that runs in time $d^{\tilde{O}(\nhalfspaces^6/\eps^2)}$ and TDS learns intersections of $\nhalfspaces$ halfspaces with respect to $\Dgeneric$ up to error $32\optcommon+O(\eps)$. For Decision Trees of halfspaces with size $\dtsize$ and depth $\nhalfspaces$ the bound is $d^{\tilde{O}(\dtsize^2\nhalfspaces^6/\eps^2)}$.
\end{corollary}

In order to apply the structural result we need from \cite{gopalan2010fooling}, we first provide a formal definition for the notion of hypercontractivity.

\begin{definition}[Hypercontractivity]\label{definition:hypercontractivity}
    Let $\Dgeneric_1$ be a distribution over $\R$ and let $T\in \N$, $T>2$, $\eta\in(0,1)$. We say that $\Dgeneric_1$ is $(T,2,\eta)$-hypercontractive if $\E[x^T] < \infty$ and for any $a\in\R$ we have 
    \[ 
        \E_{x\sim \Dgeneric_1}[(a+\eta x)^T]^{1/T} \le \E_{x\sim \Dgeneric_1}[(a+\eta x)^2]^{1/2}
    \]
\end{definition}

The following result can be used to show \Cref{lemma:l2-sandwiching-intersections}.

\begin{proposition}[Modification of Theorem 10.4 from \cite{gopalan2010fooling}]\label{proposition:auxiliary-lemma-l2-sandwiching-intersections}
    Let $r\in \N$, $\sigma\in(0,1)$, $T\in \N$, $\eta>0$ and $t>4$ be parameters and consider $\Dgeneric$ to be a product distribution over $\X\subseteq\R^d$ such that each of its independent coordinates is $(4,2,\eta)$-hypercontractive, and $(T,2,4/t)$-hypercontractive. Suppose that $T \ge Cr\log(rt)$ for some sufficiently large universal constant $C>0$ and $T$ is even. Then, for any function of the form $h:\X\to \R$, $h(\x) = \ind\{\w\cdot \x \ge \tau\}$, where $\w\in\R^d$ and $\tau\in \R$, there is a polynomial $p:\X\to\R$ such that the following are true.
    \begin{enumerate}[label=\textnormal{(}\roman*\textnormal{)}]
        \item The degree of $p$ is at most $\degbound = \poly(\log t, \frac{1}{\eta})\cdot \frac{1}{\sigma} + O(\frac{T}{r})$.
        \item For any $\x\in \X$ we have $p(\x) \ge h(\x)$.
        \item The expected distance between $p$ and $h$ is bounded by $\E_{\x\sim\Dgeneric}[p(\x)-h(\x)] \le O(\sigma^{\frac{1}{2}}+ \frac{rt\log(rt)}{T})$.
        \item The values of $p$ are upper bounded with high probability, i.e., $\pr_{\x\sim\Dgeneric}[p(\x)>1+\frac{1}{r^2}] \le 2^{-T/r}$.
        \item The $L_{2r}(\Dgeneric)$ norm of $p$ is bounded by $\|p\|_{L_{2r}(\Dgeneric)} \le 1 + \frac{2}{r^2}$.
    \end{enumerate}
\end{proposition}

\begin{proof}[Proof of \Cref{lemma:l2-sandwiching-intersections}]
    Let $\concept\in \C$ be an intersection of ${\nhalfspaces}$ halfspaces over $\X$, i.e., $\concept$ can be written in the following form
    \begin{align*}
        \concept(\x) &= 2\prod_{j=1}^{\nhalfspaces} h_j(\x) - 1, \text{ where }h_j(\x) = \ind\{\w_j\cdot \x + \tau_j\} \text{ for some }\w_j\in \R^d, \tau_j\in\R
    \end{align*}
    Note that if $\concept$ is a Decision Tree of halfspaces of size $\dtsize$ and depth $\aczdepth$, then $\concept$ can be written as a sum of at most $\dtsize$ intersections of $\aczdepth$ halfspaces and it suffices to use accuracy parameter $\eps/\dtsize$ for each intersection. 
    
    Back to the case where $\concept$ is an intersection of $\aczdepth$ halfspaces, we will apply \Cref{proposition:auxiliary-lemma-l2-sandwiching-intersections} in a way similar to the proof of Lemma 10.1 in \cite{gopalan2010fooling}. However, our goal here is to show that \Cref{proposition:auxiliary-lemma-l2-sandwiching-intersections} implies the existence of $\L_2$ (rather than $\L_1$) sandwiching polynomials for $\concept$. We use the following standard fact about the Gaussian and Uniform distributions.
    
    \begin{claim}[Hypercontractivity of Gaussian and Uniform marginals, see e.g. \cite{krakowiak1988hypercontraction,wolff2007hypercontractivity,gopalan2010fooling}]
        If $\Dgeneric$ is either the standard Gaussian $\Gauss(0,I_d)$ over $\R^d$ or the uniform distribution over the hypercube $\cube{d}$, then, for some universal constant $C>0$, each of the coordinates of $\Dgeneric$ is $(\lceil Ct^2\rceil,2,\frac{4}{t})$-hypercontractive for any $t>0$ and, in particular, each one is also $(4,2,\frac{1}{\sqrt{3}})$-hypercontractive.
    \end{claim}

    We may apply \Cref{proposition:auxiliary-lemma-l2-sandwiching-intersections} for each $h_j$ with parameters $r = 2\nhalfspaces, \sigma = \frac{\eps^2}{C{\nhalfspaces}^4}, t = C\frac{{\nhalfspaces}^3}{\eps}\log({\nhalfspaces}/\eps)$, $\eta=1/\sqrt{3}$ and $T=Ct^2$, for some sufficiently large universal constant $C$ to obtain a polynomial $p_j$ of degree $k = \widetilde{O}(\frac{{\nhalfspaces}^5}{\eps^2})$ such that the following are true.
    \begin{align}
        &p_j(\x)\ge h_j(\x) \text{ for all }\x\in\X \label{property:intersections-l2-pointwise-bound} \\
        &\eps_1 := \E_\Dgeneric[p_j(\x)-h_j(\x)] = O\Bigr(\frac{\eps}{{\nhalfspaces}^2}\Bigr) \label{property:intersections-l2-eps1} \\
        &\eps_2 := \pr_\Dgeneric\Bigr[p_j(\x) > 1+\frac{1}{4{\nhalfspaces}^2}\Bigr] \le 2^{-\Omega(\frac{{\nhalfspaces}^5}{\eps^2} \log^2({\nhalfspaces}/\eps))} \label{property:intersections-l2-eps2} \\
        &\|p_j\|_{L_{4m}(\Dgeneric)} \le 1+\frac{1}{2{\nhalfspaces}^2} \label{property:intersections-l2-pnorm}
    \end{align}

    We will now construct a polynomial $\pup$ of degree $\widetilde{O}(\frac{{\nhalfspaces}^6}{\eps^2})$ such that $\pup(\x)\ge \concept(\x)$ for all $\x\in \X$ and also $\E_{\Dgeneric}[(\pup(\x)-\concept(\x))^2] \le \eps/4$. This implies the existence of a corresponding polynomial $\pdown$ with $\pdown(\x)\le \concept(\x)$ for all $\x\in\X$ and $\E_{\Dgeneric}[(\pup(\x)-\pdown(\x))^2] \le \eps$ via a symmetric argument. Our proof consists of a hybrid argument similar to the one used in the proof of Lemma 10.1 in \cite{gopalan2010fooling}, modified to provide a bound for the $\L_2$ error of approximation.

    We pick $\pup = 2p-1$, where $p = \prod_{j=1}^{\nhalfspaces}p_j$. Let $p^{(0)} = \prod_{j=1}^{\nhalfspaces} h_j$, $p^{(i)} = (\prod_{j=1}^ip_j)(\prod_{j=i+1}^{\nhalfspaces} h_j)$ and $p^{({\nhalfspaces})}=p$. We then have the following.
    \begin{align*}
        \|p-h\|_{\L_2(\Dgeneric)} &= \|p^{({\nhalfspaces})}-p^{(0)}\|_{\L_2(\Dgeneric)} \le \sum_{i=1}^{\nhalfspaces} \|p^{(i)}-p^{(i-1)}\|_{\L_2(\Dgeneric)} \\
        &=\sum_{i=1}^{\nhalfspaces} {{}}\Bigr\| {{}}\Bigr(\prod_{j=1}^{i-1}p_j{{}}\Bigr) {{}}\Bigr(\prod_{j=i+1}^{\nhalfspaces} h_j{{}}\Bigr) (p_i-h_i) {{}}\Bigr\|_{\L_2(\Dgeneric)} \\
        &\le \sum_{i=1}^{\nhalfspaces} {{}}\Bigr\| {{}}\Bigr(\prod_{j\neq i} p_j{{}}\Bigr) (p_i-h_i) {{}}\Bigr\|_{\L_2(\Dgeneric)} \tag{by property \eqref{property:intersections-l2-pointwise-bound}}
    \end{align*}
    For any fixed $i\in[{\nhalfspaces}]$ we have
    \begin{align*}
        {{}}\Bigr\| {{}}\Bigr(\prod_{j\neq i} p_j{{}}\Bigr) (p_i-h_i) {{}}\Bigr\|_{\L_2(\Dgeneric)}^2 &= \E_{\Dgeneric}{{}}\Bigr[ {{}}\Bigr(\prod_{j\neq i} p_j^2(\x){{}}\Bigr) (p_i(\x)-h_i(\x))^2 {{}}\Bigr] \\
        &\le \E_{\Dgeneric}{{}}\Bigr[ {{}}\Bigr(\prod_{j\neq i} p_j^2(\x){{}}\Bigr) (p_i(\x)-h_i(\x))p_i(\x) {{}}\Bigr] \tag{since $h_i\ge 0$ and $p_i\ge h_i$}
    \end{align*}
    In order to bound the quantity $\E_{\Dgeneric}[ (\prod_{j\neq i} p_j^2) (p_i-h_i)p_i ]$, we split the expectation according to the event $\event$ that $ (\prod_{j\neq i} p_j)\sqrt{p_i} < 2$. In particular, we have that $\E_{\Dgeneric}[ (\prod_{j\neq i} p_j^2) (p_i-h_i)p_i \ind\{\event\} ]$ is at most $4\eps_1$ by property \eqref{property:intersections-l2-eps1} and $\E_{\Dgeneric}[ (\prod_{j\neq i} p_j^2) (p_i-h_i)p_i \ind\{\neg\event\} ]$ is bounded as follows.
    \begin{align*}
        \E_{\Dgeneric}{{}}\Bigr[ {{}}\Bigr(\prod_{j\neq i} p_j^2(\x){{}}\Bigr) &(p_i(\x)-h_i(\x))p_i(\x) \ind{{}}\Bigr\{ {{}}\Bigr(\prod_{j\neq i} p_j(\x){{}}\Bigr)\sqrt{p_i(\x)} \ge 2 {{}}\Bigr\}{{}}\Bigr] \le \\
        &\le \E_{\Dgeneric}{{}}\Bigr[ {{}}\Bigr(\prod_{j\in [{\nhalfspaces}]} p_j^2(\x){{}}\Bigr) \ind{{}}\Bigr\{ {{}}\Bigr(\prod_{j\neq i} p_j(\x){{}}\Bigr)\sqrt{p_i(\x)} \ge 2 {{}}\Bigr\}{{}}\Bigr] \tag{by property \eqref{property:intersections-l2-pointwise-bound}}
    \end{align*}
    We now observe that whenever $(\prod_{j\neq i} p_j(\x){{}})\sqrt{p_i(\x)} \ge 2$, there must exist some index $j'$ such that $p_{j'}(\x)>1+\frac{1}{4\nhalfspaces^2}$ and, therefore, we can further bound the above quantity by the following one.
    \begin{align*}
        &\E_{\Dgeneric}{{}}\Bigr[\sum_{j'=1}^{\nhalfspaces}\ind{{}}\Bigr\{p_{j'}(\x)>1+\frac{1}{4\nhalfspaces^2}{{}}\Bigr\} {{}}\Bigr(\prod_{j\in [{\nhalfspaces}]} p_j^2(\x){{}}\Bigr) {{}}\Bigr] = \sum_{j'=1}^{\nhalfspaces}\E_{\Dgeneric}{{}}\Bigr[\ind{{}}\Bigr\{p_{j'}(\x)>1+\frac{1}{4\nhalfspaces^2}{{}}\Bigr\} {{}}\Bigr(\prod_{j\in [{\nhalfspaces}]} p_j^2(\x){{}}\Bigr) {{}}\Bigr]
    \end{align*}
    In the above expression we used linearity of expectation. We now apply H\"older's inequality and obtain the bound $\sum_{j'=1}^{\nhalfspaces} ({\pr_\Dgeneric[p_{j'}(\x) > 1+\frac{1}{4\nhalfspaces^2}]})^{\frac{1}{2}} \prod_{j=1}^{\nhalfspaces}\bigr( \E_\Dgeneric[p_j^{4\nhalfspaces}(\x)]\bigr)^{\frac{1}{2\nhalfspaces}}$. Due to properties \eqref{property:intersections-l2-eps2} and \eqref{property:intersections-l2-pnorm}, we finally have the bound $\nhalfspaces\sqrt{\eps_2}\cdot \prod_{j=1}^{\nhalfspaces}\|p_j\|_{L_{4\nhalfspaces}}^2 \le \nhalfspaces\sqrt{\eps_2} (1+\frac{1}{2\nhalfspaces^2})^{2\nhalfspaces} \le 3\nhalfspaces\sqrt{\eps_2}$.
    Therefore, in total, we have $\|p-h\|_{L_{2}(\Dgeneric)}^2 \le 4\nhalfspaces^2\eps_1 + 3\nhalfspaces^3 \eps_2 \le \eps$, which implies that $\|\pup-\concept\|_{\L_2{(\Dgeneric)}} \le \eps$ and $\pup\ge \concept$.  
\end{proof}

\section{Lower Bounds}\label{appendix:lower-bounds}

\subsection{Lower Bound for Realizable TDS Learning of Monotone Functions}

We now prove \Cref{thm: tds learning of monotone is hard}, which we restate here for convenience.

\begin{theorem}[Hardness of TDS Learning Monotone Functions]
Let the accuracy parameter $\epsilon$ be at most $0.1$ and the success probability parameter $\delta$ also be at most $0.1$. Then,
    in the realizable setting, any TDS learning algorithm for the class of monotone functions over $\{\pm 1\}^d$ with accuracy parameter $\epsilon$ and success probability at least $1-\delta$ requires either $2^{0.04 d}$ training samples or $2^{0.04 d}$ testing samples for all sufficiently large values of $d$.
\end{theorem}

We will need the following standard fact, see for example \cite{rubinfeld2022testing} for a proof:
\begin{fact}
\label{fact: substituting uniform with subsample does likely indistinguishable}
    For any distribution $D$ over any domain, let multisets $T_1$ and $T_2$ be sampled as follows:
    \begin{enumerate}
        \item Set $T_1$ is $N$ i.i.d. samples from $D$.
        \item First, multiset $S$ is formed by taking $M$ i.i.d. samples from $D$. Then, multiset $T_2$ is formed by taking $N$ i.i.d. uniform elements from $S$.
    \end{enumerate}
Then, the statistical distance between the distributions of $T_1$ and $T_2$ is at most $\frac{N^2}{M}$.
\end{fact}

Now, we prove \Cref{thm: tds learning of monotone is hard}.
\begin{proof}[Proof of \Cref{thm: tds learning of monotone is hard}]

We fix $\delta \leq 0.1$ and also fix $\epsilon \leq 0.1$. Let $\mathcal{A}$ be an algorithm that takes $N\leq 2^{0.04d}$ testing samples and $N\leq 2^{0.04d}$ training samples, and either outputs REJECT, or (ACCEPT, $\hat{f}$) for a function $\hat{f}:\{\pm 1\}^d \rightarrow \{\pm 1\}$. We argue that for, a sufficiently large $d$, the algorithm $\mathcal{A}$ will fail to be a TDS-learning algorithm for monotone functions over $\{\pm 1\}^d$.

Let $f$ be some function mapping $\{\pm 1\}^d \rightarrow \{\pm 1\}$ and let a multiset $S$ consist of elements in $\{\pm 1\}^d$. We define $\mathcal{T}(f, S)$ to be a random
variable supported on $\{\text{Yes}, \text{No}\}$ determined as follows (informally, if $\mathcal{A}$ is a TDS-learner for monotone functions, then $\mathcal{T}(f, S)$ will allow us to distinguish a uniform distribution over $S$ from the uniform distribution over $\{\pm 1\}^d$):
\begin{enumerate}
    \item Let $\Strain \subset \{\pm 1\}^d \times \{\pm 1\}$ consist of $N$ pairs $(\x, f(\x))$, where $\x$ are drawn i.i.d. uniformly from $\{\pm 1\}^d$.
    \item Let $\Stest$ consist of $N$ i.i.d. uniform samples from set $S$.
    \item The algorithm $\mathcal{A}$ is run on $(\Strain, \Stest)$. 
    \item If $\mathcal{A}$ outputs REJECT, then output $\mathcal{T}(f, S)=$No.
    \item If $\mathcal{A}$  outputs (ACCEPT, $\hat{f}$), then
    \begin{enumerate}
        \item Obtain a new set $X_2$ of $10000$ i.i.d. uniform samples from $S$.
        \item If, on the majority of points $\x$ in $X_2$, we have $\hat{f}(\x)=1$, then output No. 
        \item Otherwise, output Yes.
    \end{enumerate}
\end{enumerate}

For a multiset $S$ consisting of elements in $\{\pm 1\}^d$, let $f_S$ be the monotone function defined as follows:
\[
f_S(\x):=
\begin{cases}
    +1 &\text{ if there exists }\z\in S:~\x \succeq \z,\\
    -1 &\text{ otherwise}.
\end{cases}
\]
First, we observe that if $\mathcal{A}$ is indeed a $(\epsilon, \delta)$-TDS learning algorithm for monotone functions over $\{\pm 1\}^d$, then:
\begin{itemize}
    \item $\mathcal{T}(-1, \{\pm 1\}^d)$=Yes with probability at least $\frac{2}{3}$ (from here on, by $-1$ we mean the function that maps every element in $\{\pm 1\}^d$ into $-1$). This is true because, by the definition of a TDS learner, since $\Strain$ comes from the uniform distribution over $\{\pm 1\}^d$, with probability at least $1-2\delta=0.8$ the algorithm $\mathcal{A}$ will output (ACCEPT, $\hat{f}$) for some $\hat{f}$ satisfying $\pr_{\x \sim \{\pm 1\}^d}[\hat{f}(x) \neq -1] \leq \epsilon=0.1$. Then, via a standard Hoeffding bound, with probability at least $0.9$ on the majority of elements $\x$ in $X_2$ we have $\hat{f}(\x)=-1$ and then $\mathcal{T}(-1, \{\pm 1\}^d)$=Yes.
    \item For any multiset $S$ with elements in $\{\pm 1\}^d$, we have
    $\mathcal{T}(f_S, S)=$No with probability at least $\frac{2}{3}$. Indeed, from the definition of a TDS learning algorithm, we see that, with probability at least $1-\delta=0.9$, the algorithm $\mathcal{A}$ will either output 
    \begin{itemize}
        \item Output reject, in which case $\mathcal{T}(f_S, S)=$No.
        \item Output (ACCEPT, $\hat{f}$) with $\pr_{\x \sim S}[\hat{f}(\x) \neq f_S(\x)]\leq \epsilon=0.1$. But we know that $f_S$ takes values $+1$ on all elements in $S$. Therefore, $\pr_{\x \sim S}[\hat{f}(\x) \neq f_S(\x)]\leq 0.1$. Then, via a standard Hoeffding bound, with probability at least $0.9$ on the majority of elements $\x$ in $X_2$ we have $\hat{f}(\x)=+1$ and then $\mathcal{T}(f_S, S)$=No.
    \end{itemize}
\end{itemize}
In particular, if $S$ is obtained by picking $M=2^{0.1d}$ i.i.d. elements from $\{\pm 1\}^d$, we have
\begin{equation}
\label{eq: a TDS learner lets you distinguish}
\left \lvert
\pr_{
\substack{
    S \sim \Unif(\cube{d})^{\otimes M}\\
    \text{Randomness of }\mathcal{T}
}
}
    [
    \mathcal{T}(f_S, S)=\text{Yes}
    ]
    -\pr_{
\substack{
    \text{Randomness of }\mathcal{T}
}
}
    [
    \mathcal{T}(-1, \{\pm 1\}^d)=\text{Yes}
    ]
    \right \rvert
    >\frac{1}{3}.
\end{equation}
The rest of the proof argues, via a hybrid argument, that this is impossible. To be specific, we claim that for sufficiently large $d$ the following two inequalities must hold
\begin{equation}
\label{eq: substituting uniform with subsample does likely indistinguishable}
\left \lvert
\pr_{
\substack{
    S \sim \Unif(\cube{d})^{\otimes M}\\
    \text{Randomness of }\mathcal{T}
}
}
    [
    \mathcal{T}(-1, S)=\text{Yes}
    ]
    -\pr_{
\substack{
    \text{Randomness of }\mathcal{T}
}
}
    [
    \mathcal{T}(-1, \{\pm 1\}^d)=\text{Yes}
    ]
    \right \rvert
     \leq \frac{N^2}{M}.
\end{equation}

\begin{equation}
\label{eq: swapping function to be always minus one is undetectable}
\left \lvert
\pr_{
\substack{
    S \sim \Unif(\cube{d})^{\otimes M}\\
    \text{Randomness of }\mathcal{T}
}
}
    [
    \mathcal{T}(f_S, S)=\text{Yes}
    ]
    -\pr_{
\substack{
    S \sim \Unif(\cube{d})^{\otimes M}\\
    \text{Randomness of }\mathcal{T}
}
}
    [
    \mathcal{T}(-1, S)=\text{Yes}
    ]
    \right \rvert
    \leq  2\left(\frac{3}{4}\right)^d M N.
\end{equation}
We observe that Equation \ref{eq: substituting uniform with subsample does likely indistinguishable} follows immediately from \Cref{fact: substituting uniform with subsample does likely indistinguishable}, because if Equation \ref{eq: substituting uniform with subsample does likely indistinguishable} didn't hold, then we would be able to achieve advantage greater than $\frac{M}{N^2}$ when distinguishing $N$ i.i.d. uniform samples from $\{\pm 1\}^d$ from $N$ i.i.d. uniform examples from $S$.

Now we prove Equation \ref{eq: swapping function to be always minus one is undetectable}. Let $\Strain^{\mathcal{T}(f_S , S)}$ denote the collection of pairs $\{(\x, f_S (\x))\}$ sampled in Step 1 of $\mathcal{T}(f_S , S)$.  
Analogously, let $\Strain^{ \mathcal{T}(-1, S)}$ denote the collection of pairs $(\x, -1)$ in set used in procedure $\mathcal{T}(-1, S)$. In either case, the elements in $\Strain^{\mathcal{T}(f_S , S)}$ and $\Strain^{\mathcal{T}(-1 , S)}$ are i.i.d. uniformly random elements in $\{\pm 1\}^d$. 
Let $E^{ \mathcal{T}(-1, S)}$ be the event, over the choice of $S$ and the choice of $\Strain^{\mathcal{T}(-1, S)}$, that for every $(\x, -1) \in \Strain^{\mathcal{T}(-1, S)}$ there is no $\z$ in $S$ satisfying $\x\succeq \z$. Analogously, let $E^{ \mathcal{T}(f_S, S)}$ be the event, over the choice of $S$ and the choice of $\Strain^{\mathcal{T}(f_S , S)}$, that for every $(\x, f_S(\x))\in \Strain^{\mathcal{T}(f_S , S)}$ there is no $\z$ in $S$ satisfying $\x \succeq \z$. We observe that 
\begin{align}
\label{eq: conditioned on f_S beeing zero on all sample distribution of distinguisher is same}
\pr_{
\substack{
    S \sim \Unif(\cube{d})^{\otimes M}\\
    \text{Randomness of }\mathcal{T}
}
}
    \left[
    \mathcal{T}(f_S, S)=\text{Yes}
\bigg \vert 
    E^{ \mathcal{T}(f_S, S)}
    \right]
    =
    \pr_{
\substack{
    S \sim \Unif(\cube{d})^{\otimes M}\\
    \text{Randomness of }\mathcal{T}
}
}
    \left[
    \mathcal{T}(-1, S)=\text{Yes}
    \bigg \vert 
    E^{ \mathcal{T}(-1, S)}
    \right]
\end{align}
which is true because, subject to $E^{ \mathcal{T}(f_S, S)}$ or $E^{ \mathcal{T}(-1, S)}$, the function $f_S$ takes values of $-1$ on every element $\x$ in $\Strain^{\mathcal{T}(f_S , S)}$ and $\Strain^{\mathcal{T}(-1 , S)}$ respectively. We also see that the random variables $(S,\Strain^{\mathcal{T}(f_S , S)})$ and $(S,\Strain^{\mathcal{T}(-1 , S)})$ are identically distributed (conditioned on $E^{ \mathcal{T}(f_S, S)}$ and $E^{ \mathcal{T}(-1, S)}$ respectively). We also observe that 
\begin{equation}
\label{eq: very likely all labels are zero}
\pr_{
\substack{
    S \sim \Unif(\cube{d})^{\otimes M}\\
    \text{Randomness of }\mathcal{T}
}
}
    \left[
    E^{ \mathcal{T}(f_S, S)}
    \right]
    =
    \pr_{
\substack{
    S \sim \Unif(\cube{d})^{\otimes M}\\
    \text{Randomness of }\mathcal{T}
}
}
    \left[
    E^{ \mathcal{T}(-1, S)}
    \right]
    \leq \left(\frac{3}{4}\right)^d M N,
\end{equation}
where the equality of the two probabilities follows immediately by definition, and the upper bound of $\left(\frac{3}{4}\right)^d M N$ is true for the following reason. Let $\z$ and $\x$ be a pair of i.i.d. uniformly random elements in $\{\pm 1\}^d$, then $\pr[\x \succeq \z] = \left(\frac{3}{4}\right)^d$ as each bit of $\x$ and $\z$ are independent and for each of the bits we have $\x \geq \z$ with probability exactly $3/4$. Now, taking a union bound over every $(\x, -1)\in \Strain^{ \mathcal{T}(-1, S)}$ and $\z \in S$, we obtain the bound in Equation \ref{eq: very likely all labels are zero}.

Overall, combining Equation \ref{eq: substituting uniform with subsample does likely indistinguishable} with Equation \ref{eq: swapping function to be always minus one is undetectable} and substituting $N\leq 2^{0.04 d}$ and $M=2^{0.1 d}$ we get
\begin{multline*}
\left \lvert
\pr_{
\substack{
    S \sim \Unif(\cube{d})^{\otimes M}\\
    \text{Randomness of }\mathcal{T}
}
}
    [
    \mathcal{T}(f_S, S)=\text{Yes}
    ]
    -\pr_{
\substack{
    \text{Randomness of }\mathcal{T}
}
}
    [
    \mathcal{T}(-1, \{\pm 1\}^d)=\text{Yes}
    ]
    \right \rvert
    \leq \\ \frac{N^2}{M}
    +2\left(\frac{3}{4}\right)^d M N
    =2^{-\Omega(d)},
\end{multline*}
which is in contradiction with Equation \ref{eq: a TDS learner lets you distinguish} for a sufficiently large value of $d$. This proves that $\mathcal{A}$ is not a $(\epsilon, \delta)$-TDS learning algorithm for monotone functions. 
\end{proof}

\subsection{Lower Bound for Realizable TDS Learning of Convex Sets}

We now prove \Cref{thm: tds learning of convex is hard} which we restate here for convenience.

\begin{theorem}[Hardness of TDS Learning Convex Sets]
	Let the accuracy parameter $\epsilon$ be at most $0.1$ and the success probability parameter $\delta$ also be at most $0.1$. Then,
	in the realizable setting, any TDS learning algorithm for the class of indicators of convex sets under the standard Gaussian distribution on $\R^d$ requires either $2^{0.04 d}$ training samples or $2^{0.04 d}$ testing samples for all sufficiently large values of $d$.
\end{theorem}

We will need the following standard facts about Gaussian distributions:
\begin{fact}[Concentration of Gaussian norm, see e.g. Lemma 8.1 in \cite{birge1997model}]
	\label{fact: concentration of Gaussian norm}
For any $\eta>0$ it is the case that 
\[
\pr_{\x \in \mathcal{N}(0, I_d)}
\left[
d-2\sqrt{d\ln \left(\frac{2}{\eta}\right)}
\leq
\norm{\x}_2^2
\leq
d+2\sqrt{d\ln \left(\frac{2}{\eta}\right)}
+2\ln \left(\frac{2}{\eta}\right)
\right]
\geq 1- \eta
\]
\end{fact}
\begin{fact}[Concentration of Gaussian norm. See e.g. \cite{rubinfeld2022testing}.]
	\label{fact: pairs of Gaussians far}
	For any $r>0$ it is the case that 
	\[
	\pr_{\x^1, \x^2 \in \mathcal{N}(0, I_d)}
	\left[
	\norm{\x^1-\x^2}_2\leq r
	\right]
	 \leq
	\left(
	\frac{64 r^2}{d}\right)^{d/2}
	\]
\end{fact}

Recall that we use $\B_a$ to denote the origin-centered closed ball in $\R^d$ of radius $a$. Using ${\conv}(\cdot)$ to denote the convex hull of a set of points, will state the following geometric observation of \cite{rubinfeld2022testing} about convex hulls of a collection of point. 
\begin{fact}[\cite{rubinfeld2022testing}]
	\label{fact: far away points have well behaved convex hull}
For any $a>0$, let $\{\x^i\}_{i=1}^M$ be a collection of points in $\B_b \setminus \B_a$. If for every pair of points $(\x^i, \x^j)$ the $\norm{\x^i-\x^j}_2$ is greater than $2\sqrt{b^2-a^2}$, then for every $i$ and $j$ we have
\[
{\conv}
(\x^i, \B_a) 
\cap {\conv}
(\x^i, \B_a)
=\B_a 
\]
and also
\[
{\conv}
(\x^1, \cdots, \x^M, \B_a)
=\cup_i {\conv}
(\x^i, \B_a).
\]
\end{fact}

For the rest of the section we will set 
\begin{align}
	\label{eq: definition of a and b}
	a=\sqrt{d-2\sqrt{d\ln \left(\frac{1}{50}\right)}} && b=\sqrt{d+2\sqrt{d\ln \left(\frac{1}{50}\right)}
	+2\ln \left(\frac{1}{50}\right)},
\end{align}
and from Fact \ref{fact: concentration of Gaussian norm} we see that the norm a standard Gaussian vector in $\R^d$ falls in interval $(a,b)$ with probability at least $0.99$. 

Now, we are ready to prove Theorem \ref{thm: tds learning of convex is hard}.
\begin{proof}[Proof of \Cref{thm: tds learning of monotone is hard}]
	
	We fix $\delta \leq 0.1$ and also fix $\epsilon \leq 0.1$. Let $\mathcal{A}$ be an algorithm that takes $N\leq 2^{0.04d}$ testing samples and $N\leq 2^{0.04d}$ training samples, and either outputs REJECT, or (ACCEPT, $\hat{f}$) for a function $\hat{f}:\R^d \rightarrow \{\pm 1\}$. We argue that for, a sufficiently large $d$, the algorithm $\mathcal{A}$ will fail to be a TDS-learning algorithm for convex sets under the Gaussian distribution on $\R^d$. 
	
	For a set $S$ we will define $g_S$ as the indicator of the convex set ${\conv}(S\cap (\B_b\setminus \B_a), \B_a)$. And in this section we denote the uniform distribution over $S$ as $\mathbb{U}_S$.
	
	Let $f$ be some function mapping $\R^d \rightarrow \{\pm 1\}$ and let a set $D$ be a distribution over $\R^d$. We define $\mathcal{H}(f, D)$ to be a random
	variable supported on $\{\text{Yes}, \text{No}\}$ determined as follows (informally, if $\mathcal{A}$ is a TDS-learner for convex sets, then $\mathcal{H}(f, D)$ will allow us to distinguish $D$ from the Gaussian distribution over $\R^d$):
	\begin{enumerate}
		\item Let $\Strain \subset \R^d \times \{\pm 1\}$ consist of $N$ pairs $(\x, f(\x))$, where $\x$ are drawn i.i.d. from $\mathcal{N}(0, I_d)$.
		\item Let $\Stest$ consist of $N$ i.i.d. uniform samples from $D$.
		\item The algorithm $\mathcal{A}$ is run on $(\Strain, \Stest)$. 
		\item If $\mathcal{A}$ outputs REJECT, then output $\mathcal{H}(f, S)=$No.
		\item If $\mathcal{A}$  outputs (ACCEPT, $\hat{f}$), then
		\begin{enumerate}
			\item Obtain a new set $X_2$ of $10000$ i.i.d. samples from $D$.
			\item If, on the majority of points $\x$ in $X_2$, we have $\hat{f}(\x)=-1$, then output No. 
			\item Otherwise, output Yes.
		\end{enumerate}
	\end{enumerate}

	First, we observe that if $\mathcal{A}$ is indeed a $(\epsilon, \delta)$-TDS learning algorithm for convex sets over $\R^d$ under $\mathcal{N}(0, I_d)$, then:
	\begin{itemize}
		\item $\mathcal{H}(g_{\emptyset}, \mathcal{N}(0, I_d))$=Yes with probability at least $\frac{2}{3}$ (from here on, by $-1$ we mean the function that maps every element in $\{\pm 1\}^d$ into $-1$). This is true because, by the definition of a TDS learner, since $\Strain$ comes from the uniform distribution over $\mathcal{N}(0, I_d)$, with probability at least $1-2\delta=0.8$ the algorithm $\mathcal{A}$ will output (ACCEPT, $\hat{f}$) for some $\hat{f}$ satisfying $\pr_{\x \sim \mathcal{N}(0, I_d)}[\hat{f}(\x) \neq g_{\emptyset}(\x)] \leq \epsilon=0.1$. 
		Since $a$ was chosen is such manner that $\Pr_{\x \in \mathcal{N}(0, I_d)}[\x \in \B_a]<0.01$, and $g_{\emptyset}$ is the indicator function of $\B_a$, we have $\Pr_{\x \in \mathcal{N}(0, I_d)}[g_{\emptyset}(\x) \neq -1]<0.01$. Via a union bound, we see that $\pr_{\x \sim \mathcal{N}(0, I_d)}[\hat{f}(\x) \neq -1] \leq0.11$.
		   Then, via a standard Hoeffding bound, with probability at least $0.9$ on the majority of elements $\x$ in $X_2$ we have $\hat{f}(\x)=-1$ and then $\mathcal{H}(g_{\emptyset}, \mathcal{N}(0, I_d))$=Yes.
		\item For any set $S$ with elements in $\R^d$, we have
		$\mathcal{H}(g_S, \mathbb{U}_S)=$No with probability at least $\frac{2}{3}$. Indeed, from the definition of a TDS learning algorithm, we see that, with probability at least $1-\delta=0.9$, the algorithm $\mathcal{A}$ will either 
		\begin{itemize}
			\item Output reject, in which case $\mathcal{H}(g_S, \mathbb{U}_S)=$No.
			\item Output (ACCEPT, $\hat{f}$) with $\pr_{\x \sim \mathbb{U}_S}[\hat{f}(\x) \neq g_S(\x)]\leq \epsilon=0.1$. But we know that $g_S$ takes values $+1$ on all elements in $S$. Therefore, $\pr_{\x \sim\mathbb{U}_S}[\hat{f}(\x) \neq f_S(\x)]\leq 0.1$. Then, via a standard Hoeffding bound, with probability at least $0.9$ on the majority of elements $\x$ in $X_2$ we have $\hat{f}(\x)=+1$ and then $\mathcal{H}(g_S, \mathbb{U}_S)$=No.
		\end{itemize}
	\end{itemize}
	In particular, if $S$ is obtained by picking $M=2^{0.1d}$ i.i.d. elements from $\mathcal{N}(0, I_d)$, we have
	\begin{equation}
		\label{eq: a TDS learner lets you distinguish Gaussian}
		\left \lvert
		\pr_{
			\substack{
				S \sim \mathcal{N}(0, I_d)^{\otimes M}\\
				\text{Randomness of }\mathcal{H}
			}
		}
		[
		\mathcal{H}(g_S, \mathbb{U}_S)=\text{Yes}
		]
		-\pr_{
			\substack{
				\text{Randomness of }\mathcal{H}
			}
		}
		[
		\mathcal{H}(g_{\emptyset}, \mathcal{N}(0, I_d))=\text{Yes}
		]
		\right \rvert
		>\frac{1}{3}.
	\end{equation}
	The rest of the proof argues, via a hybrid argument, that this is impossible. To be specific, we claim that for sufficiently large $d$ the following two inequalities must hold
	\begin{equation}
		\label{eq: substituting uniform with subsample does likely indistinguishable Gaussian}
		\left \lvert
		\pr_{
			\substack{
				S \sim \mathcal{N}(0, I_d)^{\otimes M}\\
				\text{Randomness of }\mathcal{H}
			}
		}
		[
		\mathcal{H}(g_{\emptyset}, \mathbb{U}_S)=\text{Yes}
		]
		-\pr_{
			\substack{
				\text{Randomness of }\mathcal{H}
			}
		}
		[
		\mathcal{H}(g_{\emptyset}, \mathcal{N}(0, I_d))=\text{Yes}
		]
		\right \rvert
		\leq \frac{N^2}{M}.
	\end{equation}
	
	\begin{align}
		\label{eq: swapping function to be always minus one is undetectable Gaussian}
		\biggr \lvert
		\pr_{
			\substack{
				S \sim \mathcal{N}(0, I_d)^{\otimes M}\\
				\text{Randomness of }\mathcal{H}
			}
		}
		[
		\mathcal{H}(g_S, \mathbb{U}_S)=\text{Yes}
		]
		&-\pr_{
			\substack{
				S \sim \mathcal{N}(0, I_d)^{\otimes M}\\
				\text{Randomness of }\mathcal{H}
			}
		}
		[
		\mathcal{H}(g_{\emptyset}, \mathbb{U}_S)=\text{Yes}
		]
		\biggr \rvert \nonumber\\
		&\leq  \left(
		\frac{64 (b^2-a^2)}{d}\right)^{d/2} (M+N)^2.
	\end{align}
	We observe that Equation \ref{eq: substituting uniform with subsample does likely indistinguishable Gaussian} follows immediately from \Cref{fact: substituting uniform with subsample does likely indistinguishable}, because if Equation \ref{eq: substituting uniform with subsample does likely indistinguishable Gaussian} didn't hold, then we would be able to achieve advantage greater than $\frac{M}{N^2}$ when distinguishing $N$ i.i.d. uniform samples from $\mathcal{N}(0, I_d)$ and $N$ i.i.d. uniform examples from $S$.
	
	Now we prove Equation \ref{eq: swapping function to be always minus one is undetectable Gaussian}. Let $\Strain^{\mathcal{H}(g_S, \mathbb{U}_S)}$ denote the collection of pairs $\{(\x, g_S (\x))\}$ sampled in Step 1 of $\mathcal{H}(g_S, \mathbb{U}_S)$.  
	Analogously, let $\Strain^{ \mathcal{H}(g_{\emptyset}, \mathbb{U}_S)}$ denote the collection of pairs $(\x, -1)$ in set used in procedure $\mathcal{H}(g_{\emptyset}, \mathbb{U}_S)$. In either case, the elements in $\Strain^{\mathcal{H}(g_S, \mathbb{U}_S)}$ and $\Strain^{\mathcal{H}(g_{\emptyset}, \mathbb{U}_S)}$ are i.i.d. elements from $\mathcal{N}(0, I_d)$. 
	Let $\event^{ \mathcal{H}(g_{S}, \mathbb{U}_S)}$ be the event, over the choice of $S$ and the choice of $\Strain^{\mathcal{H}(g_S , \mathbb{U}_S)}$, that for each pair of points $\x^1$ and $\x^2$ in $S \cup \{\x:~ (\x, g_S(x)) \in  \Strain^{\mathcal{H}(g_S , \mathbb{U}_S)}\}$
	 we have $\norm{\x^1 -\x^2}_2 > 2\sqrt{b^2 - a^2}$.
	Analogously, let $\event^{ \mathcal{H}(g_{S}, \mathbb{U}_S)}$ be the event, over the choice of $S$ and the choice of $\Strain^{\mathcal{H}(g_\emptyset , \mathbb{U}_S)}$, that for each pair of points $\x^1$ and $\x^2$ in $S \cup \{\x:~ (\x, g_\emptyset(x)) \in  \Strain^{\mathcal{H}(g_\emptyset , \mathbb{U}_S)}\}$
	we have $\norm{\x^1 -\x^2}_2 > 2\sqrt{b^2 - a^2}$.

	We first observe that subject to $\event^{ \mathcal{H}(g_{\emptyset}, \mathbb{U}_S)}$ it is the case that for every $\{(\x,g_S(\x))\}$ in $\Strain^{\mathcal{H}(g_S, \mathbb{U}_S)}$ it is the case that $g_S=g_\emptyset(x)$. For $\x \in \B_a \cup (\R \setminus \B_b)$ this is immediate because $g_S$ as the indicator of the convex set ${\conv}(S\cap (\B_b\setminus \B_a), \B_a)$. It remains to show this only for points $(\x, g_S(\x))\in \Strain^{\mathcal{H}(g_S , \mathbb{U}_S)}$ that also satisfy $\x \in \B_b\setminus \B_a$. Since $\x$ is outside $\B_a$, we have $g_{\emptyset}(\x)=-1$ and therefore we would like to show that $g_S(\x)$ also equals to $-1$. This is true because from Fact \ref{fact: far away points have well behaved convex hull} it is the case that if $\event^{ \mathcal{H}(g_{\emptyset}, \mathbb{U}_S)}$ takes place, then for every such $\x$ we have 
	\begin{multline*}
	{\conv}
	(\x, \B_a)
	\cap 
	{\conv}
	(S \cap (\B_b \setminus \B_a), \B_a)
	=
	{\conv}
	(\x, \B_a)
	\cap \left(
	\bigcup_{\z \in S \cap (\B_b \setminus \B_a)}
	{\conv}
	(\z \cap (\B_b \setminus \B_a), \B_a)
	\right)=\\
	\bigcup_{\z \in S \cap (\B_b \setminus \B_a)}
	\left(
		{\conv}
	(\x, \B_a)
	\cap \left(
	{\conv}
	(\z \cap (\B_b \setminus \B_a), \B_a)
	\right)
	\right)=\B_a,
	\end{multline*}
	which in particular implies that $\x$ is not in the convex hull ${\conv}
	(S \cap (\B_b \setminus \B_a), \B_a)$ and $g_S(\x)=-1$, concluding the proof of our observation.
	
	 We therefore conclude that distributions of $(S,\Strain^{\mathcal{H}(g_S, \mathbb{U}_S)})$ and $(S, \Strain^{\mathcal{H}(\mathcal{H}(g_\emptyset, \mathbb{U}_S)})$ are identically distributed conditioned on $\event^{ \mathcal{H}(g_S, \mathbb{U}_S)}$ and $\event^{ \mathcal{H}(g_{\emptyset}, \mathbb{U}_S)}$ respectively, which implies that
	\begin{align}
		\label{eq: conditioned on f_S beeing zero on all sample distribution of distinguisher is same Gaussian}
		\pr_{
			\substack{
				S \sim \mathcal{N}(0, I_d)^{\otimes M}\\
				\text{Randomness of }\mathcal{H}
			}
		}
		\biggr[
		\mathcal{H}(g_S, \mathbb{U}_S)=\text{Yes}
		\bigg \vert 
		\event^{ \mathcal{H}(g_S, \mathbb{U}_S)}
		\biggr]
		=
		\pr_{
			\substack{
				S \sim \mathcal{N}(0, I_d)^{\otimes M}\\
				\text{Randomness of }\mathcal{H}
			}
		}
		\biggr[
		\mathcal{H}(g_{\emptyset}, \mathbb{U}_S)=\text{Yes}
		\bigg \vert 
		\event^{ \mathcal{H}(g_{\emptyset}, \mathbb{U}_S)}
		\biggr],
	\end{align}
	 
	  We also observe that 
	\begin{equation}
		\label{eq: very likely all labels are the same Gaussian}
		\pr_{
			\substack{
				S \sim \mathcal{N}(0, I_d)^{\otimes M}\\
				\text{Randomness of }\mathcal{H}
			}
		}
		\left[
		\event^{ \mathcal{H}(g_S, \mathbb{U}_S)}
		\right]
		=
		\pr_{
			\substack{
				S \sim \mathcal{N}(0, I_d)^{\otimes M}\\
				\text{Randomness of }\mathcal{H}
			}
		}
		\left[
		\event^{ \mathcal{H}(g_{\emptyset}, \mathbb{U}_S)}
		\right]
		\leq \left(
		\frac{64 (b^2-a^2)}{d} \right)^{d/2} (M+N)^2,
	\end{equation}
	where the equality of the two probabilities follows immediately by definition, and the upper bound of $\left(\frac{64 (b^2-a^2)}{d}\right)^{d/2} (M+N)^2$ is true by applying Fact \ref{fact: pairs of Gaussians far} to each relevant pair of points.
	Therefore, we obtain the bound in Equation \ref{eq: very likely all labels are the same Gaussian}.
	
	Overall, combining Equation \ref{eq: substituting uniform with subsample does likely indistinguishable Gaussian} with Equation \ref{eq: swapping function to be always minus one is undetectable Gaussian} and substituting $N\leq 2^{0.04 d}$, $M=2^{0.1 d}$ as well as $a=\sqrt{d-2\sqrt{d\ln \left(\frac{1}{50}\right)}}$ and $b=\sqrt{d+2\sqrt{d\ln \left(\frac{1}{50}\right)}
	+2\ln \left(\frac{1}{50}\right)}$, we obtain
	\begin{multline*}
		\left \lvert
		\pr_{
			\substack{
				S \sim \mathcal{N}(0, I_d)^{\otimes M}\\
				\text{Randomness of }\mathcal{H}
			}
		}
		[
		\mathcal{H}(f_S, S)=\text{Yes}
		]
		-\pr_{
			\substack{
				\text{Randomness of }\mathcal{H}
			}
		}
		[
		\mathcal{H}(g_{\emptyset}, \mathcal{N}(0, I_d))=\text{Yes}
		]
		\right \rvert
		\leq \\ \frac{N^2}{M}
		+\left(\frac{64 (b^2-a^2)}{d} \right)^{d/2} (M+N)^2
		=
		2^{-0.02d}
		+\left(O\left(\frac{1}{\sqrt{d}}\right)\right)^{d/2}
		=2^{-\Omega(d)},
	\end{multline*}
	which is in contradiction with Equation \ref{eq: a TDS learner lets you distinguish Gaussian} for a sufficiently large value of $d$. This proves that $\mathcal{A}$ is not a $(\epsilon, \delta)$-TDS learning algorithm for convex sets. 
\end{proof}

\bibliographystyle{alpha}
\bibliography{main}
\appendix

\section{Sample Complexity of TDS Learning}\label{appendix:sc}

In the previous sections, we explored a number of computational aspects of TDS learning, deriving dimension efficient algorithms for several instantiations of our setting. In this section, we focus on the statistical aspects of TDS learning.
There are several prior works in the literature of domain adaptation that study the statistical landscape of the problem of learning under shifting distributions (see, e.g., \cite{ben2006analysis,blitzer2007learning,mansour2009domadapt,ben2010theory,david2010impossibility}). All of the previous generalization upper bounds on this problem involve some discrepancy term, which quantifies the amount of distribution shift, as well as some additional terms that are typically considered small for reasonable settings. For a concept class $\C:\X\to \cube{}$, considering that the error term $\optcommon$ (see Eq. \eqref{equation:definition:optcommon}) is small is a standard assumption in domain adaptation (see, e.g., \cite{ben2006analysis,blitzer2007learning}). Furthermore, one standard measure of discrepancy is defined as follows.

\begin{definition}[Discrepancy Distance, \cite{blitzer2007learning}]\label{definition:discrepancy}
    Let $\X\subset\R^d$ and let $\C$ be a concept class mapping $\X$ to $\cube{}$. For distributions $\Dgeneric,\Dgeneric'$ over $\X$, we define the discrepancy distance $\disc_{\C}(\Dgeneric,\Dgeneric')$ as follows.
    \[
        \disc_\C(\Dgeneric,\Dgeneric') = \sup_{\concept,\concept'\in \C}\Bigr| \pr_{\Dgeneric}[\concept(\x)\neq \concept'(\x)] - \pr_{\Dgeneric'}[\concept(\x)\neq \concept'(\x)] \Bigr|
    \]
\end{definition}
In particular, \cite{ben2006analysis,blitzer2007learning} observe that for any $\concept\in\C$ and distributions $\Dtrainjoint,\Dtestjoint$ over $\X\times\cube{}$ the following is true.
\begin{equation}
    \error(\concept;\Dtestjoint) \le \error(\concept;\Dtrainjoint) + \disc_\C(\Dtrainmarginal,\Dtestmarginal) + \optcommon(\C;\Dtrainjoint,\Dtestjoint) \label{equation:generalization-bound-population}
\end{equation}
The bound of Eq. \eqref{equation:generalization-bound-population} can be translated to a generalization bound for domain adaptation, through the use Rademacher complexity, whose definition is provided below.
\begin{definition}[Rademacher Complexity]\label{definition:rademacher-complexity}
    Let $\X\subseteq\R^d$, let $\Dgeneric$ be a distribution over $\X$ and let $\C$ be a concept class mapping $\X$ to $\cube{}$. For a set of $m$ samples $\Sunlabelled = (\x^{(1)},\x^{(2)},\dots,\x^{(m)})$ drawn independently from $\Dgeneric$, we define the empirical Rademacher complexity of $\C$ w.r.t. $\Sunlabelled$ as follows
    \[
        \rademacherempirical_\Sunlabelled(\C) = \frac{2}{m}\E \sup_{\concept\in \C} \sum_{j=1}^m \sigma_j \concept(\x^{(j)})\,, \text{ where the expectation is over }\sigma \sim\Unif(\cube{d})
    \]
    Moreover, we define the Rademacher complexity of $\C$ at $m$ w.r.t. $\Dgeneric$ as $\rademacher_m(\C;\Dgeneric) = \E[\rademacherempirical_\Sunlabelled(\C)]$, where the expectation is over ${\Sunlabelled\sim\Dgeneric^{\otimes m}}$.
\end{definition}

Corollaries 6, 7 in \cite{mansour2009domadapt}, demonstrate that the discrepancy between two distributions is upper bounded as follows.
\begin{proposition}[Bounding the Discrepancy, Corollary 7 in \cite{mansour2009domadapt}]\label{proposition:bounding-discrepancy}
    Consider $\X\subseteq\R^d$, a concept class $\C\subseteq\{\X\to\cube{d}\}$ and distributions $\Dgeneric,\Dgeneric'$ over $\X$. Then for any $\delta>0$, $m,m'\in\N$, if $\Sunlabelled,\Sunlabelled'$ are independent examples from $\Dgeneric,\Dgeneric'$, respectively, of sizes $m,m'$, the following is true.
    \[
        \disc_\C(\Dgeneric,\Dgeneric') \le \disc_\C(\Sunlabelled,\Sunlabelled') + 4\rademacherempirical_\Sunlabelled(\C)+4\rademacherempirical_{\Sunlabelled'}(\C)+{3}\,{(\log({4}/{\delta}))^{1/2}}\sqrt{\frac{1}{m}+\frac{1}{m'}}
    \]
\end{proposition}
Combining inequality \eqref{equation:generalization-bound-population} with \Cref{proposition:bounding-discrepancy} and standard generalization bounds for classification, yields a data-dependent generalization bound for domain adaptation whose only unknown parameter is $\optcommon$. In our setting this readily implies the following sample complexity upper bound in terms of the Rademacher complexity of the concept class $\C$.
\begin{corollary}[Sample Complexity upper bound for TDS learning]\label{corollary:sc-upper-bound}
    Let $\C\subseteq\{\X\to \cube{}\}$ be a hypothesis class and $\Dgeneric$ a distribution over $\X$ such that $\rademacher_m(\C;\Dgeneric) \le \eps/10$.
    The algorithm that runs the Empirical Risk Minimizer on training data and accepts only when both the empirical discrepancy distance between the training and test unlabelled examples, i.e. $\disc_\C(\Sunlabelled_\train,\Sunlabelled_\test)$, and the Rademacher complexity with respect to the test examples, i.e. $\rademacherempirical_{\Sunlabelled_\test}(\C)$, are $O(\eps)$, is an $(\eps,\delta)$-TDS learning algorithm for $\C$ up to error $2\optcommon+\eps$ with sample complexity $O(m+\frac{1}{\eps^2}\log(1/\delta))$. Moreover, if there is a concept in $\C$ with zero training error, the same is true up to error $\optcommon+ \eps$.
\end{corollary}

We emphasize that, while \Cref{corollary:sc-upper-bound} readily follows from prior results in the literature of domain adaptation, it highlights an important distinction between domain adaptation and TDS learning: A TDS learning algorithm, upon acceptance, achieves error that does not scale with the discrepancy between the training and test marginal distributions, but only a term that depends on the quantity $\optcommon$, which, as we show in \Cref{theorem:error-lower-bound-easy}, is unavoidable. 

\section{PQ Learning and Distribution-Free TDS Learning}
\label{sec: PQ implies TDS}

In recent years, there has been a vast amount of work on the problem of learning under shifting distributions. One of the most relevant models to TDS learning is PQ learning (see \cite{goldwasser2020beyond,kalai2021efficient}), which was defined by \cite{goldwasser2020beyond}. In this section, we establish a connection between PQ learning and TDS learning and, in particular, we show that TDS learning can be reduced to PQ learning, thereby inheriting all of the existing results in the latter framework. Unfortunately, to the best of our knowledge, most of the positive results on the PQ learning framework make strong assumptions regarding oracle access to solvers of learning primitives that are typically hard to solve. Nonetheless, PQ learning is an important theoretical framework for learning under arbitrary covariate shifts and it is an interesting open question whether our methods can be extended to provide positive results for the not-easier problem of PQ learning.

In the PQ learning framework, a learner outputs a pair $(h, \Xregion)$, where $h:\X\to\cube{}$ is a classifier and $\Xregion\subseteq\X$ is a subset of the feature space where one can be confident on the predictions of $h$. In particular, the PQ learning model is defined as follows.

\begin{definition}[PQ Learning, \cite{goldwasser2020beyond,kalai2021efficient}]\label{definition:pq-learning}
    Let $\X\subseteq\R^d$ be a set and $\C\subseteq{\X\to\cube{}}$ a concept class. For $\eps,\delta\in(0,1)$ we say that algorithm $\A$ PQ learns $\C$ up to error $\eps$ and probability of failure $\delta$ if for any distributions $\Dtrainjoint,\Dtestjoint$ over $\X\times\cube{}$ such that there is some $\coptcommon\in\C$ so that $y=\coptcommon(\x)$ for any $(\x,y)$ drawn from either $\Dtrainjoint$ or $\Dtestjoint$, algorithm $\A$, upon receiving a large enough number of labelled samples from $\Dtrainjoint$ and a large enough number of unlabelled samples from $\Dtestmarginal$, outputs a pair $(h,\Xregion)$ such that $h:\X\to\cube{}$, $\Xregion\subseteq\X$ and with probability at least $1-\delta$ the following is true.
    \[
        \pr_{\x\sim\Dtrainmarginal}[\x\not\in\Xregion] \le \eps \text{ and } \pr_{(\x,y)\sim\Dtestjoint}[h(\x)\neq y \text{ and }\x\in\Xregion] \le \eps
    \]
\end{definition}

We note that the above definition of PQ learning is distribution-free, i.e., the guarantees hold for any distribution and not with respect to a specific target distribution. In \Cref{definition:tds-learning} for TDS learning, the completeness criterion is stated with respect to a particular target distribution that is the same as the training distribution. However, in order to demonstrate a connection between PQ learning and TDS learning, we now define Distribution-Free TDS learning.

\begin{definition}[Distribution-free TDS Learning)]
\label{definition:dist-free-tds-learning}
    Let $\X\subseteq \R^d$ and consider a concept class $\C\subseteq\{\X\to \cube{}\}$. For $\eps,\delta\in(0,1)$, we say that an algorithm $\A$ testably learns $\C$ under distribution shifts up to error $\eps$ and probability of failure $\delta$ if the following is true. For any distributions $\Dtrainjoint, \Dtestjoint$ over $\X\times\cube{}$ such that there is some $\coptcommon\in\C$ such that $y=\coptcommon(\x)$ for any $(\x,y)$ drawn from either $\Dtrainjoint$ or $\Dtestjoint$, algorithm $\A$, upon receiving a large enough set of labelled samples $\Strain$ from the training distribution $\Dtrainjoint$ and a large enough set of unlabelled samples $\Stest$ from the test distribution $\Dtestmarginal$, either rejects $(\Strain,\Stest)$ or accepts and outputs a hypothesis $h:\X\to \cube{}$ with the following guarantees.
    \begin{enumerate}[label=\textnormal{(}\alph*\textnormal{)}]
    \item (Soundness.) With probability at least $1-\delta$ over the samples $\Strain,\Stest$ we have: 
    
    If $A$ accepts, then the output $h$ satisfies $\error(h;\Dtestjoint) \leq \eps$.
    \item (Completeness.) Whenever $\Dtestmarginal = \Dtrainmarginal$, $A$ accepts with probability at least $1-\delta$ over the samples $\Strain,\Stest$.
    \end{enumerate}
\end{definition}

We are now ready to prove that distribution-free TDS learning reduces to PQ learning.

\begin{proposition}[TDS learning via PQ learning]\label{theorem:tds-via-pq}
    \Cref{algorithm:tds-via-pq} reduces TDS to PQ learning. In particular, for $\eps,\delta\in(0,1)$, PQ learning algorithm $\A$ and a concept class $\C$, \Cref{algorithm:tds-via-pq}, upon receiving $m_P+\frac{C}{\eps^2}\log(1/\delta)$ labelled examples $\Strain$ from the training distribution and $m_Q+\frac{C}{\eps^2}\log(1/\delta)$ unlabelled examples $\Stest$ from the test distribution where $m_P,m_Q$ are such that $\A$ is an $(\eps/4,\delta)$-PQ learning algorithm for $\C$ given $m_P$ training and $m_Q$ test examples, $(\eps,\delta)$-TDS learns $\C$.
\end{proposition}
    
\begin{proof}
    Let $C>0$ be a sufficiently large universal constant. For \textbf{soundness}, we observe that upon acceptance, we have $\pr_{\x\sim\Sunlabelled_2}[\x\not\in\Xregion]$ and by a Hoeffding bound, since $m_2 \ge \frac{C}{\eps_2}\log(1/\delta)$, we have $\pr_{\x\sim\Dtestmarginal}[\x\not\in\Xregion] \le 2\eps/3$. By using the fact that $\error(h;\Dtestjoint) \le \pr_{\x\sim\Dtestmarginal}[\x\in\Xregion] + \pr_{\x\sim\Dtestmarginal}[\x\in\Xregion]$ and the guarantee of the PQ learner we obtain $\error(h;\Dtestjoint)\le \eps$, with probability at least $1-\delta$.
    For \textbf{completeness}, we use the definition of PQ learning and a Hoeffding bound to show that with probability at least $1-\delta$, \Cref{algorithm:tds-via-pq} accepts whenever $\Dtestmarginal=\Dtrainmarginal$.
\end{proof}

\begin{algorithm}
	\caption{TDS learning through PQ learning}\label{algorithm:tds-via-pq}
	\KwIn{Sets $\Strain$, $\Stest$, parameters $\eps,\delta\in(0,1)$, $(\eps'=\frac{\eps}{4},{\delta})$-PQ learner $\A$}
		Set $m_1=m_Q$, $m_2=\frac{C}{\eps^2}\log(1/\delta)$ and split $\Stest$ in $\Sunlabelled_1,\Sunlabelled_2$ with sizes $m_1,m_2$.\\
        Run algorithm $\A$ on $(\Strain,\Sunlabelled_1)$ and receive output $(h,\Xregion)$.\\
        \textbf{Reject} if $\pr_{\x\sim \Sunlabelled_2}[\x\not\in\Xregion]> \eps/3$.\\
        \textbf{Otherwise,} output $h$ and terminate.
\end{algorithm}

The simple reduction we provided in \Cref{theorem:tds-via-pq} implies that all of the positive results on PQ learning transfer to TDS learning. Moreover, note that the reduction does not alter the training and test distributions between the corresponding TDS and PQ algorithms and, therefore, would hold even in the distribution specific setting. This is not true, however, about the following corollary which is based on a reduction from PQ learning to reliable agnostic learning, which does not preserve the marginal distributions.

\begin{corollary}[Combination of Theorem 5 in \cite{kalai2021efficient} and \Cref{theorem:tds-via-pq}]
    If a concept class $\C$ is distribution -free reliably learnable, then it is TDS learnable in the distribution-free setting.
\end{corollary}

We remark that, in fact, (distribution-free) PQ learning is equivalent to (distribution-free) reliable learning (see Theorems 5, 6 in \cite{kalai2021efficient}). For a definition of reliable learning we refer the reader to \cite{kalai2012reliable}. It is known that reliable learning is no harder than agnostic learning and no easier than PAC learning.

\section{Amplifying success probability}\label{appendix:success-probability-boosting}
We will now demonstrate that it is possible to amplify the probability of success of a TDS learner through repetition. Note that this is not immediate for TDS learning as it is, for example, in agnostic learning, where one may repeat an agnostic learning algorithm and choose the hypothesis with the smallest error estimate among the outputs of the independent runs. The main obstacle is that test labels are not available. Nonetheless, we obtain the following theorem regarding amplifying the probability of success.

\begin{proposition}[Amplifying Success Probability]\label{proposition:boosting-success-probability}
    Let $\C$ be a hypothesis class, $\Dgeneric$ a distribution and suppose $\A$ is a TDS learner for $\C$ with respect to $\Dgeneric$ with error guarantee $\psi(\optcommon)+\eps$ and failure probability at most $0.1$. Then, there is a TDS learner $\A'$ for $\C$ with respect to $\Dgeneric$ with error guarantee $4\psi(\optcommon)+4\eps$ and failure probability at most $\delta$. In particular, $\A'$ repeats $\A$ for $T = O(\log(\frac{1}{\eps\delta}))$ times and rejects if most of the repetitions reject. If most repetitions accept, $\A'$ outputs the hypothesis $h = \majority(h_1,\dots,h_{T/2})$ ($h$ outputs the majority vote of $h_i$), where $h_1,\dots,h_{T/2}$ are the outputs of the first $T/2$ repetitions of $\A$ that accepted.
\end{proposition}

\begin{proof}
    We split the proof into two parts, one for soundness and one for completeness.
    
    \paragraph{Soundness.} For soundness, suppose that $\A'$ accepts. We denote with $\hat\pr$ (resp. $\hat\E$) the probabilities (resp. expectations) over the randomness of $h_1,\dots,h_{T/2}$ (which originates to the randomness of the samples given to $\A$) and with $\pr$ (resp. $\E$) the probabilities (resp. expectations) over the randomness of a pair $(\x,y)$ drawn from $\Dtestjoint$. In what follows, let $\eta = \psi(\optcommon)+\eps$. We have that for any $i=1,2,\dots,T/2$, $\hat\pr[\error(h_i,\Dtestjoint)\le \eta] \ge 0.9$, by the guarantees of $\A'$. We will show that $\hat\pr[\error(h,\Dtestjoint)\le 4\eta]\ge 1-\delta$ for a sufficiently large $T=O(\log(\frac{1}{\eps\delta}))$.

    We define $\goodevent_i$ to be the event (over the randomness of $h_i$) that $h_i$ is `good', i.e., that $\pr[h_i(\x)\neq y] \le \eta$. We define $\badregion$ to be the `bad' region of $(\x,y)$, i.e., $\badregion = \{(\x,y)\in\X\times\cube{}: \hat\pr[h_1(\x)\neq y|\goodevent_1] > 1/3\}$. Note that $\badregion$ would be the same even if we substituted $(h_1,\goodevent_1)$ above with an arbitrary $(h_i,\goodevent_i)$.

    First, we observe that $\pr[h(\x)\neq y] \le \pr[(\x,y)\in\badregion] + \pr[h(\x)\neq y|(\x,y)\not\in\badregion]$.

    We now observe that $\pr[(\x,y)\in \badregion] = \pr[\hat\pr[h_1\neq y|\goodevent_1]>1/3] \le 3 \E\hat\pr[h_1(\x)\neq y|\goodevent_1]$ by Markov's inequality. Now, we may swap the expectations to obtain $\pr[(\x,y)\in \badregion] \le 3\hat\E[\pr[h_1(\x)\neq y]|\goodevent_1] \le 3\eta$. 

    So far, we have shown $\pr[h(\x)\neq y] \le 3\eta + \pr[h(\x)\neq y|(\x,y)\not\in\badregion]$. We will bound the probability over $h_1,\dots,h_{T/2}$ that $\pr[h(\x)\neq y|(\x,y)\not\in\badregion]>\eta$. In particular, we have the following due to Markov's inequality $\hat\pr[\pr[h(\x)\neq y|(\x,y)\not\in\badregion]>\eta] \le \frac{1}{\eta}\hat\E[\pr[h(\x)\neq y|(\x,y)\not\in\badregion]]$. Once more, we may swap the expectations to obtain $\hat\pr[\pr[h(\x)\neq y|(\x,y)\not\in\badregion]>\eta] \le \frac{1}{\eta}\E[\hat\pr[h(\x)\neq y]|(\x,y)\not\in\badregion]$.

    Moreover, if we fix $(\x,y)\not\in\badregion $, then $\hat\pr[h_i(\x)=y] \ge \hat\pr[h_i(\x)=y \text{ and } \goodevent_i] \ge \frac{2}{3}\cdot \frac{9}{10} \ge 3/5$. Because $\hat\pr[\goodevent_i]\ge 0.9$ and $\hat\pr[h_i(\x)=y | \goodevent_i] \ge 2/3$ whenever $(\x,y)\not\in\badregion$, by the definition of $\badregion$. Therefore, since $h_1,\dots,h_{T/2}$ are independent, we have that $\hat\pr[h(\x)\neq y] \le \exp(-T/C)$ for some sufficiently large universal constant $C>0$, for any $(\x,y)\not\in\badregion$.

    Therefore, in total, $\hat\pr[\pr[h(\x)\neq y|(\x,y)\not\in\badregion]>\eta] \le \frac{1}{\eta}\exp(-T/C)$. We set $T = C\ln(\frac{1}{\eps\delta}) \ge C\ln(\frac{1}{\eta\delta})$ to obtain $\hat\pr[\pr[h(\x)\neq y|(\x,y)\not\in\badregion]>\eta] \le \delta$ and, hence, with probability at least $1-\delta$ over the randomness of $h$ we overall have $\pr[h(\x)\neq y]\le 4\eta$.
    
    \paragraph{Completeness.} Completeness follows by a standard Hoeffding bound.
\end{proof}

\section{Auxiliary Propositions}\label{appendix:auxiliary-propositions}

Let $\Gauss(0,I_d)$ denote the standard multivariate Gaussian distribution over $\R^d$ and $\Unif(\cube{d})$ denote the uniform distribution over the hypercube $\cube{d}$. For each of these distributions, we show that the sandwiching polynomials of any binary concept have coefficients that are absolutely bounded, that the empirical moments concentrate around the true ones and that the empirical squared error of polynomials with bounded degree and coefficients uniformly converges to the true squared error. These properties are used in order to apply \Cref{theorem:l2-sandwiching-implies-tds} to obtain TDS learning algorithms for a number of classes under the Gaussian and Uniform distributions.

\subsection{Properties of Gaussian Distribution}

We prove the following fact about the Gaussian distribution.

\begin{lemma}[Properties of the Gaussian Distribution]\label{lemma:gaussian-properties}
    Let $\Dgeneric$ be the standard Gaussian $\Gauss(0,I_d)$ over $\R^d$. Then the following are true.
    \begin{enumerate}[label=\textnormal{(}\roman*\textnormal{)}]
        \item\label{item:gaussian-coefficient-bound} (Coefficient Bound) Suppose that for some $\eps\in (0,1]$, $\degbound>0$ and some concept class $\C\subseteq{\R^d\to \cube{}}$, the $\eps$-approximate $L_2$ sandwiching degree of $\C$ w.r.t. $\Gauss(0,I_d)$ is at most $\degbound$. Then, the coefficients of the sandwiching polynomials for $\C$ are absolutely bounded by $\pbound=O(d)^\degbound$.
        \item\label{item:gaussian-concentration} (Concentration) For any $\delta\in (0,1), \mslack>0$ and $\degbound>0$, if $\Sunlabelled$ is a set of independent samples from $\Dgeneric$ with size at least $\mconc = \frac{O(d\degbound)^\degbound}{\Delta^2 \cdot \delta}$ then, with probability at least $1-\delta$ over the randomness of $\Sunlabelled$, we have that for any $\mindex\in \N^d$ with $\|\mindex\|_1\le k$ it holds $|\E_\Sunlabelled[\x^\mindex] - \E_\Dgeneric[\x^\mindex]| \le \Delta$.
        \item\label{item:gaussian-generalization} (Generalization) For any $\eps>0$, $\delta\in (0,1), \pbound>0$, $\degbound>0$, and any distribution $\Dgenericjoint$ over $\R^d\times\cube{}$ whose marginal on $\R^d$ is $\Dgeneric$, if $\Slabelled$ is a set of independent samples from $\Dgenericjoint$ with size at least $\mgen = \tilde{O}(\frac{\pbound^8}{\eps^4\delta})\cdot d^{O(\degbound)}$ then, with probability at least $1-\delta$ over the randomness of $\Slabelled$, we have that for any polynomial $p$ of degree at most $\degbound$ and coefficients that are absolutely bounded by $\pbound$ it holds $|\E_\Slabelled[(y-p(\x))^2]-\E_{\Dgenericjoint}[(y-p(\x))^2]|\le {\eps}$.
    \end{enumerate}
\end{lemma}

\begin{proof}
    We will prove each part of the Lemma separately.

    \noindent\textit{Part \ref{item:gaussian-coefficient-bound}.} Suppose that $\pup,\pdown$ are $1$-sandwiching polynomials for some concept $\concept\in\C$ with degree at most $\degbound$. Then, we have the following.
    \begin{align*}
        \|\pdown\|_{L_2(\Dgeneric)} &\le \|\pup-\concept\|_{L_2(\Dgeneric)} + \|\concept\|_{L_2{\Dgeneric}} \\
        &\le \|\pup-\pdown\|_2 + 1 \le 2
    \end{align*}
    Since $\Dgeneric$ is the standard Gaussian distribution, the quantity $\|\pdown\|_{L_2(\Dgeneric)}^2$ equals to the sum of the squares of the coefficients of the Hermite expansion of $\pdown$ (see e.g. \cite{o2014analysis}). Therefore, each Hermite coefficient of $\pdown$ is absolutely bounded by $2$. Each Hermite polynomial of degree at most $\degbound$ has coefficients that are absolutely bounded by $2^{\degbound}$. Since $\pdown$ has degree at most $\degbound$, each coefficient of $\pdown$ is absolutely bounded by $d^{O(\degbound)}$.

    \noindent\textit{Part \ref{item:gaussian-concentration}.} Suppose that $\mindex\in\N^d$ with $\|\mindex\|_1 \le k$. Then, the worst case regarding moment concentration is $\mindex_1 = \degbound$. For a sample $\Sunlabelled$ from $\Dgeneric$, we apply Chebyshev's inequality on the random variable $z = |\E_\Sunlabelled[\x_1^\degbound] - \E_{\Dgeneric}[\x_1^\degbound]|$ and by bounding $\E[z^2]$ by $\E_\Dgeneric[\x_1^{2\degbound}]$ we have that for any $\Delta>0$, $z\le \Delta$ with probability at least $1-\frac{(C\degbound)^\degbound}{|\Sunlabelled|\Delta^2}$, where the randomness is over the random choice of $\Sunlabelled$ and $C>0$ is a sufficiently large universal constant (for bounds on the Gaussian moments, see, e.g., Proposition 2.5.2 in \cite{vershynin2018high}). Since we need the result to hold for all $\mindex$ simultaneously, the result follows by a union bound.

    \noindent\textit{Part \ref{item:gaussian-generalization}.} We define $\Polys$ to be the class of polynomials over $\R^d$ with degree at most $\degbound$ and coefficients that are absolutely bounded by $\pbound$. Let $T>0$ to be disclosed later and $m = |\Slabelled|$. We will first show that with probability at least $1-\delta/2$ over the choice of $\Slabelled$, we have 
    \[
        \E_{\Dgenericjoint}[(y-p(\x))^2] \le \E_{\Slabelled}[(y-p(\x))^2] + \eps \text{ for all }p\in \Polys
    \]
    We aim to apply some standard uniform convergence argument, but in order to do so we first need to ensure certain boundedness conditions as follows.
    \[
        \E_{\Dgenericjoint}[(y-p(\x))^2] = \E_{\Dgenericjoint}[(y-p(\x))^2 \cdot \ind\{\forall q\in\Polys: |q(\x)|\le T\}] + \E_{\Dgenericjoint}[(y-p(\x))^2 \cdot \ind\{\exists q\in \Polys: |q(\x)|> T\}]
    \]
    where we have $\E_{\Dgenericjoint}[(y-p(\x))^2 \cdot \ind\{\forall q\in\Polys: |q(\x)|\le T\}]\le \E_{\Dgenericjoint}[(y-p(\x))^2 \;|\;\forall q\in\Polys: |q(\x)|\le T]$. Let $\Dgenericjoint'$ be the distribution that corresponds to $\Dgenericjoint$ conditioned on the event $\{\forall q\in\Polys: |q(\x)|\le T\}$ and let $\Slabelled' = \{(\x,y)\in \Slabelled: |q(\x)|\le T, \forall q\in \Polys\}$. By standard arguments using Rademacher complexity bounds for bounded losses (see, e.g., Theorems 5.5 and 10.3 in \cite{mohri2018foundations}) we have that for some sufficiently large universal constant $C>0$, with probability at least $1-\delta/4$, we have for any $p\in \Polys$
    \begin{equation}\label{equation:rademacher-bound}
        \E_{\Dgenericjoint'}[(y-p(\x))^2] \le \E_{\Slabelled'}[(y-p(\x))^2] + T^4\cdot \frac{\pbound + \sqrt{\log(1/\delta)}}{\sqrt{m/C}}
    \end{equation}
    We now need to link $\E_{\Slabelled'}[(y-p(\x))^2]$ to $\E_{\Slabelled}[(y-p(\x))^2]$. We have the following.
    \begin{align*}
        \E_{\Slabelled}[(y-p(\x))^2] &\ge (1-\pr_\Slabelled[\exists q\in \Polys: |q(\x)|> T])\E_{\Slabelled'}[(y-p(\x))^2] \\
        &\ge \E_{\Slabelled'}[(y-p(\x))^2] - \pr_\Slabelled[\exists q\in \Polys: |q(\x)|> T] \cdot 2T^2 \tag{since $y\in\cube{}$ and $p\in \Polys$}
    \end{align*}
    We will upper bound the quantity $\pr_\Slabelled[\exists q\in \Polys: |q(\x)|> T]$. We have
    \begin{align}
        \pr_\Slabelled[\exists q\in \Polys: |q(\x)|> T] &= \pr_\Slabelled{{}}\Bigr[\exists (q_\mindex)_{\|\mindex\|_1\le \degbound}\in [-\pbound,\pbound]^{d^\degbound}:  {{}}\Bigr|\sum_{\mindex}q_\mindex \x^\mindex{{}}\Bigr|> T {{}}\Bigr] \nonumber \\
        &\le \sum_{\mindex: \|\mindex\|_1\le \degbound} \pr_{\Slabelled}{{}}\Bigr[ |\x^\mindex| \ge \frac{T}{\pbound d^{\degbound}} {{}}\Bigr] \nonumber\\
        &\le \sum_{\mindex: \|\mindex\|_1\le \degbound} \pr_{\Dgeneric}{{}}\Bigr[ |\x^\mindex| \ge \frac{T}{\pbound d^{\degbound}} {{}}\Bigr] + \frac{d^{\degbound}}{\sqrt{2m}} \log\Bigr(\frac{8}{\delta}\Bigr), \text{ w.p. at least } 1-\delta/4\label{equation:bound-empirical-concentration-of-polynomials}
    \end{align}
    In the last step, we used a standard Chernoff-Hoeffding bound. We now bound $\sum_{\mindex: \|\mindex\|_1\le \degbound} \pr_{\Dgeneric}[ |\x^\mindex| \ge \frac{T}{\pbound d^{\degbound}} ]$. Recall that $\Dgeneric = \Gauss(0,I_d)$ and therefore the worst case for $\mindex$ regarding concentration is the case $\mindex_1 = \degbound$. We therefore obtain the following via Gaussian concentration.
    \begin{align}
        \sum_{\mindex: \|\mindex\|_1\le \degbound} \pr_{\Dgeneric}{{}}\Bigr[ |\x^\mindex| \ge \frac{T}{\pbound d^{\degbound}} {{}}\Bigr] &\le d^{\degbound} \pr_{\Dgeneric}{{}}\Bigr[ |\x_1^\degbound| \ge \frac{T}{\pbound d^{\degbound}} {{}}\Bigr] \nonumber \\
        &\le d^\degbound \exp\Bigr( -\frac{1}{2}\cdot \frac{T^{1/\degbound}}{\pbound^{1/\degbound} d} \Bigr) \label{equation:bound-monomial-tail-probability}
    \end{align}
    
    It remains to bound the term $\E_{\Dgenericjoint}[(y-p(\x))^2 \cdot \ind\{\exists q\in \Polys: |q(\x)|> T\}]$. By applying the Cauchy-Schwarz inequality, it is sufficient to bound $\sqrt{\E_{\Dgenericjoint}[(y-p(\x))^4]}\cdot\sqrt{\pr_\Dgeneric[\exists q\in \Polys: |q(\x)|> T]}$. For the second term, we use Equation \eqref{equation:bound-monomial-tail-probability}. For the first term, we have the following for some sufficiently large constant $C>0$.

    \begin{align}
        \E_{\Dgeneric}[(y-p(\x))^4] &\le 8+8\E_{\Dgeneric}[p^4(\x)] \nonumber\\
        &\le 8 + \pbound^4d^{4\degbound}\sum_{\|\mindex\|_1\le 4\degbound} \prod_{i:\mindex_i>0}\E_{\Dgeneric}[\x_i^{\mindex_i}] \tag{since $\deg(p^4)\le 4\degbound$ and $|(p^4)_{\mindex}| \le \pbound^4d^{4\degbound}$} \nonumber \\
        &\le \pbound^4d^{8\degbound}(C\degbound)^{2\degbound} \tag{since $\Dgeneric = \Gauss(0,I_d)$, see Proposition 2.5.2 in \cite{vershynin2018high}}
    \end{align}

    Using the above inequality along with \eqref{equation:rademacher-bound}, \eqref{equation:bound-empirical-concentration-of-polynomials} and \eqref{equation:bound-monomial-tail-probability} we obtain that $\E_{\Dgenericjoint}[(y-p(\x))^2] - \E_{\Slabelled}[(y-p(\x))^2]$ is upper bounded by the following quantity for some sufficiently large universal constant $C>0$
    \begin{align*}
        &T^4\cdot \frac{\pbound + \sqrt{\log(1/\delta)}}{\sqrt{m/C}} + 2T^2d^\degbound \exp\left( -\frac{1}{2}\cdot\left( \frac{T}{\pbound d^{\degbound}} \right)^{1/\degbound} \right)  + \\
        &+ 2T^2\frac{d^{\degbound}}{\sqrt{2m}} \log\Bigr(\frac{10}{\delta}\Bigr)+\pbound^2d^{4\degbound}(C\degbound)^{\degbound}d^{\degbound/2} \exp\left( -\frac{1}{4}\cdot\left( \frac{T}{\pbound d^{\degbound}} \right)^{1/\degbound} \right)\,,
    \end{align*}
    which is at most $\epsilon$ when we choose $m,T$ as follows for some universal constant $C>0$ (possibly larger than the previously defined constants for which we used the same letter) for the choice $T= C\pbound(4d)^{\degbound}\degbound \log(\frac{\pbound d \degbound}{\eps})$ and $m=\frac{C}{\eps^2}(\pbound^2+\log(\frac{1}{\delta}))\pbound^8(4d)^{8\degbound}\degbound^8 \log(\frac{\pbound d \degbound}{\eps}) = \tilde{O}(\frac{\pbound}{\eps^2}) \cdot O(d)^{8\degbound}\cdot \log(1/\delta)$. 

    In order to bound the symmetric difference, we also need to bound the quantity $\E_{\Slabelled}[(y-p(\x))^2] - \E_{\Dgenericjoint}[(y-p(\x))^2]$, which we may do following a similar reasoning, but requiring, at times, bounds on quantities that correspond to empirical expectations (instead of expectations over the population distribution). In particular, we will require a bound on $\E_\Slabelled[(y-p(\x))^4]$, which can be reduced to bounding $\E_\Slabelled[p^4(\x)]$, for which we may use part \ref{item:gaussian-concentration}, demanding $m\ge d^{O(\degbound)}/\delta$ to obtain
    \[
        \E_\Slabelled[p^4(\x)] \le 2\pbound^4d^{4\degbound}(C\degbound)^{\degbound}
    \]
    Overall, this step will introduce the additional requirement that $m\ge \frac{\pbound^8}{\eps^4\delta}d^{16\degbound}(C\degbound)^{4\degbound}\log^2(\frac{1}{\delta})$. Therefore, overall, for $m\ge \mgen = \tilde{O}(\frac{\pbound^8}{\eps^2\delta})\cdot d^{O(\degbound)} \cdot \log^2(\frac{1}{\delta})$, we have the desired result. 
\end{proof}

\subsection{Properties of Uniform Distribution}

We prove the following fact about the uniform distribution.

\begin{lemma}[Properties of the Uniform Distribution]\label{lemma:uniform-properties}
    Let $\Dgeneric$ be the uniform distribution over the hypercube $\Unif(\cube{d})$ and $C>0$ some sufficiently large constant. Then the following are true.
    \begin{enumerate}[label=\textnormal{(}\roman*\textnormal{)}]
        \item\label{item:uniform-coefficient-bound} (Coefficient Bound) Suppose that for some $\eps\in(0,1]$, $\degbound>0$ and some concept class $\C\subseteq{\R^d\to \cube{}}$, the $\eps$-approximate $L_2$ sandwiching degree of $\C$ w.r.t. $\Dgeneric$ is at most $\degbound$. Then, the coefficients of the sandwiching polynomials for $\C$ are absolutely bounded by $\pbound=2$.
        \item\label{item:uniform-concentration} (Concentration) For any $\delta\in (0,1), \mslack>0$ and $\degbound>0$, if $\Sunlabelled$ is a set of independent samples from $\Dgeneric$ with size at least $\mconc = \frac{C\degbound}{\Delta^2}\log(\frac{d}{\delta})$ then, with probability at least $1-\delta$ over the randomness of $\Sunlabelled$, we have that for any $\mindex\in \N^d$ with $\|\mindex\|_1\le k$ it holds $|\E_\Sunlabelled[\x^\mindex] - \E_\Dgeneric[\x^\mindex]| \le \Delta$.
        \item\label{item:uniform-generalization} (Generalization) For any $\eps>0$, $\delta\in (0,1), \pbound>0$, $\degbound>0$, and any distribution $\Dgenericjoint$ over $\R^d\times\cube{}$ whose marginal on $\R^d$ is $\Dgeneric$, if $\Slabelled$ is a set of independent samples from $\Dgenericjoint$ with size at least $\mgen =  \tilde{O}(\frac{1}{\eps^2})\cdot \pbound^{O(1)}\cdot d^{O(\degbound)} \cdot \log(\frac{1}{\delta})$ then, with probability at least $1-\delta$ over the randomness of $\Slabelled$, we have that for any polynomial $p$ of degree at most $\degbound$ and coefficients that are absolutely bounded by $\pbound$ it holds $|\E_\Slabelled[(y-p(\x))^2]-\E_{\Dgenericjoint}[(y-p(\x))^2]|\le {\eps}$.
    \end{enumerate}
\end{lemma}

\begin{proof}
    We will prove each part of the Lemma separately.

    \noindent\textit{Part \ref{item:uniform-coefficient-bound}.} Suppose that $\pup,\pdown$ are $1$-sandwiching polynomials for some concept $\concept\in\C$ with degree at most $\degbound$. Then, we have the following.
    \begin{align*}
        \|\pdown\|_{L_2(\Dgeneric)} &\le \|\pup-\concept\|_{L_2(\Dgeneric)} + \|\concept\|_{L_2{\Dgeneric}} \\
        &\le \|\pup-\pdown\|_2 + 1 \le 2
    \end{align*}
    Since $\Dgeneric$ is the uniform distribution, the quantity $\|\pdown\|_{L_2(\Dgeneric)}^2$ equals to the sum of the squares of the coefficients of $\pdown$ (see e.g. \cite{o2014analysis}). Therefore, each coefficient of $\pdown$ is absolutely bounded by $2$.

    \noindent\textit{Part \ref{item:uniform-concentration}.} Suppose that $\mindex\in\{0,1\}^d$ with $\|\mindex\|_1 \le k$. For a sample $\Sunlabelled$ from $\Dgeneric$, we apply Hoeffding's inequality on the random variable $z = |\E_\Sunlabelled[\x^{\mindex}] - \E_{\Dgeneric}[\x^\mindex]|$ and by observing that $\x^\mindex\in\{\pm 1\}$ we have that the probability that $z>\Delta$ is at most $2\exp(-|\Sunlabelled|\Delta^2/10)$. We obtain the desired result by a union bound.

    \noindent\textit{Part \ref{item:uniform-generalization}.} We define $\Polys$ to be the class of polynomials over $\cube{d}$ with degree at most $\degbound$ and coefficients that are absolutely bounded by $\pbound$. Let $T>0$ to be disclosed later and $m = |\Slabelled|$. We will show that with probability at least $1-\delta$ over the choice of $\Slabelled$, we have 
    \[
        |\E_{\Dgenericjoint}[(y-p(\x))^2] - \E_{\Slabelled}[(y-p(\x))^2]| \le \eps \text{ for all }p\in \Polys
    \]
    We apply some standard uniform convergence argument, by observing that $(y-p(\x))^2 \le 2+2\pbound^2d^k$. In particular by standard arguments using Rademacher complexity bounds for bounded losses (see, e.g., Theorems 5.5 and 10.3 in \cite{mohri2018foundations}) we obtain the desired result.
\end{proof}



\end{document}